\documentclass{article}

\usepackage[utf8]{inputenc}
\usepackage[T1]{fontenc}
\usepackage{epigraph}
\usepackage{lettrine}
\usepackage{type1ec}
\usepackage{latexsym}
\usepackage[makeroom]{cancel}

\usepackage[inline]{enumitem}

\usepackage{setspace}
\usepackage{afterpage}
\usepackage{newlfont}
\usepackage{array}
\usepackage{tabularx}
\usepackage{dcolumn}
\usepackage{etoolbox}

\usepackage[bf]{caption}
\usepackage{tikz}
\usepackage{xcolor}
\usetikzlibrary{arrows}
\usepackage{bussproofs}
\usepackage{proof}
\usepackage{stmaryrd}
\usepackage{multirow}
\usepackage{multicol}

\usepackage{lmodern}

\usepackage[hide]{ed}
\usepackage{mathrsfs}
\usepackage{graphicx}
\usepackage{amsmath}
\usepackage{amssymb}
\usepackage{amsthm}

\usepackage{comment}
\usepackage{authblk}

\usepackage[framemethod=tikz]{mdframed}
\newmdenv[innerlinewidth=0.2pt, roundcorner=4pt,innerleftmargin=6pt,
innerrightmargin=6pt,innertopmargin=6pt,innerbottommargin=6pt]{mybox}

\usepackage{tikz}
\usepackage{bussproofs}
\usepackage{multirow}
\usepackage{amssymb}

\usepackage{thm-restate}

\newcommand{\tiz}{\textcolor{red}}

\newcommand{\AND}{\bigand}

\newcommand{\tr}{t}
\newcommand{\ttr}{t'}

\newcommand{\varseq}{\mathit S}

\newcommand{\mf}{\mathsf}

\newcommand{\ufor}{\Vd^{\forall}}
\newcommand{\efor}{\Vd^{\exists}}

\newcommand{\leftrule}[1]{$#1_{\mathsf{L}}$}
\newcommand{\rightrule}[1]{$#1_{\mathsf{R}}$}

\newcommand{\init}{$\mf{init}$}
\newcommand{\lbot}{\leftrule{\bot}}

\newcommand{\initcl}{$\mf{init}^{c\ell}$}
\newcommand{\lbotcl}{$\bot_{\mf L}^{c\ell}$}
\newcommand{\limpcl}{$\imp_{\mf L}^{c\ell}$}
\newcommand{\rimpcl}{$\imp_{\mf R}^{c\ell}$}
\newcommand{\llandcl}{$\land_{\mf L}^{c\ell}$}
\newcommand{\rlandcl}{$\land_{\mf R}^{c\ell}$}
\newcommand{\llorcl}{$\lor_{\mf L}^{c\ell}$}
\newcommand{\rlorcl}{$\lor_{\mf R}^{c\ell}$}

\newcommand{\diam}{\Diamond}
\newcommand{\boxm}[1]{\Box^-#1}
\newcommand{\fp}{f}
\newcommand{\BoxI}{\Box_1}
\newcommand{\BoxM}{\Box_2}
\newcommand{\diamM}{\diam_2}

\newcommand{\Boxtwo}{\Box_2}
\newcommand{\Boxthree}{\Box_3}

\newcommand{\diamtwo}{\diam_2}
\newcommand{\diamthree}{\diam_3}

\newcommand{\G}{\Gamma}
\newcommand{\D}{\Delta}

\newcommand{\Der}{\triangledown}

\newcommand{\N}{\mathcal N}
\newcommand{\R}{\mathcal R}

\newcommand{\W}{\mathcal W}

\newcommand{\V}{\mathcal V}
\newcommand{\FW}{\mathbb{F}}

\newcommand{\MNM}{MNM}

\newcommand{\M}{\mathcal M}
\newcommand{\lan}{\mathcal L}
\newcommand{\lanone}{\mathcal L_1}

\newcommand{\atm}{\mathit{Atm}}
\newcommand{\wff}{\mathit{Wff}}
\newcommand{\CC}{\mathscr{C}}
\newcommand{\DD}{\mathscr{D}}
\newcommand{\EE}{\mathscr{E}}
\newcommand{\U}{\mathscr{U}}
\newcommand{\VV}{\mathscr{V}}
\newcommand{\ZZ}{\mathscr{Z}}
\newcommand{\aU}{\alpha_{\mathscr{U}}}
\newcommand{\aUp}{\alpha_{\mathscr{U'}}}
\newcommand{\canonic}[1]{#1}
\newcommand{\Mc}{\canonic{\M}}
\newcommand{\Vc}{\canonic{\V}}
\newcommand{\Nc}{\canonic{\N}}
\newcommand{\Wc}{\canonic{\W}}
\newcommand{\Rc}{\canonic{\R}}

\newcommand{\FWc}{\canonic{\FW}}
\newcommand{\lessc}{\canonic{\less}}
\newcommand{\morec}{\canonic{\more}}

\newcommand{\Uc}{\U_{C}}

\newcommand{\Ua}{\U_{A}}
\newcommand{\Ub}{\U_{B}}

\newcommand{\Psicd}{\Psi_{CD}}

\newcommand{\Psic}{\Psi_{C}}
\newcommand{\Psid}{\Psi_{D}}

\newcommand{\pow}{\mathcal P}

\newcommand{\EM}{\logicnamestyle{M}}

\newcommand{\EMN}{\logicnamestyle{MN}}
\newcommand{\EMC}{\logicnamestyle{MC}}
\newcommand{\EMCN}{\logicnamestyle{MCN}}

\newcommand{\EMP}{\logicnamestyle{MP}}
\newcommand{\EMNP}{\logicnamestyle{MNP}}
\newcommand{\EMCP}{\logicnamestyle{MCP}}

\newcommand{\EMD}{\logicnamestyle{MD}}
\newcommand{\EMND}{\logicnamestyle{MND}}
\newcommand{\EMCD}{\logicnamestyle{MCD}}

\newcommand{\EMT}{\logicnamestyle{MT}}
\newcommand{\EMNT}{\logicnamestyle{MNT}}
\newcommand{\EMCT}{\logicnamestyle{MCT}}

\newcommand{\K}{\logicnamestyle{K}}
\newcommand{\KT}{\logicnamestyle{KT}}
\newcommand{\KD}{\logicnamestyle{KD}}
\newcommand{\KP}{\logicnamestyle{KP}}

\newcommand{\IK}{\logicnamestyle{I.K}}

\newcommand{\WK}{\logicnamestyle{W.K}}
\newcommand{\WKD}{\logicnamestyle{W.KD}}
\newcommand{\WKT}{\logicnamestyle{W.KT}}
\newcommand{\WM}{\logicnamestyle{W.M}}
\newcommand{\WMN}{\logicnamestyle{W.MN}}
\newcommand{\WMC}{\logicnamestyle{W.MC}}
\newcommand{\WMD}{\logicnamestyle{W.MD}}
\newcommand{\WMND}{\logicnamestyle{W.MND}}
\newcommand{\WMCD}{\logicnamestyle{W.MCD}}
\newcommand{\WMT}{\logicnamestyle{W.MT}}
\newcommand{\WMNT}{\logicnamestyle{W.MNT}}
\newcommand{\WMCT}{\logicnamestyle{W.MCT}}
\newcommand{\WMP}{\logicnamestyle{W.MP}}
\newcommand{\WMNP}{\logicnamestyle{W.MNP}}

\newcommand{\varK}{\logicnamestyle{{}^*K^*}}
\newcommand{\varMC}{\logicnamestyle{{}^*MC^*}}
\newcommand{\seqvarK}{\Gone{{}^*K^*}}
\newcommand{\seqvarMC}{\Gone{{}^*MC^*}}

\newcommand{\CM}{\logicnamestyle{C.M}}
\newcommand{\CMC}{\logicnamestyle{C.MC}}
\newcommand{\CMN}{\logicnamestyle{C.MN}}
\newcommand{\CMP}{\logicnamestyle{C.MP}}
\newcommand{\CMD}{\logicnamestyle{C.MD}}
\newcommand{\CMT}{\logicnamestyle{C.MT}}
\newcommand{\CMNP}{\logicnamestyle{C.MNP}}
\newcommand{\CMND}{\logicnamestyle{C.MND}}
\newcommand{\CMNT}{\logicnamestyle{C.MNT}}

\newcommand{\CMCD}{\logicnamestyle{C.MCD}}
\newcommand{\CMCT}{\logicnamestyle{C.MCT}}
\newcommand{\CK}{\logicnamestyle{C.K}}
\newcommand{\CKD}{\logicnamestyle{C.KD}}
\newcommand{\CKT}{\logicnamestyle{C.KT}}

\newcommand{\MM}{\logicnamestyle{M.M}}
\newcommand{\MMC}{\logicnamestyle{M.MC}}
\newcommand{\MMN}{\logicnamestyle{M.MN}}
\newcommand{\MMP}{\logicnamestyle{M.MP}}
\newcommand{\MMD}{\logicnamestyle{M.MD}}
\newcommand{\MMT}{\logicnamestyle{M.MT}}
\newcommand{\MMNP}{\logicnamestyle{M.MNP}}
\newcommand{\MMND}{\logicnamestyle{M.MND}}
\newcommand{\MMNT}{\logicnamestyle{M.MNT}}

\newcommand{\MMCD}{\logicnamestyle{M.MCD}}
\newcommand{\MMCT}{\logicnamestyle{M.MCT}}
\newcommand{\MMCN}{\logicnamestyle{M.MCN}}
\newcommand{\MK}{\logicnamestyle{M.K}}
\newcommand{\MKD}{\logicnamestyle{M.KD}}
\newcommand{\MKT}{\logicnamestyle{M.KT}}

\newcommand{\seqCLvar}{\Gone{C.L}}
\newcommand{\seqMLvar}{\Gone{M.L}}
\newcommand{\seqWLvar}{\Gone{W.L}}
\newcommand{\seqLvar}{\Gone{L}}

\newcommand{\seqWK}{\Gone{W.K}}
\newcommand{\seqWKD}{\Gone{W.KD}}
\newcommand{\seqWKT}{\Gone{W.KT}}
\newcommand{\seqWM}{\Gone{W.M}}
\newcommand{\seqWMN}{\Gone{W.MN}}
\newcommand{\seqWMC}{\Gone{W.MC}}
\newcommand{\seqWMD}{\Gone{W.MD}}
\newcommand{\seqWMND}{\Gone{W.MND}}
\newcommand{\seqWMCD}{\Gone{W.MCD}}
\newcommand{\seqWMT}{\Gone{W.MT}}
\newcommand{\seqWMNT}{\Gone{W.MNT}}
\newcommand{\seqWMCT}{\Gone{W.MCT}}
\newcommand{\seqWMP}{\Gone{W.MP}}
\newcommand{\seqWMNP}{\Gone{W.MNP}}

\newcommand{\seqCK}{\Gone{C.K}}
\newcommand{\seqCKD}{\Gone{C.KD}}
\newcommand{\seqCKT}{\Gone{C.KT}}
\newcommand{\seqCM}{\Gone{C.M}}
\newcommand{\seqCMN}{\Gone{C.MN}}
\newcommand{\seqCMC}{\Gone{C.MC}}
\newcommand{\seqCMD}{\Gone{C.MD}}
\newcommand{\seqCMND}{\Gone{C.MND}}
\newcommand{\seqCMCD}{\Gone{C.MCD}}
\newcommand{\seqCMT}{\Gone{C.MT}}
\newcommand{\seqCMNT}{\Gone{C.MNT}}
\newcommand{\seqCMCT}{\Gone{C.MCT}}
\newcommand{\seqCMP}{\Gone{C.MP}}
\newcommand{\seqCMNP}{\Gone{C.MNP}}

\newcommand{\seqL}{\Gone{L}}
\newcommand{\SC}{\mathsf{SC}}

\newcommand{\seqK}{\Gone{K}}
\newcommand{\seqKD}{\Gone{KD}}
\newcommand{\seqKT}{\Gone{KT}}
\newcommand{\seqEM}{\Gone{M}}
\newcommand{\seqEMN}{\Gone{MN}}
\newcommand{\seqEMC}{\Gone{MC}}
\newcommand{\seqEMD}{\Gone{MD}}
\newcommand{\seqEMND}{\Gone{MND}}
\newcommand{\seqEMCD}{\Gone{MCD}}
\newcommand{\seqEMT}{\Gone{MT}}
\newcommand{\seqEMNT}{\Gone{MNT}}
\newcommand{\seqEMCT}{\Gone{MCT}}
\newcommand{\seqEMP}{\Gone{MP}}
\newcommand{\seqEMNP}{\Gone{MNP}}

\newcommand{\Sfour}{\logicnamestyle{S4}}
\newcommand{\Sfive}{\logicnamestyle{S5}}

\newcommand{\CPL}{\CL}
\newcommand{\IPL}{\IL}

\newcommand{\Mvar}{\logicnamestyle{M}}
\newcommand{\MLvar}{\logicnamestyle{M.L}}
\newcommand{\CLvar}{\logicnamestyle{C.L}}
\newcommand{\WLvar}{\logicnamestyle{W.L}}
\newcommand{\ILvar}{\logicnamestyle{I.L}}

\newcommand{\MSigma}{\logicnamestyle{M\Sigma}}
\newcommand{\MMSigma}{\logicnamestyle{M.M\Sigma}}
\newcommand{\xLvar}{\xstar}

\newcommand{\xstar}{\logicnamestyle{x}^*}

\newcommand{\xCstar}{\logicnamestyle{xC}^*}
\newcommand{\xDstar}{\logicnamestyle{xD}^*}
\newcommand{\xTstar}{\logicnamestyle{xT}^*}

\newcommand{\xNstar}{\logicnamestyle{xN}^*}
\newcommand{\xPstar}{\logicnamestyle{xP}^*}

\newcommand{\hilbertaxiomstyle}[1]{${#1}$}

\newcommand{\axT}{\hilbertaxiomstyle{T}}
\newcommand{\axD}{\hilbertaxiomstyle{D}}

\newcommand{\ax}{\AxiomC}
\newcommand{\uinf}{\UnaryInfC}
\newcommand{\binf}{\BinaryInfC}

\newcommand{\llab}{\LeftLabel}
\newcommand{\rlab}{\RightLabel}
\newcommand{\disp}{\DisplayProof}

\newcommand{\semcond}[1]{{#1}}

\newcommand{\cN}{\semcond N}

\newcommand{\cC}{\semcond C}
\newcommand{\cT}{\semcond T}
\newcommand{\cP}{\semcond P}
\newcommand{\cD}{\semcond D}

\newcommand{\cX}{\semcond X}

\newcommand{\cCs}{\semcond{C-s}}
\newcommand{\cDs}{\semcond{D-s}}
\newcommand{\cTs}{\semcond{T-s}}

\newcommand{\seq}{\Seq}    
\newcommand{\Seq}{\Rightarrow}

\newcommand{\tto}{\imp\coimp}
\newcommand{\bigand}{\bigwedge}    
\newcommand{\bigor}{\bigvee}
\newcommand{\vd}{\vdash}
\newcommand{\Vd}{\Vdash}
\newcommand{\ie}{i.e.}
\newcommand{\ih}{i.h.}

\newcommand{\Gone}[1]{\mathsf{G1\text{-}{#1}}}

\newcommand{\gtrecp}{\mathbf{G3cp}}

\newcommand{\gtremp}{\mathbf{G3mp}}

\newcommand{\gtreip}{\mathbf{G3ip}}
\newcommand{\gtreipS}{\mathbf{G3ip'}}

\newcommand{\wij}{Wijesekera}

\newcommand{\intu}{intuitionistic}
\newcommand{\const}{constructive}
\newcommand{\neigh}{neighbourhood}
\newcommand{\seg}{segment}

\newcommand{\fint}{\iota}

\newcommand{\lwk}{$\mathsf{Lwk}$}

\newcommand{\mlwk}{$\mathsf{wk}_{\mathsf{L}}^m$}
\newcommand{\mlctr}{$\mathsf{ctr}_{\mathsf{L}}^m$}
\newcommand{\lwkcl}{$\mathsf{wk}_{\mathsf{L}}^{cl}$}
\newcommand{\rwkcl}{$\mathsf{wk}_{\mathsf{R}}^{cl}$}
\newcommand{\lctrcl}{$\mathsf{ctr}_{\mathsf{L}}^{cl}$}
\newcommand{\rctrcl}{$\mathsf{ctr}_{\mathsf{R}}^{cl}$}

\newcommand{\ilwk}{$\mathsf{wk}_{\mathsf{L}}^i$}
\newcommand{\irwk}{$\mathsf{wk}_{\mathsf{R}}^i$}
\newcommand{\ilctr}{$\mathsf{ctr}_{\mathsf{L}}^i$}

\newenvironment{proof-sketch}{\noindent{\em Sketch of Proof.}\hspace*{0.5em}}{\qed\medskip}

\newcommand{\cut}{$\mf{cut}$}

\newcommand{\mix}{$\mf{mix}$}
\newcommand{\mcut}{\cut}

\newcommand{\axfour}{\hilbertaxiomstyle{4}}
\newcommand{\axfive}{\hilbertaxiomstyle{5}}
\newcommand{\axB}{\hilbertaxiomstyle{B}}

\newcommand{\nec}{\hilbertaxiomstyle{nec}}
\newcommand{\monbox}{\hilbertaxiomstyle{mon_\Box}}
\newcommand{\mondiam}{\hilbertaxiomstyle{mon_\diam}}
\newcommand{\axfourbox}{\hilbertaxiomstyle{4_\Box}}

\newcommand{\axCdiam}{\hilbertaxiomstyle{C_\Diamond}}
\newcommand{\axNdiam}{\hilbertaxiomstyle{N_\Diamond}}

\newcommand{\axCbox}{\hilbertaxiomstyle{C_\Box}}
\newcommand{\axNbox}{\hilbertaxiomstyle{N_\Box}}

\newcommand{\axKdiam}{\hilbertaxiomstyle{K_\Diamond}}

\newcommand{\axKbox}{\hilbertaxiomstyle{K_\Box}}
\newcommand{\axTbox}{\hilbertaxiomstyle{T_\Box}}
\newcommand{\axTdiam}{\hilbertaxiomstyle{T_\diam}}
\newcommand{\axPbox}{\hilbertaxiomstyle{P_\Box}}
\newcommand{\axPdiam}{\hilbertaxiomstyle{P_\diam}}

\newcommand{\axdual}{\hilbertaxiomstyle{dual}}
\newcommand{\axdualand}{\axmnc}
\newcommand{\Rdualand}{\Rmnc}
\newcommand{\axdualor}{\axmem}
\newcommand{\Rdualor}{\Rmem}
\newcommand{\axmnc}{\hilbertaxiomstyle{mnc}}
\newcommand{\Rmnc}{\hilbertaxiomstyle{Rmnc}}
\newcommand{\axmem}{\hilbertaxiomstyle{mem}}
\newcommand{\Rmem}{\hilbertaxiomstyle{Rmem}}

\newcommand{\rulestyle}[1]{$\mathsf{#1}$}
\newcommand{\ruleKbox}{\rulestyle{K_\Box}}
\newcommand{\ruleKdiam}{\rulestyle{K_\diam}}

\newcommand{\ruleNdiam}{\rulestyle{N_\diam}}
\newcommand{\ruleMbox}{\rulestyle{M_\Box}}
\newcommand{\ruleMdiam}{\rulestyle{M_\diam}}
\newcommand{\ruleCbox}{\rulestyle{C_\Box}}
\newcommand{\ruleCdiam}{\rulestyle{C_\diam}}
\newcommand{\ruleTbox}{\rulestyle{T_\Box}}
\newcommand{\ruleTdiam}{\rulestyle{T_\diam}}

\newcommand{\rulePbox}{\rulestyle{P_\Box}}
\newcommand{\rulePdiam}{\rulestyle{P_\diam}}

\newcommand{\ruleKboxcl}{$\mf K_\Box^{c\ell}$}
\newcommand{\ruleKdiamcl}{$\mf K_\diam^{c\ell}$}
\newcommand{\ruleNboxcl}{$\mf N_\Box^{c\ell}$}
\newcommand{\ruleNdiamcl}{$\mf N_\diam^{c\ell}$}
\newcommand{\ruleMboxcl}{$\mf M_\Box^{c\ell}$}
\newcommand{\ruleMdiamcl}{$\mf M_\diam^{c\ell}$}
\newcommand{\ruleCboxcl}{$\mf C_\Box^{c\ell}$}
\newcommand{\ruleCdiamcl}{$\mf C_\diam^{c\ell}$}
\newcommand{\ruleTboxcl}{$\mf T_\Box^{c\ell}$}
\newcommand{\ruleTdiamcl}{$\mf T_\diam^{c\ell}$}
\newcommand{\ruleDcl}{$\mf D^{c\ell}$}
\newcommand{\rulePboxcl}{$\mf P_\Box^{c\ell}$}
\newcommand{\rulePdiamcl}{$\mf P_\diam^{c\ell}$}
\newcommand{\ruleCDcl}{$\mf{CD}^{c\ell}$}
\newcommand{\ruleDboxcl}{$\mf D_\Box^{c\ell}$}
\newcommand{\ruleDdiamcl}{$\mf D_\diam^{c\ell}$}
\newcommand{\rulemncMcl}{$\mf{mnc}_{\mf M}^{c\ell}$}
\newcommand{\rulememMcl}{$\mf{mem}_{\mf M}^{c\ell}$}
\newcommand{\rulemncCcl}{$\mf{mnc}_{\mf C}^{c\ell}$}
\newcommand{\rulememCcl}{$\mf{mem}_{\mf C}^{c\ell}$}

\newcommand{\ruleiNdiam}{\rulestyle{N_\diam}}

\newcommand{\ruleiTbox}{$\mathsf{T}_{\Box}^i$}

\newcommand{\ruleiDbox}{\rulestyle{D_\Box}}

\newcommand{\ruleiPbox}{\rulestyle{P_\Box}}

\newcommand{\ruleiCDbox}{\rulestyle{CD_\Box}}

\newcommand{\ruleTboxone}{\rulestyle{T_\Box^{1}}}
\newcommand{\ruleTboxzero}{\rulestyle{T_\Box^{0}}}

\newcommand{\rulemKbox}{$\mathsf{K}_{\Box}^m$}
\newcommand{\rulemKdiam}{$\mathsf{K}_{\diam}^m$}
\newcommand{\rulemCbox}{$\mathsf{C}_{\Box}^m$}
\newcommand{\rulemCdiam}{$\mathsf{C}_{\diam}^m$}
\newcommand{\rulemTbox}{$\mathsf{T}_{\Box}^m$}
\newcommand{\rulemTdiam}{$\mathsf{T}_{\diam}^m$}
\newcommand{\rulemD}{$\mathsf{D}^m$}
\newcommand{\rulemNbox}{$\mathsf{N}_{\Box}^m$}
\newcommand{\rulemMbox}{$\mathsf{M}_{\Box}^m$}
\newcommand{\rulemMdiam}{$\mathsf{M}_{\diam}^m$}
\newcommand{\rulemPdiam}{$\mathsf{P}_{\diam}^m$}
\newcommand{\rulemCD}{$\mathsf{CD}^m$}

\newcommand{\less}{\leq}
\newcommand{\more}{\geq}

\newcommand{\imp}{\supset}
\newcommand{\coimp}{\subset}

\newcommand{\nb}[1]{\textcolor{red}{$\vert$}\mbox{}\marginpar{\scriptsize\raggedright\textcolor{red}{#1}}}

\newcommand{\logicnamestyle}[1]{\mathsf{#1}}

\newcommand{\classical}{\mathsf{Classical}}

\newcommand{\logic}{\logicnamestyle{L}}
\newcommand{\Lone}{\logicnamestyle{L_1}}

\newcommand{\MPL}{\ML}
\newcommand{\ML}{\logicnamestyle{MPL}}
\newcommand{\IL}{\logicnamestyle{IPL}}
\newcommand{\CL}{\logicnamestyle{CPL}}

\newcommand{\CLstar}{\logicnamestyle{C^*}}

\newcommand{\ruleimncM}{\rulestyle{mnc_M}}
\newcommand{\ruleimncC}{\rulestyle{mnc_C}}
\newcommand{\ruleimncK}{\rulestyle{mnc_K}}

\newcommand{\iinit}{$\mf{init}^i$}
\newcommand{\ilbot}{\leftrule{\bot^i}}
\newcommand{\illand}{\leftrule{\land^i}}
\newcommand{\irland}{\rightrule{\land^i}}
\newcommand{\illor}{\leftrule{\lor^i}}
\newcommand{\irlor}{\rightrule{\lor^i}}
\newcommand{\ilimp}{\leftrule{\imp^i}}
\newcommand{\irimp}{\rightrule{\imp^i}}
\newcommand{\minit}{$\mf{init}^m$}

\newcommand{\mlland}{\leftrule{\land^m}}
\newcommand{\mrland}{\rightrule{\land^m}}
\newcommand{\mllor}{\leftrule{\lor^m}}
\newcommand{\mrlor}{\rightrule{\lor^m}}
\newcommand{\mlimp}{\leftrule{\imp^m}}
\newcommand{\mrimp}{\rightrule{\imp^m}}

\newtheorem{theorem}{Theorem}[section]
\newtheorem{lemma}[theorem]{Lemma}

\newtheorem{proposition}[theorem]{Proposition}

\newtheorem{definition}{Definition}[section]

\providecommand{\keywords}[1]
{
  \begin{small}
  \noindent
  \textbf{\textit{Keywords---}} #1
  \end{small}
}

\title{Minimal modal logics, constructive modal logics and their relations}
\author{Tiziano Dalmonte}

\affil{\begin{normalsize}Free University of Bozen-Bolzano, Bolzano, Italy\\
tiziano.dalmonte@gmail.com\end{normalsize}}
\date{}

\begin{document}

\maketitle

  \begin{abstract}
\noindent
We present a family of minimal modal logics (namely, modal logics based on minimal propositional logic)
corresponding each to a different classical modal logic.
The minimal modal logics are defined based on their classical counterparts in two distinct ways:
(1) via embedding into fusions of classical modal logics through a natural extension of the Gödel-Johansson translation of minimal logic into modal logic S4;
(2) via extension to modal logics of the multi- vs.~single-succedent 
correspondence
 of sequent calculi for classical and minimal logic.
We show that, despite being mutually independent, the two methods turn out to be equivalent for a wide class of modal systems.
Moreover, we compare the resulting minimal version of K with 
the constructive modal logic CK studied in the literature,
displaying tight relations among the two systems.
Based on these relations, we also define a constructive correspondent for each minimal system,
thus obtaining a family of constructive modal logics which includes CK as well as other constructive modal logics studied in the literature.

\medskip
  \end{abstract}
  
\section{Introduction}

Although modal logics 
are usually
 defined as extensions of classical logic,
significant attention has been also 
devoted 
to the analysis of modalities over non-classical basis,
such as 
relevant 
\cite{Bilkova:2010,Ferenz:2023,Fuhrmann:1990,Mares:1993,seki2003general,seki2003sahlqvist},
linear \cite{Garg:2006,Marion:2009,Porello:2015} or other substructural logics \cite{Dagostino,Kamide:2002,Sedlar:2016}.
In this context, a major role is played by \intu\ logic, 
many modal extensions of which have been studied
with motivations 
ranging from philosophical or legal reasoning to computer science applications.

By analogy with intuitionistic connectives, 
modalities $\Box$ and $\diam$ over \intu\ logic are usually assumed to be not interdefinable.
This peculiarity allows for the definition of a wide variety of \intu\ modal 
systems,
since the modalities can validate 
distinct
principles and can interact in several ways.
In particular, different intuitionistic counterparts of the same classical modal logics are definable.
If we consider for instance \intu\ counterparts of classical modal logic $\K$,
three systems are well-attested in the literature:
so-called Intuitionistic K ($\IK$) \cite{FischerServi:1980,fischerservi:1984,Simpson:1994,Plotkin}, 
(the propositional fragment of) \wij's Constructive Concurrent Dynamic Logic~\cite{wij,Wijesekera:2005}
(we call it $\WK$), 
and Constructive K ($\CK$) 
\cite{Bellin,depaiva:2011,Mendler1}.%
\footnote{This list includes only systems with both modalities $\Box$ and $\diam$, however several intuitionistic mono-modal versions of $\K$ have been also studied (see \cite{Simpson:1994} for a survey).}
Axiomatically, these systems have increasing strength, from the weakest $\CK$ to the strongest $\IK$, 
and are definable extending intuitionistic propositional logic ($\IPL$) as follows:
\begin{center}
\begin{tabular}{l}
$\CK$ := $\IPL$ + $\Box(A \imp B) \imp (\Box A \imp \Box B)$, $\Box(A \imp B) \imp (\diam A \imp \diam B)$, \ax{$A$}\uinf{$\Box A$}\disp \\
$\WK$ := $\CK$ + $\neg\diam\bot$ \\
$\IK$ := $\WK$ +  $\diam(A \lor B) \imp \diam A \lor \diam B$,  $(\diam A \imp \Box B) \imp \Box(A \imp B)$ \\
\end{tabular}
\end{center}

Once an intuitionistic counterpart of a classical modal logic is defined, 
the question arises about how to define analogous counterparts of other classical modal logics.
To this aim, general theories of intuitionistic modal logics have been proposed.
An elegant characterisation
 of $\IK$ and related systems was 
 provided by Simpson~\cite{Simpson:1994}:
Given a classical modal logic $\logic$, its intuitionistic counterpart $\ILvar$ is defined as the set of formulas $A$
such that
\[A \in \ILvar \text{ if and only if } \mathscr T\vd_{\mathsf{QIL}} \forall x(A^x)\]
where the superscript $x$ denotes the well-known standard translation of modal formulas into fist-order sentences 
 with respect to the variable $x$,
and $\mathscr T$ is a set of first-order sentences 
that express
the frame conditions 
corresponding to
 the modal axioms of $\logic$ in the relational semantics
(e.g., reflexivity, transitivity, simmetry for \axT, \axfour, \axB).
The logics $\ILvar$ have been also shown to be related to products of modal logics \cite{YB}
via suitable translations of modal formulas \cite{FischerServi:1980,wolter:1999}.

Furthermore, a general 
characterisation of \wij-style modal logics was 
given in \cite{dalmonte:2022}
on the basis of a restriction of sequent calculi for classical modal logics to sequents with at most one formula in the succedent,
thus extending to modal logics a relation that is known to hold between classical and intuitionistic sequent calculi since Gentzen~\cite{gentzen}.
At the same time, the couterparts $\WLvar$ of classical logics $\logic$ presented in \cite{dalmonte:2022} coincide with the sets of modal formulas $A$ such that
\[A \in \WLvar \text{ if and only if } A^g \in \Sfour\oplus\logic\]
where $\Sfour\oplus\logic$ denotes the fusion of $\Sfour$ and $\logic$ \cite{YB},
and $g$ is a natural extension of G{\"o}del 
translation of $\IPL$ into $\Sfour$ \cite{goedel1933}.


By contrast, no such general criterion for the definition of $\CK$ and related systems have been proposed so far in the literature.
Constructive counterparts for 
the whole
modal cube from $\K$ to $\Sfive$ 
obtained  by the
combinations of  the axioms \axD, \axT, \axfour, \axB\ and \axfive\
have been  presented in \cite{arisaka,Alechina,Mendler2}
and 
endowed with nested sequent calculi \cite{arisaka}
and some of them also with relational semantics \cite{Alechina,Mendler1} and
2-sequent calculi \cite{Mendler2}.
These logics are defined based on their axiomatic systems 
by extending $\CK$ with pairs of corresponding $\Box$- and $\diam$-axioms,
like \axT\ $(\Box A \imp A) \land (A \imp \diam A)$ and \axfour\
$(\Box A \imp \Box \Box A) \land (\diam \diam A \imp \diam A)$,
whose need comes from the loss of the duality between $\Box$ and $\diam$ in intuitionistic logic
(the axiom \axD\ 
is an exception as it is taken in the usual formulation $\Box A \imp \diam A$ only).

However, 
the validity of corresponding $\Box$- and $\diam$-principles in constructive modal logics does
not hold in general,
so that it is not clear how to extend this family of systems
from a purely axiomatical point of view,
especially when it comes to logics weaker than $\K$.
For instance, $\CK$ validates 
$\Box A \land \Box B \imp \Box(A \land B)$ and the necessitation rule $A/\Box A$,
but does not validate the corresponding $\diam$-principles 
$\diam(A \lor B) \imp \diam A \lor \diam B$ and $\neg\diam \bot$.
Moreover, $\Box$- and $\diam$-versions also exist for the axiom \axD,
namely $\neg(\Box A \land \Box\neg A)$ and $\diam A \lor \diam\neg A$,
but neither of them is derivable in $\CKD$.

In this paper we address this problem by 
presenting 
a systematisation of constructive modal logics that includes the systems
$\CK$, $\CKD$ and $\CKT$ already studied in the literature 
as well as new constructive counterparts of further classical modal logics.
At the same time, we propose a view on 
 constructive modal logics according to which the constructive modalities are 
 strictly connected with minimal modalities.

 More precisely, we first introduce a family of minimal modal logics 
 corresponding each to a different classical modal logic.
 To our knowledge, this is the first study of modal logics based on
 minimal propositional logic ($\MPL$).
 Following an approach similar to the one of \cite{dalmonte:2022} for the definition of \wij-style modal logics,
 our minimal modal logics are defined in two distinct but ultimately equivalent ways,
 at least for the set of classical logics considered here.
 First, the minimal modal logics $\MLvar$ are defined by means of a reduction to fusions of classical modal logics of the form $\Sfour\oplus\logic$ through a natural extension of the Gödel-Johansson translation of $\MPL$ into $\Sfour$
 (in turn, this translation is a combination of Johansson's translation of $\MPL$ into $\IPL$ \cite{Johansson:1937}
 and Gödel's translation of $\IPL$ into $\Sfour$ \cite{goedel1933}):
The logics $\MLvar$ will coincide with the sets of modal formulas $A$ such that
$A^{\tr} \in \Sfour\oplus\logic$,
where $\tr$ is the aforementioned translation.
To this aim, we provide the minimal modal logics with a modular semantic characterisation.
Second, the same minimal modal logics are defined by restricting sequent calculi for classical modal logics to sequents with exactly one formula in the succedent.
This extends to modal logics a relation that holds between sequent calculi for classical and minimal propositional logic
fistly observed by Johansson \cite{Johansson:1937}
(see \cite{Troelstra:2000} for an extended presentation of sequent calculi for $\CPL$, $\IPL$ and $\MPL$ and corresponding bounds on the cardinality of succedents of sequents).

Then, we observe that $\MK$, the minimal counterpart of $\K$, 
is strictly connected with $\CK$.
In particular, $\CK$ coincides with the extension of $\MK$ with the principle of ex falso quodlibet
$\bot\imp A$ (exactly as $\IPL = \MPL + \bot\imp A$),
which means that 
the two systems share the same modal principles, despite over a different propositional base.
We show that tight relations between $\MK$ and $\CK$ can be observed also in terms of the semantics and of the sequent calculi.
More precisely, $\CK$ is characterised by the class of model for $\MK$ with suitable restrictions which ensure the validity of $\bot\imp A$,
moreover the sequent calculus for $\CK$ defined in \cite{Bellin} can be obtained by adding the modal sequent rules for $\MK$ to an intuitionistic sequent calculus.
By extending these relations to the other systems, we then define a constructive correspondent for each minimal modal logic,
thus obtaining a family of constructive modal logics with corresponding semantics and sequent calculi.

This paper is organised as follows.
In Sections~\ref{subsec:syn prel} and \ref{subsec:sem prel}, we present some 
preliminary notions needed throughout the paper.
In Section~\ref{sec:MK companion}, we present the definition of $\MK$, the minimal counterpart of $\K$,
via a reduction to $\Sfour\oplus\K$.
We also prove the soundness and completeness of $\MK$ with respect to a suitable class of models.
In Section~\ref{sec:seq MK}, we present the definition of $\MK$ by means of the single-succedent restriction of a sequent calculus
for $\K$. We show that the logic obtained in this way coincides with the logic defined in the previous section.
In Section~\ref{sec:rel MK CK}, we analyse the relations between $\MK$ just defined and $\CK$,
both from the point of view of the semantics and of the sequent calculi.
In Section~\ref{sec:family}, we apply the same methods, based on reductions to fusions and restriction of sequent calculi,
to define a family of minimal counterparts of some standard classical modal logics.
By extending to these systems the relations observed between $\MK$ and $\CK$, we also define a corresponding family of constructive modal logics.
Finally, Section~\ref{sec:discussion} contains some discussion of the results.

%
%

\subsection{Syntactic preliminaries}\label{subsec:syn prel}
Given a countable set  $\atm = \{p_0, p_1, p_2, ... \}$ of propositional variables and a finite set $\mathbb{M}$ of unary modal operators, the language $\lan_\mathbb{M}^{\atm}$ is defined by the following BNF grammar,
where $p\in\atm$,   
  $\circ \in \{\land, \lor, \imp\}$,
  and   
  $\heartsuit\in\mathbb{M}$:
  \[A ::= p \mid \bot \mid A \circ A \mid \heartsuit A.\] 

  In the following, we use $p$, $q$, $r$ as metavariables for elements of $\atm$, and $A$, $B$, $C$, $D$ as metavariables for formulas. Moreover, we define 
  $\top := \bot\imp\bot$,
  $\neg A := A \imp \bot$,
  and $A \tto B:=(A \imp B) \land (B \imp A)$.
  Minimal modal logics will be defined in a language $\lan_{\{\Box, \diam\}}^{\atm}$
  containing the modalities $\Box$ and $\diam$.
  For the sake of simplicity, we denote $\lan_{\{\Box, \diam\}}^{\atm}$ as $\lan$.
  
  We consider the following axiomatisation for minimal propositional logic ($\MPL$), formulated in $\lan$
  (see e.g.~\cite{segerbergMin}):%
  \footnote{In this paper we consider axiomatic systems to be defined by axiom schemata and rule schemata. For the sake of simplicity, we simply refer to axiom schemata and rule schemata as axioms and rules.}
\begin{center}
\begin{tabular}{lll}
$A\land B \imp A$ \ &  $(A \imp B) \imp ((A\imp C) \imp (A\imp B\land C))$ \\
$A\land B\imp B$ & $(A \imp C) \imp ((B\imp C) \imp (A\lor B\imp C))$  \\
$A\imp A \lor B$ & $(A \imp (B \imp C)) \imp ((A\imp B) \imp (A\imp C))$ \\ 
$B\imp A \lor B$ & $A \imp (B \imp A)$ \\
\end{tabular}
\
\ax{$A\imp B$}
\ax{$A$}
\binf{$B$}
\disp
\end{center}
As usual, we can define intuitionistic propostitional logic ($\IPL$) as the extension of $\MPL$ with \emph{ex falso quodlibet}:%
\footnote{Given two axiomatic systems $\logic$ and $\logic'$ and an axiom $A$, we denote
$\logic + A$ the axiomatic extension of $\logic$ with the axiom $A$, and $\logic \oplus \logic'$ the fusion
of $\logic$ and $\logic'$ (cf.~Definition~\ref{def:fusion}).}
\[\IPL := \MPL + \bot\imp A,\]
and classical propostitional logic ($\CPL$) as the extension of $\IPL$ with \emph{excluded middle}:
\[\CPL := \IPL + A \lor \neg A.\]

%
  
  In this paper, we shall define minimal and constructive counterparts of classical modal logics.
  The classical modal logics here considered are extensions of
  $\CPL$, formulated in $\lan$,  with modal axioms and rules.
For instance, the well-known logic $\K$ is defined extending $\CPL$ with%
\footnote{As usual, classical monomodal logics can be equivalently defined based on a single modality $\Box$, taking $\diam$ 
defined as $\diam A \tto \neg\Box\neg A$ (or viceversa).
Since $\Box$ and $\diam$ are not interdefinable in minimal modal logics, we prefer to assume both $\Box$ and $\diam$ primitive also in the classical systems in order to simplify the comparison.}
\begin{center}
\axKbox \   $\Box(A \imp B) \imp (\Box A \imp \Box B)$ \qquad 
\axdual \   $\Box A \tto \neg\diam\neg A$
\qquad
\ax{$A$}\llab{\nec}\uinf{$\Box A$}\disp,
\end{center}
and $\Sfour$ is defined extending $\K$ with
\axTbox \   $\Box A \imp A$
and
\axfourbox \   $\Box A \imp \Box\Box A$.
Further classical modal logics will be introduced in Section~\ref{sec:family}.
  
  For any logic defined in the following
  (no matter whether based on classical, intuitionistic or minimal logic),
  we consider the standard notions of derivability:  
  Given a logic $\logic$ formulated in $\lan_{\mathbb M}^{\atm}$
  and formulas $A, B_1, ..., B_n$ of $\lan_{\mathbb M}^{\atm}$,
the rule 
$B_1, ..., B_n / A$ is \emph{derivable} in $\logic$
if there is a finite sequence of formulas 
ending with $A$ in which every formula is an (instance of an) axiom of $\logic$, 
or it belongs to $\{B_1, ..., B_n\}$,
or it is obtained from previous formulas by the application of a rule of $\logic$.
A formula $A$ is
  \emph{derivable} in $\logic$, 
written 
$\logic \vd A$,
if the 
rule $\emptyset/ A$ is derivable in $\logic$.
Finally,  $A$ is (locally) \emph{derivable} 
in $\logic$ from a set of formulas $\Phi$ of $\lan_{\mathbb M}^{\atm}$,
written $\Phi\vd_\logic A$, if there is a finite set 
$\{B_1, ..., B_n\}\subseteq\Phi$ such that 
$\logic \vd B_1 \land ... \land B_n \imp A$.

\subsection{Semantic preliminaries}\label{subsec:sem prel}

We shall define semantics 
 for minimal modal logics by suitably extending relational models for minimal propositional logic.
We consider to this purpose relational models for minimal propositional logic as defined in \cite{segerbergMin}: 
A \emph{minimal relational model} 
  is a tuple $\M = \langle \W, \less, \FW, \V\rangle$, 
where 
$\W$ is a non-empty set of worlds, 
$\less$ is a reflexive and transitive binary relation on $\W$, 
$\FW\subseteq\W$ is a $\less$-upward  closed set
(that is, if $w\in\FW$ and $w\less v$, then $v\in\FW$)
of so-called \emph{fallible worlds}, and
$\V : \atm \longrightarrow \pow(\W)$ is a \emph{hereditary} valuation function
(that is, if $w\in\V(p)$ and $w\less v$, then $v\in\V(p)$).
We write $v\more w$ for $w \less v$.
The forcing relation $\M, w\Vd A$ is 
inductively defined as follows:
\begin{center}
\begin{tabular}{lllll}
$\M, w \Vd p$ & iff & $w\in\V(p)$; \\
$\M, w \Vd \bot$ & iff & $w\in\FW$; \\
$\M, w \Vd B \land C$ & iff & $\M, w \Vd B$ and $\M, w \Vd C$; \\
$\M, w \Vd B \lor C$ & iff & $\M, w \Vd B$ or $\M, w \Vd C$; \\
$\M, w \Vd B \imp C$ & iff & for all $v\more w$, $\M, v \Vd B$ implies $\M, v \Vd C$. \\
\end{tabular}
\end{center}

Given a formula $A$, we say that $A$ is \emph{valid in a model} $\M$, 
written $\M\models A$,
if $\M, w \Vd A$ for all worlds $w$ of $\M$.
The same definition of validity applies to all kinds of models considered in this paper.
In the following, we simply write $w \Vd A$ when $\M$ is clear from the context.

Minimal relational models are a generalisation of the well-known intuitionistic relational models firstly introduced by Kripke~\cite{kripkeInt}.
In particular, a minimal relational model is an \emph{intuitionistic relational model} if $\FW = \emptyset$.
We point out that an alternative semantics for $\IPL$ can be obtained from minimal relational models
by preserving the fallible worlds but assuming the condition 
$\FW\subseteq\V(p)$ for all $p\in\atm$
which ensures the validity of ex falso quodlibet $\bot\imp A$.
We will consider this latter kind of restriction in 
Section~\ref{subsec:rel MK CK sem}.

  \section{Minimal K via bimodal companion}\label{sec:MK companion}
  Our first approach toward
   the definition of 
   minimal modal logics is based on 
  reductions into fusions of classical modal logics. 
  We consider to this purpose the following notions
  of fusion, 
  G{\"o}del-Johansson translation and bimodal companion.

\begin{definition}[Fusion]\label{def:fusion}
 Let $\logic_1$ and $\logic_2$ be \emph{classical} modal logics respectively defined in the languages $\lan_{\{\Box_1, \diam_1\}}^{\atm'}$
and $\lan_{\{\Box_2, \diam_2\}}^{\atm'}$
sharing the same propositional variables and propositional connectives but with disjoint sets of modalities. 
The \emph{fusion} $\logic_1\oplus\logic_2$ of $\logic_1$ and $\logic_2$ is the smallest logic in the language $\lan_{\{\Box_1, \diam_1,\Box_2, \diam_2\}}^{\atm'}$ containing $\logic_1 \cup \logic_2$
and closed under the rules of $\logic_1$ and $\logic_2$.
We denote $\lan_{\{\Box_1, \diam_1,\Box_2, \diam_2\}}^{\atm'}$ 
as $\lan_{1,2}'$.
\end{definition}

\begin{definition}[Extended G{\"o}del-Johansson translation]\label{def:translation}
Let $\atm' = \atm \cup \{f\}$, with $f\notin\atm$.
The \emph{extended G{\"o}del-Johansson translation} 
$\tr: \lan\longrightarrow
\lan_{1,2}'$
is inductively defined as follows:
\begin{center}
\begin{tabular}{rcl}
$\bot^{\tr}$ & $=$ & $\BoxI\fp$ \\

$p^{\tr}$ & $=$ & $\BoxI p$ \\

$(A \land B)^{\tr}$ & $=$ & $A^{\tr} \land B^{\tr}$ \\

$(A \lor B)^{\tr}$ & $=$ & $A^{\tr} \lor B^{\tr}$ \\

$(A \imp B)^{\tr}$ & $=$ & $\BoxI (A^{\tr} \imp B^{\tr})$ \\

$(\Box A)^{\tr}$ & $=$ & $\BoxI\BoxM A^{\tr}$ \\

$(\diam A)^{\tr}$ & $=$ & $\BoxI\diamM A^{\tr}$ \\
\end{tabular}
\end{center}
\end{definition}

\begin{definition}[Bimodal companion]\label{def:bim companion}
A fusion of classical modal logics $\logic_1\oplus\logic_2$ 
in the language 
$\lan_{1,2}'$
is the \emph{bimodal companion}of a minimal modal logic $\Mvar$ in the langiage $\lan$ if it holds:
\[\Mvar \vd A \text{ \  if and only if \ } \logic_1\oplus\logic_2 \vd A^{\tr}.\]
\end{definition}


The above translation $\tr$ is based on 
G{\"o}del's \cite{goedel1933} reduction of
intuitionistic propositional logic into $\Sfour$.
The clauses for the modal formulas extend the translation in the trivial way,
and are considered for instance in
\cite{Fairtlough:1997,wolter:1999},
while the reduction of $\bot$ into a distinguished propositional constant $\fp$ goes back to
 Johansson~\cite{Johansson:1937}.
A similar translation that employs 
Johannson's solution was already applied 
for the embedding
of a constructive modal logic into a classical multimodal logic in \cite{Fairtlough:1997}.

Given a classical modal logic $\logic$, 
we shall define its minimal counterpart $\MLvar$ 
by considering modal companions of the form $\Sfour\oplus\logic$.

\begin{definition}[Minimal counterpart of a classical logic]\label{def:minimal counterpart}
Given a classical modal logic $\logic$, 
the \emph{minimal counterpart} of $\logic$ is the logic 
$\MLvar$ in $\lan$ such that
$\Sfour\oplus\logic$ is the bimodal companion of $\MLvar$.
\end{definition}

In other words, the minimal counterpart $\MLvar$ 
of $\logic$ is the solution of the equation
\begin{center}
$(\ast)$ \quad $\MLvar \vd A$ \ if and only if \ $\Sfour\oplus\logic \vd A^{\tr}$.
\end{center}

Clearly, 
the solution of ($\ast$), if it exists,  is unique
(modulo equivalent axiomatisations).
Indeed, if both $\MLvar$ and $\MLvar'$ are solutions to ($\ast$),
then $\MLvar \vd A$ iff $\Sfour\oplus\logic \vd A^{\tr}$ iff $\MLvar' \vd A$, hence $\MLvar = \MLvar'$.
%
We start by presenting 
the minimal counterpart of the classical modal logic $\K$.

%

\begin{definition}[Minimal $\K$]
The minimal modal logic $\MK$ is defined extending
$\MPL$ with the following axioms and rule:
\begin{center}
\axKbox \   $\Box(A \imp B) \imp (\Box A \imp \Box B)$ \quad 
\axKdiam \  $\Box(A \imp B) \imp (\diam A \imp \diam B)$
\quad
\ax{$A$}\llab{\nec}\uinf{$\Box A$}\disp
\end{center}
\end{definition}

In order to prove that $\MK$ 
is the minimal counterpart of $\K$
(that is, it is the solution of the equation $(\ast)$, for $\logic$ replaced with $\K$), 
we first 
provide a semantics for $\MK$,
which is
defined by suitably extending
minimal relational models for $\MPL$ (cf.~Section~\ref{subsec:sem prel}) 
with an additional relation dealing with the modalities.

\begin{definition}
[Minimal birelational semantics]\label{def:birelational semantics}
A \emph{minimal birelational model} is a tuple
$\M = \langle \W, \less, \FW, \R, \V\rangle$, 
where $\langle \W, \less, \FW, \V\rangle$ is a minimal relational model,
and $\R$ is a binary relation on $\W$.
The forcing relation $\M, w\Vd A$ is 
inductively defined extending the clauses
for $p, \bot, \land, \lor, \imp$
in Section~\ref{subsec:sem prel}
with the following
clauses for the modalities:
\begin{center}
\begin{tabular}{lllll}
$\M, w \Vd \Box B$ & iff & for all $v\more w$, for all $u$, if $v\R u$, then $\M, u \Vd B$; \\
$\M, w \Vd \diam B $ & iff & for all $v\more w$, there is $u$ such that $v\R u$ and $\M, u \Vd B$. \\
\end{tabular}
\end{center}
\end{definition}

The semantics in Definition~\ref{def:birelational semantics} essentially coincides with 
Wijesekera's \cite{wij} semantics for
(the propositional fragment of)
Constructive Concurrent Dynamic Logic
with the only difference of the addition of the fallible worlds,
which is due to the fact that the base
models are minimal rather than intuitionistic.

The generalisation of the standard clauses for $\Box, \diam$ in the relational semantics
to all $\less$-successors
is the simplest way to preserve the hereditary property of minimal relational models.

\begin{proposition}[Hereditary property]
Given a minimal birelational model $\M$ and a formula $A$ of $\lan$, for every worlds $w$ and $v$ of $\M$ it holds:
If $w \Vd A$ and $w \less v$, then $v\Vd A$.
\end{proposition}
\begin{proof}
Immediate by induction on the construction of $A$.
\end{proof}

\begin{theorem}[Soundness]\label{th:soundness MK}
For all $A\in\lan$,
if $A$ is derivable in $\MK$, then $A$ is valid in every minimal birelational model.
\end{theorem}
\begin{proof}
By showing that all the axioms and rules of $\MK$ are valid, respectively validity preserving, in every model for $\MK$.
We consider the modal principles.

\begin{enumerate}[leftmargin=*, align=left]
\item[(\axKbox)]
Suppose that $w \Vd \Box(A \imp B)$ and $w \Vd \Box A$. Then for all $v, u$, if $w \less v$ and $v \R u$, then
$u \Vd A \imp B$ and $u \Vd A$, hence $u \Vd B$. 
Thus, $w \Vd\Box B$.
Therefore $\M\models \Box(A \imp B)\imp(\Box A\imp\Box B)$.

\item[(\axKdiam)] Suppose that $w \Vd \Box(A \imp B)$ and $w \Vd \diam A$. Then for all $v$, if $w \less v$, then
there is $u$ such that $v \R u$ and $u \Vd A$. Moreover, $u \Vd A \imp B$.
Thus $u \Vd B$, hence $w \Vd \diam B$.
Therefore $\M\models \Box(A \imp B)\imp(\diam A\imp\diam B)$.

\item[(\nec)] Suppose that $\M\models A$. Then for all $w, v, u$, if $w \less v$ and $v \R u$, then $u\Vd A$, hence $w \Vd \Box A$. Therefore $\M\models \Box A$.\qedhere
\end{enumerate}
\end{proof}



We now present a completeness proof for $\MK$ with respect to the minimal birelational semantics
by the canonical model technique.
The proof adapts the completeness proof for $\WK$ by \wij~\cite{wij} with the addition of the impossible worlds.
For every logic $\logic$ in $\lan$, we call $\logic$-\emph{full} any set $\Phi$ of formulas of $\lan$ 
such that 
if $\Phi\vd_{\logic} A$, then $A\in\Phi$ (closure under derivation), and 
if $A\lor B\in\Phi$, then $A\in\Phi$ or $B\in\Phi$ (disjunction property).
Moreover, for every set of formulas $\Phi$, 
we denote $\boxm\Phi$ the set $\{A \mid \Box A\in\Phi\}$.
One can prove in a standard way the following lemma.

\begin{lemma}[Lindenbaum]\label{lemma:lind rel}
For every set $\Phi$ of formulas of $\lan$, there is a $\MK$-full set $\Psi$ such that $\Phi\subseteq\Psi$.
Moreover, if $\Phi\not\vd_{\MK} A$, then there is a $\MK$-full set $\Psi$ such that $\Phi\subseteq\Psi$ and $A\notin\Psi$.
\end{lemma}


\begin{definition}\label{def:rel seg}
For every logic $\logic$ in $\lan$, 
an $\logic$-\emph{relational \seg},  
or just \emph{\seg}, is a pair $(\Phi, \U)$, 
where $\Phi$ is an $\logic$-full set, and $\U$ is a set of $\logic$-full sets such that:
\begin{itemize}
\item if $\Box A\in\Phi$, then for all $\Psi\in\U$, $A\in\Psi$; and
\item if $\diam A\in\Phi$, then there is $\Psi\in\U$ such that $A\in\Psi$.
\end{itemize}
\end{definition}

The following holds.

\begin{lemma}\label{lemma:rel seg}
For every $\MK$-full set $\Phi$, there exists an $\MK$-relational segment $(\Phi, \U)$.
\end{lemma}
\begin{proof}
Given a $\MK$-full set $\Phi$, we define
$\U = \{\Psi \ \MK\textup{-full} \mid \boxm{\Phi}\subseteq\Psi \textup{ and } B\in\Psi
\textup{ for some } \diam B \in \Phi \}$.
Then by definition, for all $\Box A\in\Phi$ and all $\Psi\in\U$, $A\in\Psi$.
Moreover, suppose that $\diam A \in \Phi$.
By Lemma~\ref{lemma:lind rel}, there is an $\MK$-full set $\Psi$ such that $\boxm{\Phi}\cup\{A\}\subseteq\Psi$,
then $A\in\Psi$ and $\Psi\in\U$. 
Hence $(\Phi,\U)$ is an $\MK$-segment.
\end{proof}

\begin{definition}\label{def:can model rel}
For every logic $\logic$ in $\lan$,
the \emph{canonical birelational model} for $\logic$ is
the tuple $\Mc = \langle \Wc, \lessc, \FWc, \Rc, \Vc \rangle$,
where:
\begin{itemize}
\item $\Wc$ is the class of all $\logic$-relational segments;
\item for all $(\Phi, \U), (\Psi, \VV)\in\Wc$,  $(\Phi, \U) \lessc (\Psi, \VV)$ 
if and only if $\Phi\subseteq\Psi$;
\item for all $(\Phi, \U)\in\Wc$,  $(\Phi,\U)\in\FWc$ if and only if  $\bot\in\Phi$;
\item for all $(\Phi, \U), (\Psi, \VV)\in\Wc$,  $(\Phi, \U) \Rc (\Psi, \VV)$ 
if and only if $\Psi\in\U$;
\item for all $(\Phi, \U)\in\Wc$,  $(\Phi,\U)\in\Vc(p)$ if and only if $p\in\Phi$.
\end{itemize}
\end{definition}

It is easy to see that the canonical birelational model for $\MK$ is a minimal birelational model
(Definition~\ref{def:birelational semantics}).
We prove the following lemma.

\begin{lemma}\label{lemma:can model rel}
Let $\Mc = \langle \Wc, \lessc, \FWc, \Rc, \Vc \rangle$ be the canonical birelational model for $\MK$. Then for all $(\Phi, \U)\in\Wc$ and all $A \in\lan$, 
$(\Phi, \U)\Vd A$ if and only if $A\in\Phi$.
\end{lemma}
\begin{proof}
By induction on the construction of $A$.
For the base 
 case $A = p$ and the inductive cases $A = B \land C, B\lor C$ the proof is immediate.
We show the other cases, writing $\vd$ for $\vd_{\MK}$.

\begin{enumerate}[leftmargin=*, align=left]
\item[($A = \bot$)] $(\Phi, \U)\Vd \bot$ iff $(\Phi, \U)\in\FWc$ iff, by definition, $\bot\in\Phi$.

\item[($A = B \imp C$)] 
Suppose that $B\imp C\in\Phi$, and assume $(\Phi, \U)\lessc(\Psi, \VV)$ and $(\Psi, \VV)\Vd B$.
By definition, $\Phi\subseteq\Psi$, thus $B\imp C\in\Psi$.
Moreover by \ih, $B\in\Psi$, hence $C\in\Psi$, thus by \ih, $(\Psi, \VV)\Vd C$.
Therefore $(\Phi, \U)\Vd B \imp C$.
Now suppose that $B\imp C\notin\Phi$.
Then $\Phi\not\vd B\imp C$, thus $\Phi\cup\{B\}\not\vd C$.
By Lemma~\ref{lemma:lind rel}, there is $\Psi$ 
$\MK$-full such that $\Phi\cup\{B\}\subseteq\Psi$ and $C\notin\Psi$.
Then by Lemma~\ref{lemma:rel seg} and Definition~\ref{def:can model rel}, there is an $\MK$-segment 
$(\Psi, \VV)$, and $(\Psi, \VV)\in\Wc$.
Thus by definition, $(\Phi, \U)\lessc(\Psi, \VV)$, and by \ih, $(\Psi, \VV)\Vd B$ and $(\Psi, \VV)\not\Vd C$.
Therefore $(\Phi, \U)\not\Vd B\imp C$.


\item[($A = \Box B$)] 
Suppose that $\Box B\in\Phi$.
Then for all $(\Psi, \VV)\morec(\Phi,\U)$, $\Box B\in\Psi$.
Moreover by definition, if $(\Psi, \VV)\Rc(\Theta, \ZZ)$,
then $\Theta\in\VV$,
hence by definition of segment, $B\in\Theta$.
Then by \ih, $(\Theta, \ZZ) \Vd B$, therefore $(\Phi, \U)\Vd \Box B$. 
Now suppose that $\Box B\notin\Phi$.
Then $\boxm{\Phi}\not\vd B$ 
(indeed, if $\boxm{\Phi}\vd B$, 
then $\vd C_1 \land ... \land C_n \imp B$ for some $\Box C_1, ..., \Box C_n\in\Phi$,
then by \nec, $\vd \Box(C_1 \land ... \land C_n \imp B)$,
and by \axKbox, $\vd \Box(C_1 \land ... \land C_n) \imp \Box B$;
since $\vd \Box C_1 \land ... \land \Box C_n \imp \Box(C_1 \land ... \land C_n)$,
we have $\vd \Box C_1 \land ... \land \Box C_n \imp \Box B$,
hence $\Phi\vd\Box B$, therefore $\Box B\in\Phi$,
against the assumption).
By Lemma~\ref{lemma:lind rel}, there is $\Psi$ $\MK$-full such that $\boxm\Phi\subseteq\Psi$ and $B\notin\Psi$.
We define $\VV = \{\Psi\} \cup \{\Theta \ \MK\textup{-full} \mid \boxm{\Phi}\subseteq\Theta \textup{ and } C\in\Theta \textup{ for some } \diam C \in \Phi\}$.
Given that, by Lemma~\ref{lemma:lind rel}, such a set $\Theta$ exists for every $\diam C \in \Phi$,
we have that $(\Phi, \VV)$ is an $\MK$-segment.
Furthermore, by Lemma~\ref{lemma:rel seg}, there exists an $\MK$-segment $(\Psi, \ZZ)$,
hence since $\Psi\in\VV$, by definition, $(\Phi, \VV)\Rc(\Psi, \ZZ)$.
Moreover, since $B\notin\Psi$, by \ih, $(\Psi, \ZZ)\not\Vd B$, then since $(\Phi, \U)\less(\Phi, \VV)$,
we have
$(\Phi, \U)\not\Vd\Box B$.


\item[($A = \diam B$)]
Suppose that $\diam B\in\Phi$. Then for all $(\Psi, \VV)\morec(\Phi,\U)$, $\diam B\in\Psi$.
Thus by Definition~\ref{def:rel seg}, there is $\Theta\in\VV)$ such that $B\in\Theta$, 
and by Lemma~\ref{lemma:rel seg}, there is a segment $(\Theta, \ZZ)\in\Mc$.
Moreover, by definition, $(\Psi, \VV)\Rc(\Theta, \ZZ)$, and by \ih, $(\Theta, \ZZ)\Vd B$.
It follows that $(\Phi,\U)\Vd\diam B$.
Now suppose that $\diam B\notin\Phi$.
Then for every $\diam C\in\Phi$,
$\boxm{\Phi}\cup\{C\}\not\vd B$ 
(indeed, if $\boxm{\Phi}\cup\{C\}\vd B$, 
then $\vd D_1 \land ... \land D_n \land C \imp B$ for some $\Box D_1, ..., \Box D_n\in\Phi$,
thus $\vd D_1 \land ... \land D_n \imp (C \imp B)$,
hence $\vd \Box(D_1 \land ... \land D_n) \imp \Box (C \imp B)$,
then by \axKdiam\ and valid principles,
$\vd \Box D_1 \land ... \land \Box D_n \imp (\diam C \imp \diam  B)$,
so 
$\vd \Box D_1 \land ... \land \Box D_n \land \diam C \imp \diam  B$,
which implies 
$\Phi\vd \diam B$,
hence, finally, $\diam B\in\Phi$,
against the assumption).
We define $\VV = \{\Psi \ \MK\textup{-full} \mid \boxm{\Phi}\subseteq\Psi, B \notin \Psi \textup{ and } C\in\Psi \textup{ for some } \diam C \in \Phi\}$.
By Lemma~\ref{lemma:lind rel},
such a set $\Psi$ exists for every $\diam C \in \Psi$. 
It is easy to see that $(\Phi, \VV)$ is a $\MK$-segment,
hence $(\Phi, \VV)\in\Wc$.
Moreover, by definition, for all $(\Psi, \ZZ)$ such that
$(\Phi, \VV)\Rc(\Psi, \ZZ)$,
$B\notin\Psi$, thus by \ih, $(\Psi, \ZZ)\not\Vd B$.
Given that $(\Phi, \U)\less(\Phi, \VV)$,
we obtain $(\Phi, \U)\not\Vd\diam B$.\qedhere
\end{enumerate}
\end{proof}

%
%
%
%

\begin{theorem}[Completeness]\label{th:compl MK}
For all $A\in\lan$,
if $A$ is valid in every minimal birelational model, then $A$ is derivable in $\MK$.
\end{theorem}
\begin{proof}
Suppose that 
$\MK\not\vd A$.
Then by Lemma~\ref{lemma:lind rel}, there is an $\MK$-full set $\Psi$ such that $A\notin\Psi$,
hence by Lemma~\ref{lemma:rel seg}, 
there exists an $\MK$-segment $(\Psi,\U)$. 
By Definition~\ref{def:can model rel}, $(\Psi,\U)$ belongs to the canonical model $\Mc$ for $\MK$,
then by Lemma~\ref{lemma:can model rel}, $(\Psi,\U)\not\Vd A$.
Since $\Mc$ is a minimal birelational model, 
we conclude that it is not the case that 
$A$ is valid in all models for $\MK$.
\end{proof}

Based on this semantic characterisation, 
we now show that $\MK$ is the solution to the equation ($\ast$) for $\logic$ replaced with $\K$,
hence, according to our criterion, it is the minimal counterpart of classical $\K$.

\begin{theorem}\label{th:companion MK}
For all $A\in\lan$, 
$A$ is derivable in $\MK$ if and only if $A^{\tr}$ is derivable in $\Sfour\oplus\K$.
\end{theorem}
\begin{proof}
We recall that $\Sfour\oplus\K$ is sound and complete with respect to the class of all classical birelational models $\langle \W, \R_1, \R_2, \V\rangle$, where $\R_1$ and $\R_2$ are binary relations on $\W$ and $\R_1$ is reflexive and transitive.

\smallskip
\noindent
($\Rightarrow$)
Suppose that $\Sfour\oplus\K\not\vd A^{\tr}$. 
Then there are a model $\M = \langle \W, \R_1, \R_2, \V\rangle$
for $\Sfour\oplus\K$ and a world $w$ such that $\M, w\not \Vd A^{\tr}$.
We define $\M' = \langle \W, \less, \FW, \R, \V'\rangle$ over the same set $\W$ of $\M$,
where
$\less \ = \R_1$, $\R = \R_2$,
for all $p\in\atm$, $\V'(p) = \{v \mid \textup{for all } u,  v\R_1 u \textup{ implies } u\in\V(p)\}$,
and $\FW = \{v \mid \textup{for all } u,  v\R_1 u \textup{ implies } u\in\V(\fp)\}$.
It is easy to verify that $\M'$ is a minimal birelational model,
in particular $\V(p)$ and $\FW$ are $\less$-upward closed.
We show that for all $v\in\W$ and all $B\in\lan$ it holds:
\[\M', v \Vd B \text{ \ if and only if \ } \M, v \Vd B^{\tr},\]
from which it follows that $\M', w \not\Vd A$,
therefore $\MK\not\vd A$.
The proof is by induction on the construction of $B$.
The cases $B = C \land D$ and $B = C \lor D$ are immediate by \ih.
We consider the other cases.

\begin{enumerate}[leftmargin=*, align=left]
\item[($B = p$)] $\M', v \Vd p$ iff 
$v\in\V'(p)$ iff (by definition of $\V'$)
for all $u$,  $v\R_1 u$ implies $u\in\V(p)$; iff
for all $u$,  $v\R_1 u$ implies $\M, u\Vd p$; iff
$\M, v \Vd \Box_1 p$.

\item[($B = \bot$)] $\M', v \Vd \bot$ iff 
$v\in\FW$ iff (by definition of $\FW$)
for all $u$,  $v\R_1 u$ implies $u\in\V(f)$; iff
for all $u$,  $v\R_1 u$ implies $\M, u\Vd f$; iff
$\M, v \Vd \Box_1 f$.


\item[($B = C \imp D$)] $\M', v \Vd C \imp D$ iff
for all $u\more v$, $\M', u \Vd C$ implies $\M', u \Vd D$; iff
(by definition of $\less$ and \ih)
for all $u$, if $v \R_1 u$, then $\M, u \Vd C^{\tr}$ implies $\M, u \Vd D^{\tr}$; iff
for all $u$, if $v \R_1 u$, then $\M, u \Vd C^{\tr} \imp D^{\tr}$; iff
$\M, v \Vd \Box_1(C^{\tr} \imp D^{\tr})$.

\item[($B = \Box C$)] $\M', v \Vd \Box C$ iff
for all $u\more v$, for all $z$, if $u \R z$, then $\M', z \Vd C$;
iff
(by definition of $\less$ and $\R$ and \ih)
for all $u,z$, if $v \R_1 u$ and $u \R_2 z$, then $\M, z \Vd C^{\tr}$;
iff
$\M, v \Vd \Box_1\Box_2 C^{\tr}$.

\item[($B = \diam C$)] $\M', v \Vd \diam C$ iff
for all $u\more v$, there is $z$ such that $u \R z$ and $\M', z \Vd C$;
iff
(by definition of $\less$ and $\R$ and \ih)
for all $u$, if $v \R_1 u$, then there is $z$ such that  $u \R_2 z$ and $\M, z \Vd C^{\tr}$;
iff
$\M, v \Vd \Box_1\diam_2 C^{\tr}$.
\end{enumerate}

\smallskip
\noindent
($\Leftarrow$)
Suppose that $\MK\not\vd A$. 
Then there are a model $\M = \langle \W, \less, \FW, \R, \V\rangle$
for $\MK$ and a world $w$ such that $\M, w\not \Vd A$.
We define $\M'' = \langle \W, \R_1, \R_2, \V''\rangle$ over the same set $\W$ of $\M$,
where
$\R_1 = \ \less$, $\R_2 = \R$,
for all $p\in\atm$, 
$\V''(p) = \V(p)$,
and $\V''(\fp) = \FW$.
$\M''$ is a model for $\Sfour\oplus\K$.
We show that for all $v\in\W$ and all $B\in\lan$ it holds:
\[\M, v \Vd B \text{ \ if and only if \ } \M'', v \Vd B^{\tr},\]
from which it follows that $\M'', w \not\Vd A$,
therefore $\Sfour\oplus\K\not\vd A^{\tr}$.
The proof is by induction on the construction of $B$.
The cases $B = C \land D$ and $B = C \lor D$ are immediate by \ih.
We consider the other cases.

\begin{enumerate}[leftmargin=*, align=left]
\item[($B = p$)] $\M, v \Vd p$ iff $v\in\V(p)$ iff
(since $\V$ is $\less$-closed)
for all $u \more v$, $u\in\V(p)$; iff
(by definition of $\R_1$ and $\V''$)
for all $u$, if $v\R_1 u$, then $u\in\V''(p)$; iff
$\M'', v \Vd \Box_1 p$.

\item[($B = \bot$)] $\M, v \Vd \bot$ iff $v\in\FW$ iff
(since $\FW$ is $\less$-closed)
for all $u \more v$, $u\in\FW$; iff
(by definition of $\R_1$ and $\V''$)
for all $u$, if $v\R_1 u$, then $u\in\V''(f)$; iff
$\M'', v \Vd \Box_1 f$.


\item[($B = C \imp D$)] $\M, v \Vd C \imp D$ iff
for all $u\more v$, $\M, u \Vd C$ implies $\M, u \Vd D$; iff
(by definition of $\R_1$ and \ih)
for all $u$, if $v \R_1 u$, then $\M'', u \Vd C^{\tr}$ implies $\M'', u \Vd D^{\tr}$; iff
$\M'', v \Vd \Box_1(C^{\tr} \imp D^{\tr})$.

\item[($B = \Box C$)] $\M, v \Vd \Box C$ iff
for all $u\more v$, for all $z$, if $u \R z$, then $\M, z \Vd C$;
iff
(by definition of $\R_1$ and $\R_2$ and \ih)
for all $u,z$, if $v \R_1 u$ and $u \R_2 z$, then $\M'', z \Vd C^{\tr}$;
iff
$\M'', v \Vd \Box_1\Box_2 C^{\tr}$.

\item[($B = \diam C$)] $\M, v \Vd \diam C$ iff
for all $u\more v$, there is $z$ such that $u \R z$ and $\M, z \Vd C$;
iff
(by definition of $\R_1$ and $\R_2$ and \ih)
for all $u$, if $v \R_1 u$, then there is $z$ such that  $u \R_2 z$ and $\M'', z \Vd C^{\tr}$;
iff
$\M'', v \Vd \Box_1\diam_2 C^{\tr}$.\qedhere
\end{enumerate}
\end{proof}

%

\section{Minimal K via sequent calculus}\label{sec:seq MK}

From the point of view of the sequent calculi,
classical and minimal propositional logic stay in a clear and neat relation:
Given a suitable sequent calculus $\SC$ for $\CPL$,
a calculus for $\MPL$ can be obtained by restricting the rules of $\SC$
to \emph{single-succedent} sequents,
namely sequents with exactly one formula in the succedent.
This relation is particularly evident in G1-style sequent calculi \cite{Troelstra:2000}.
In this section, we extend this relation to modal logics 
and define a minimal version of $\K$ by restricting a standard G1-calculus for $\K$ to
single-succedent sequents.
We show that the resulting logic is precisely $\MK$ introduced in the previous section.
We consider the following standard definitions.

\begin{definition}\label{def:sequent}
A \emph{sequent} is a pair $\G\Seq\D$, where $\G$ and $\D$
(respectively, the \emph{antecedent} and the \emph{succedent} of the sequent)
 are finite, possibly empty multisets of formulas of $\lan$. 
A sequent $\G\Seq\D$ is interpreted as a formula of $\lan$
via the \emph{formula interpretation} $\fint$
as $\bigand\G \imp \bigor\D$
if $\G$ is non-empty, and 
as $\bigor\D$ if $\G$ is empty, where
$\bigor\emptyset$ is interpreted as $\bot$.
A \emph{sequent calculus} $\SC$ is a set of
initial sequents and sequent rules.%
\footnote{More precisely, as for axiomatic systems, we rather consider sequent and rule schemata, omitting this specification throughout the text.}
A \emph{derivation} of a sequent $\varseq$ 
in a calculus $\SC$
 is a tree where each node is labelled by a sequent, the root is labelled by $\varseq$, 
 the leaves are labelled by initial sequents 
 and each node is obtained by the immediate predecessor(s) by the application of a rule of 
 $\SC$. 
 A sequent $\varseq$ is \emph{derivable} 
in a calculus $\SC$ if there is a derivation of $\varseq$ in $\SC$.
A formula $A$ is derivable in $\SC$ if the sequent $\seq A$ 
is derivable in $\SC$.
A sequent calculus $\SC$ is a calculus for a logic $\logic$ if
for every formula $A$, $A$ is derivable in $\SC$ if and only if it is derivable in $\logic$. 
\end{definition}

%
%
%
%

In order to analyse the sequent calculi, we also consider the following standard concepts about the sequent rules. 

\begin{definition}\label{def:sequent}
A rule $R$ is \emph{admissible} in a calculus $\SC$ if
whenever the premisses of $R$ are derivable in $\SC$,
the conclusion is also derivable in $\SC$.
A formula is \emph{principal} in the application of a rule if it occurs in the conclusion and not in the premiss(es),
while it is \emph{active} if it occurs in (at least) one premiss and not in the conclusion.
Contraction rules are an exception to these definitions:
the principal and active formula is the one formula $A$
which has $n$ occurrences in the conclusion and $n + 1$ occurrences in the premiss.
All formulas which are neither principal nor active are the \emph{context}.
\end{definition}

\begin{figure}[t]
\begin{small}
\textbf{Sequent calculus $\Gone\CPL$ for $\CPL$.}

\vspace{0.2cm}
\initcl\ \ $A \Seq A$
\hfill
\ax{$\G, A_i \Seq \D$}
\llab{\llandcl}
\rlab{($i = 1, 2$)}
\uinf{$\G, A_1\land A_2 \Seq \D$}
\disp
\hfill
\ax{$\G \Seq A, \D$}
\ax{$\G \Seq B, \D$}
\llab{\rlandcl}
\binf{$\G \Seq A\land B, \D$}
\disp

\vspace{0.2cm} 
\lbotcl\ \ $\bot \Seq$
\hfill
\ax{$\G, A \Seq \D$}
\ax{$\G, B \Seq \D$}
\llab{\llorcl}
\binf{$\G, A \lor B \Seq \D$}
\disp
\hfill
\ax{$\G \Seq A_i, \D$}
\llab{\rlorcl}
\rlab{($i = 1, 2$)}
\uinf{$\G \Seq A_1\lor A_2, \D$}
\disp

\vspace{0.2cm} 
\ax{$\G \Seq A, \D$}
\ax{$\G, B \Seq \D$}
\llab{\limpcl}
\binf{$\G, A \imp B \Seq \D$}
\disp
\hfill
\ax{$\G, A\Seq B, \D$}
\llab{\rimpcl}
\uinf{$\G \Seq A\imp B, \D$}
\disp
\hfill
\ax{$\G \Seq \D$}
\llab{\lwkcl}
\uinf{$\G, A \Seq \D$}
\disp

\vspace{0.2cm} 
\ax{$\G \Seq \D$}
\llab{\rwkcl}
\uinf{$\G \Seq A, \D$}
\disp
\hfill
\ax{$\G, A, A \Seq \D$}
\llab{\lctrcl}
\uinf{$\G, A \Seq \D$}
\disp
\hfill
\ax{$\G \Seq A, A, \D$}
\llab{\rctrcl}
\uinf{$\G \Seq A, \D$}
\disp

\vspace{0.3cm}
\textbf{Sequent calculus $\Gone\MPL$ for $\MPL$.}

\vspace{0.2cm}
\minit\ \
$A \Seq A$
\hfill
\ax{$\G, A_i \Seq C$}
\llab{\mlland}
\rlab{($i = 1, 2$)}
\uinf{$\G, A_1\land A_2 \Seq C$}
\disp
\hfill
\ax{$\G \Seq A$}
\ax{$\G \Seq B$}
\llab{\mrland}
\binf{$\G \Seq A\land B$}
\disp

\vspace{0.2cm} 
\ax{$\G, A \Seq C$}
\ax{$\G, B \Seq C$}
\llab{\mllor}
\binf{$\G, A \lor B \Seq C$}
\disp
\hfill
\ax{$\G \Seq A_i$}
\llab{\mrlor}
\rlab{($i = 1, 2$)}
\uinf{$\G \Seq A_1\lor A_2$}
\disp
\hfill
\ax{$\G, A\Seq B$}
\llab{\mrimp}
\uinf{$\G \Seq A\imp B$}
\disp

\vspace{0.2cm} 
\ax{$\G \Seq A$}
\ax{$\G, B \Seq C$}
\llab{\mlimp}
\binf{$\G, A \imp B \Seq C$}
\disp
\hfill
\ax{$\G \Seq C$}
\llab{\mlwk}
\uinf{$\G, A \Seq C$}
\disp
\hfill
\ax{$\G, A, A \Seq C$}
\llab{\mlctr}
\uinf{$\G, A \Seq C$}
\disp
\end{small}
\caption{\label{fig:prop seq rules}Sequent calculi $\Gone\CPL$ and $\Gone\MPL$.}
\end{figure}

The well-known G1-sequent calculi $\Gone\CPL$ and $\Gone\MPL$ \cite{Troelstra:2000} for $\CPL$ and $\MPL$ are displayed in Figure~\ref{fig:prop seq rules}.
It is easy to see that $\Gone\MPL$ corresponds to the single-succedent restriction of $\Gone\CPL$.
In particular, the initial sequent \lbotcl\ and the rule \rctrcl\
are dropped in $\Gone\MPL$ as they have respectively no formula in the succedent, and two occurrences of the active formula in the succedent of the premiss.
\rwkcl\ is also dropped as it requires 
either no formula in the succedent of the premiss or at least two formulas in the succedent of the conclusion.
Concerning the other rules,
the right context is removed from the sequents with an active or principal formula in the succedent, this is the case for instance of initial sequents \minit\ and of the rule \mrland, as well as of the left premiss of the rule \mlimp.
In the other rules, the right context is converted from an arbitrary multiset $\D$ to a single formula $C$.

In order to extend the multi- vs. single-succedent relation to modal logics,
we consider the G1-sequent calculus $\Gone\K$ for $\K$,
defined extending $\Gone\CPL$ with the  modal rules
\ruleKboxcl\ and \ruleKdiamcl\ in Figure~\ref{fig:K seq rules}.
In these rules and the following, 
given a multiset $\G = A_1, ..., A_n$,
 we denote $\Box\G$ and $\diam\G$ the multisets
$\Box A_1, ..., \Box A_n$ and $\diam A_1, ..., \diam A_n$, respectively. 
As for axiomatic systems, sequent calculi for $\K$ are more commonly defined in terms of $\Box$ only.
Here we consider a formulation of the calculus with both $\Box$ and $\diam$ explicit in order to better display the relation with minimal modal logics, where $\Box$ and $\diam$ are not interdefinable.
The rules \ruleKboxcl\ and \ruleKdiamcl\ 
for $\K$ with explicit $\Box$ and $\diam$ can be found e.g. in \cite{indr:2021book}.

\begin{figure}[t]
\begin{small}
\centering
\ax{$\Sigma \Seq A, \Pi$}
\llab{\ruleKboxcl}
\uinf{$\Box\Sigma \Seq \Box A,  \diam\Pi$}
\disp
\hfill
\ax{$\Sigma, A \Seq \Pi$}
\llab{\ruleKdiamcl}
\uinf{$\Box\Sigma, \diam A \Seq \diam\Pi$}
\disp
\hfill
\ax{$\Sigma \Seq A$}
\llab{\rulemKbox}
\uinf{$\Box\Sigma \Seq \Box A$}
\disp
\hfill
\ax{$\Sigma, A \Seq B$}
\llab{\rulemKdiam}
\uinf{$\Box\Sigma, \diam A \Seq \diam B$}
\disp
\end{small}
\caption{\label{fig:K seq rules}Modal rules for $\Gone\K$ and $\Gone\MK$.} 
\end{figure}

On the basis of $\Gone\K$,
we now define the calculus $\Gone\MK$ as the single-succedent
restriction of $\Gone\K$.
As a result, 
$\Gone\MK$ contains the rules of $\Gone\MPL$
and the modal rules \rulemKbox\ and \rulemKdiam.
Indeed, in the rule \ruleKboxcl, the succedent of the conclusion must have a $\Box$-formula $\Box A$ and can have additional $\diam$-formulas. Then, its single-succedent restriction only preserves $\Box A$.
Concerning \ruleKdiamcl/\rulemKdiam, the consequent of the conclusion of \ruleKdiamcl\ has an arbitrary number of $\diam$-formulas. 
Correspondingly, the consequent of the conclusion of \rulemKdiam\ has exactly one $\diam$-formula.

In the remaining part 
of this section,
we show that $\Gone\MK$ is equivalent to the logic $\MK$ defined in the previous section.
The proof is based on the following theorem, 
which entails that the addition of the cut rule to $\Gone\MK$ does not extend the set of derivable sequents.
To do its length, the proof of Theorem~\ref{th:cut MK}
is presented in the appendix.

\begin{restatable}{theorem}{ThCutMK}\label{th:cut MK}
The following rule \cut\ is admissible in $\Gone\MK$:
\begin{center}
\ax{$\G \seq A$}
\ax{$\Sigma, A \seq C$}
\llab{\mcut}
\binf{$\G, \Sigma \seq C$}
\disp
\end{center}
\end{restatable}
\begin{proof}
The proof is in the appendix.
\end{proof}

\begin{theorem}
For all $A\in\lan$, 
$A$ is derivable in $\Gone\MK$ if and only if 
$A$ is derivable in $\MK$.
\end{theorem}
\begin{proof}
($\Rightarrow$) 
$\fint(A \seq A) = A \imp A$ is derivable in $\MK$, 
moreover we can show that
for all rules $\varseq_1, ..., \varseq_n/\varseq$ of $\Gone\MK$,
the rule $\fint(\varseq_1), ..., \fint(\varseq_n)/\fint(\varseq)$
is derivable in $\MK$,
where $\fint$ is the formula interpretation of sequents as defined in Definition~\ref{def:sequent}.
For the propositional rules the proof is standard.
We show in Figure~\ref{fig:der MK}  the derivations of the modal rules,
considering the representative cases where $\Sigma = C_1, C_2$.
The cases where $\Sigma$ contains less or more formulas are a simplification or a generalisation of these cases.

($\Leftarrow$)
The proof consists in showing that all axioms and rules of $\MK$ are derivable, respectively admissible in $\Gone\MK$.
We omit the derivations of the propositional axioms which are standard. The derivations 
of the modal axioms and rule are displayed in Figure~\ref{fig:der G1MK}.
\end{proof}

\begin{figure}[t]
\begin{small}
\begin{multicols}{2}
1.  $A \imp (B \imp A \land B)$ \hfill ($\MPL$) \quad \

2. $\Box(A \imp (B \imp A \land B))$ \hfill (1, \nec) \quad \

3. $\Box A \imp \Box(B \imp A \land B)$ \hfill (2, \axKbox) \quad \

4. $\Box A \imp (\Box B \imp \Box(A \land B))$ \hfill (3, \axKbox) \quad \

5. $\Box A \land \Box B \imp \Box(A \land B)$ \hfill (4, $\MPL$) \quad \

\columnbreak

1. $C_1 \land C_2 \imp A$ \hfill (premiss of \rulemKbox) \

2.  $\Box(C_1 \land C_2 \imp A)$ \hfill (1, \nec) \

3.  $\Box (C_1 \land C_2) \imp \Box A$ \hfill (2, \axKbox, mp) \

4.  $\Box C_1 \land \Box C_2 \imp \Box (C_1 \land C_2)$ \hfill (derivable) \

5.  $\Box C_1 \land \Box C_2 \imp \Box A$ \hfill (3, 4) \
\end{multicols}


\begin{center}
\begin{tabular}{lr}
1. $C_1 \land C_2 \land A \imp B$ & (premiss of \rulemKdiam) \\
2. $C_1 \land C_2 \imp(A \imp B)$ & (1, $\MPL$) \\
3.  $\Box(C_1 \land C_2 \imp (A \imp B))$ & (2, \nec) \\
4.  $\Box (C_1 \land C_2) \imp \Box (A \imp B)$ & (3, \axKbox, mp) \\
5.  $\Box C_1 \land \Box C_2 \imp \Box (C_1 \land C_2)$ & (derivable) \\
6.  $\Box C_1 \land \Box C_2 \imp \Box (A \imp B)$ & (4, 5) \\
7.  $\Box C_1 \land \Box C_2 \imp (\diam A \imp \diam B)$ & (6, \axKdiam) \\
8.  $\Box C_1 \land \Box C_2 \land \diam A \imp \diam B$ & (7, $\MPL$) \\
\end{tabular}
\end{center}
\end{small}
\caption{\label{fig:der MK}Derivations in $\MK$.}
\end{figure}

\begin{figure}[t]
\begin{small}
\ax{$(\ast) \ A \imp B, A \seq B$}
\rlab{\rulemKbox}
\uinf{$\Box(A \imp B), \Box A \seq \Box B$}
\rlab{\mrimp}
\uinf{$\Box(A \imp B) \seq \Box A \imp \Box B$}
\rlab{\mrimp}
\uinf{$\seq \Box(A \imp B) \imp (\Box A \imp \Box B)$}
\disp
\hfill
\ax{$(\ast) \ A \imp B, A \seq B$}
\rlab{\rulemKdiam}
\uinf{$\Box(A \imp B), \diam A \seq \diam B$}
\rlab{\mrimp}
\uinf{$\Box(A \imp B) \seq \diam A \imp \diam B$}
\rlab{\mrimp}
\uinf{$\seq \Box(A \imp B) \imp (\diam A \imp \diam B)$}
\disp
\ \
\ax{$\seq A$}
\rlab{\rulemKbox}
\uinf{$\seq \Box A$}
\disp

\vspace{0.2cm}
\ax{$\seq A$}
\ax{$\seq A \imp B$}
\ax{$(\ast) \ A \imp B, A \seq B$}
\rlab{\cut}
\binf{$A \seq B$}
\rlab{\cut}
\binf{$\seq B$}
\disp
\hfill
\ax{$A \seq A$}
\ax{$B \seq B$}
\rlab{\mlwk}
\uinf{$B, A \seq B$}
\rlab{\mlimp}
\binf{$(\ast) \ A \imp B, A \seq B$}
\disp
\end{small}
\caption{\label{fig:der G1MK}Derivations in $\Gone\MK$.}
\end{figure}

\section{Relating minimal K and constructive K}\label{sec:rel MK CK}

The minimal modal logic $\MK$ just defined is strictly related to the constructive modal logic $\CK$
studied in the literature.
In particular, from an axiomatical perspective, $\CK$ coincides with the extension of $\MK$ with ex falso quodlibet $\bot\imp A$, exactly as $\IPL$ amounts to $\MPL$ + $\bot\imp A$.
Except for the different propositional base, $\MK$ and $\CK$ share the same modal principles.
In this section, we show that $\CK$ is also  strictly related to $\MK$
both from a semantical and from a proof theoretical point of view.

%
%
%
%

\subsection{Semantics}\label{subsec:rel MK CK sem}

As recalled in Section \ref{subsec:sem prel},
disregarding the modalities, there are two 
ways to transform relational models for $\MPL$ into relational models for $\IPL$:
(1) 
assuming $\FW = \emptyset$, 
thus obtaining Kripke's intuitionistic relational models,
or
(2) preserving the 
fallible worlds but
ensuring the validity of 
ex falso quodlibet 
by
assuming $\FW\subseteq\V(p)$ for all $p\in\atm$.
Interestingly, the two ways are equivalent for propositional logic,
as they both provide a semantics for $\IPL$,
but they are not equivalent in presence of the modalities.
In particular,
if applied to minimal birelational models,
the restriction (1) gives relational models for $\WK$ as defined in \cite{wij}.
By contrast, a suitable adaptation of (2) which esures the validity of $\bot\imp A$ also
 in presence of the modalities
gives birelational models for $\CK$.

\begin{definition}[Constructive birelational semantics]\label{def:const birel model}
A minimal birelational model $\M = \langle \W, \less, \FW, \R, \V\rangle$ is a 
\emph{constructive birelational model} if for all $w \in \FW$ it holds:
\begin{itemize}
\item[(i)] $w \in\V(p)$ for all $p \in \atm$;
\item[(ii)] if $w \R v$, then $v\in\FW$;
\item[(iii)] there is $v$ such that $w \R v$.
\end{itemize}
\end{definition}

Analogous constructive birelational models for $\CK$ were defined in \cite{Mendler1}.
The models in Definition~\ref{def:const birel model} are slightly different because of the latter condition (iii)
which is not considered in \cite{Mendler1}.
We observe however that this (or a similar) condition
is necessary in order to ensure the validity of ex falso quodlibet over the whole language $\lan$.
To see this, consider a model $\M$ 
satisfying (i) and (ii) but not (iii), where $\W = \FW = \{w\}$, $w \less w$ and not $w \R w$.
It is easy to verify that $w \Vd \bot$ but $w \not\Vd \diam p$, and hence $\M\not\models \bot\imp\diam p$.

We prove that $\CK$ is sound and complete
with respect to constructive birelational models.

\begin{theorem}\label{th:sound CK}
For all $A\in\lan$,
if $A$ is derivable in $\CK$, then $A$ is valid in every constructive birelational model $\M$.
In particular, $\M\models\bot\imp B$ for every $B\in\lan$.
\end{theorem}
\begin{proof}
The proof extends the proof of Theorem~\ref{th:soundness MK} by showing that $\M\models\bot\imp A$
for every $A$.
Suppose that $w \Vd \bot$. Then $w \in \FW$. 
We show by induction on the construction of $A$ that $w \Vd A$.
($A = p$) By Definition~\ref{def:const birel model}, item (i), $w\in\V(p)$, then $w \Vd p$.
($A = \bot$) By hypothesis.
($A = B \land C, B \lor C$) Immediate by applying the \ih.
($A = B \imp C$) Immediate by \ih\ and $\less$-upward closure of $\FW$.
($A = \Box B$)
Suppose $w \less v$. Since $\FW$ is $\less$-upward closed, $v \in\FW$.
Then by Definition~\ref{def:const birel model}, item (ii),
for all $u$ such that $v \R u$, $u\in\FW$.
Then by \ih, $u \Vd B$, therefore $w \Vd \Box B$.
($A = \diam B$)
Suppose $w \less v$. Since $\FW$ is $\less$-upward closed, $v \in\FW$.
Then by Definition~\ref{def:const birel model}, item (iii),
there is $u$ such that $v \R u$, and by item (ii), $u\in\FW$.
Then by \ih, $u \Vd B$, therefore $w \Vd \diam B$.
\end{proof}


We now prove that $\CK$ is complete
with respect to constructive birelational models.
First, note that Lemmas~\ref{lemma:lind rel} and \ref{lemma:rel seg} also hold for $\CK$
(in particular, for Lemma~\ref{lemma:rel seg} the proof is the same,
uniformly replacing $\MK$ with $\CK$).
We additionally prove the following lemma.

\begin{lemma}
Let $\Mc = \langle \Wc, \lessc, \FWc, \Rc, \Vc \rangle$ be the canonical birelational model for $\CK$
(Definition~\ref{def:can model rel}). 
Then for all $(\Phi, \U)\in\Wc$ and all $A \in\lan$, 
$(\Phi, \U)\Vd A$ if and only if $A\in\Phi$.
Moreover, $\Mc$
 is a constructive birelational model.
\end{lemma}
\begin{proof}
The first claim is proved exactly as Lemma~\ref{lemma:can model rel}.
For the second claim, we show that $\M$ satisfies the conditions of Definition~\ref{def:const birel model}.
Suppose that $(\Phi, \U)\in\FWc$. Then $\bot\in\Phi$.
Since $\Phi$ is closed under derivation, by ex falso quodlibet we obtain $\Phi = \lan$,
which entails the following.
(i) For all $p\in\atm$, $p\in\Phi$, hence by definition, $(\Phi, \U)\in\Vc(p)$.
(ii) $\Box\bot\in\Phi$, hence by Definition~\ref{def:rel seg}, $\bot\in\Psi$ for all $\Psi\in\U$.
Then $(\Phi,\U)\Rc(\Psi,\VV)$ entails $\bot\in\Psi$, thus $(\Psi, \VV)\in\FWc$.
(iii) $\diam\bot\in\Phi$, hence by Definition~\ref{def:rel seg}, there is $\Psi\in\U$ such that $\bot\in\Psi$.
By Lemma~\ref{lemma:rel seg} 
(which holds for $\CK$ as well),
there exists a $\CK$-segment $(\Psi,\VV)$.
Then $(\Phi,\U)\Rc(\Psi,\VV)$ and $(\Psi, \VV)\in\FWc$.
\end{proof}

As a consequence of the lemma, we obtain the completeness of $\CK$ 
(cf.~proof of Theorem~\ref{th:compl MK}).

\begin{theorem}
\label{th:compl CK}
For all $A\in\lan$,
$A$ is derivable in $\CK$ if and only if $A$ is valid in every constructive birelational model.
\end{theorem}
%
%
%

\subsection{Sequent calculus}

We have considered in Section~\ref{sec:seq MK} the multi- vs.~single-succedent correspondence between the sequent calculi $\Gone\CPL$ and $\Gone\MPL$.
A similar relation holds between sequent calculi for $\CPL$ and $\IPL$.
In particular, the calculus $\Gone\IPL$ for $\IPL$ can be defined by restricting $\Gone\CPL$
to sequents with \emph{at most one} formula in the succedent (cf.~\cite{Troelstra:2000}).
The resulting calculus is displayed in Figure~\ref{fig:G1IPL}, where $0 \leq |\delta| \leq 1$.
If we apply the same restriction to $\Gone\K$, we obtain the calculus
\[\Gone\WK =  \Gone\IPL + \text{\rulemKbox} + \text{\rulemKdiam} + \ax{$\Sigma, A \seq$}\uinf{$\Box\Sigma, \diam A \seq$}\disp\]
for $\WK$ defined in \cite{wij}.
By contrast, in order to obtain a calculus for $\CK$,
we need to extend $\Gone\IPL$ with the minimal modal rules only:
 \[\Gone\CK =  \Gone\IPL + \text{\rulemKbox} + \text{\rulemKdiam}.\]
 In this way we obtain the sequent calculus for $\CK$ defined 
 and proved to be cut-free in \cite{Bellin}.

\begin{figure}[t]
\begin{small}
\iinit\ \
$A \Seq A$
\hfill
\ilbot\ \ $\bot\seq$
\hfill
\ax{$\G, A_i \Seq \delta$}
\llab{\illand}
\rlab{($i = 1, 2$)}
\uinf{$\G, A_1\land A_2 \Seq \delta$}
\disp
\hfill
\ax{$\G \Seq A$}
\ax{$\G \Seq B$}
\llab{\irland}
\binf{$\G \Seq A\land B$}
\disp

\vspace{0.2cm} 
\ax{$\G, A \Seq \delta$}
\ax{$\G, B \Seq \delta$}
\llab{\illor}
\binf{$\G, A \lor B \Seq \delta$}
\disp
\hfill
\ax{$\G \Seq A_i$}
\llab{\irlor}
\rlab{($i = 1, 2$)}
\uinf{$\G \Seq A_1\lor A_2$}
\disp
\hfill
\ax{$\G, A\Seq B$}
\llab{\irimp}
\uinf{$\G \Seq A\imp B$}
\disp

\vspace{0.2cm} 
\ax{$\G \Seq A$}
\ax{$\G, B \Seq \delta$}
\llab{\ilimp}
\binf{$\G, A \imp B \Seq \delta
$}
\disp
\hfill
\ax{$\G \Seq \delta$}
\llab{\ilwk}
\uinf{$\G, A \Seq \delta$}
\disp
\hfill
\ax{$\G \Seq$}
\llab{\irwk}
\uinf{$\G \Seq A$}
\disp
\hfill
\ax{$\G, A, A \Seq \delta$}
\llab{\ilctr}
\uinf{$\G, A \Seq \delta$}
\disp
\end{small}
\caption{\label{fig:G1IPL}Sequent calculus $\Gone\IPL$.}
\end{figure}

%
%
%
%
%
%
%
%
%

\section{A framework of minimal and constructive modal logics}\label{sec:family}

We have seen that the two considered methods,
respectively based on bimodal companion and sequent calculus restriction, define the same minimal counterpart of $\K$.
In 
section, we show that the equivalence of the two methods is not a peculiarity of $\K$ only:
We apply the two methods to a family of 14 standard classical modal logics, 
and show that they are equivalent for all of them,
thus obtaining a minimal counterpart for each of these 
systems.

In order to apply our sequent-based approach,
we consider a family of classical modal logics enjoying standard
 cut-free Gentzen calculi 
(this restriction excludes well-known modal logics
for which such calculi are not available, such as $\Sfive$).
We also require the logics to have a uniform semantic characterisation,
we consider to this purpose a neighbourhood semantics that uniformly covers all considered systems,
that include both normal and non-normal modal logics.

%
%
%
%
%
%
%
%

\begin{figure}[t]
\centering
\begin{small}
\begin{tabular}{llllllcllll}
\multirow{2}{*}{\ax{$A \imp B$}
\llab{\monbox \ }
\uinf{$\Box A \imp \Box B$}
\disp} &
\axdual &  $\Box A \tto \neg\diam\neg A$ 
&  \axNbox & $\Box\top$ \\   

&
\axKbox &  $\Box(A \imp B) \imp (\Box A \imp \Box B)$ 
& \axNdiam & $\neg\diam\bot$ \\

\multirow{2}{*}{\ax{$A \imp B$}
\llab{\mondiam \ }
\uinf{$\diam A \imp \diam B$}
\disp} &
\axKdiam &  $\Box(A \imp B) \imp (\diam A \imp \diam B)$ \
&  \axTbox & $\Box A \imp A$ \\

& \axCbox &  $\Box A\land \Box B \imp \Box(A\land B)$
& \axTdiam & $A \imp \diam A$ \\

\multirow{2}{*}{\ax{$A$}
\llab{\nec \ }
\uinf{$\Box A$}
\disp} & 
 \axCdiam &  $\diam(A \lor B) \imp \diam A \lor \diam B$
&  \axPbox & $\neg\Box\bot$ \\

& \axD & $\Box A \imp \diam A$ &  \axPdiam & $\diam\top$ \\
\end{tabular}
\end{small}
\caption{\label{fig:axioms}Modal axioms and rules.}
\end{figure}

Specifically, we consider 14 
 classical modal logics that are
  axiomatically
defined in the language
$\lan$ extending 
$\CPL$, formulated in $\lan$,
with the following modal axioms and rules from Figure~\ref{fig:axioms}:
%
%
%
%
%
%
%
%
%
%
%
\begin{center}
\begin{tabular}{lllllllllllll}
$\EM$  := \axdual, \monbox &&
  $\EMD$   :=  $\EM$  + \axD &&
$\EMT$  :=   $\EM$  + \axTbox \\

$\EMN$  :=   $\EM$ + \axNbox &&  
$\EMND$  :=   $\EMN$  + \axD &&
$\EMNT$   :=   $\EMN$  + \axTbox \\  

$\EMC$  :=  $\EM$ + \axCbox &&
$\EMCD$   :=  $\EMC$ + \axD &&
$\EMCT$  :=   $\EMC$  + \axTbox \\ 

\vspace{0.2cm}
$\K$   :=   $\EM$ + \axNbox, \axCbox &&
$\KD$   := $\K$ + \axD &&
$\KT$   :=  $\K$  + \axTbox \\

  $\EMP$  :=  $\EM$ + \axPbox \\
  $\EMNP$   :=   $\EMN$ + \axPbox \\
\end{tabular}
\end{center}
Each classical modal logic is denoted $\MSigma$, where 
$\mathsf\Sigma \subseteq\{\mathsf{C,N,P,D,T}\}$ corresponds to the list of 
 axioms
 among \axCbox, \axNbox, \axPbox, \axD, \axTbox\
  extending $\EM$.
The only exceptions to this notation are $\K$, $\KD$ and $\KT$ for which we adopt the standard names.
Note
however  that $\K$ amounts to $\logicnamestyle{MCN}$, 
this  axiomatisation of $\K$ is equivalent to the more standard one with \nec\ and \axKbox\
considered in Section~\ref{subsec:syn prel} (cf.~e.g.~\cite{Chellas:1980}).
As usual, given the duality between $\Box$ and $\diam$ in classical logics, 
the above systems can be equivalently defined by
replacing \monbox, \axNbox, \axCbox, \axPbox, and \axTbox, with their 
$\diam$-versions \mondiam, \axNdiam, \axCdiam, \axPdiam, and \axTdiam\ 
 (Figure~\ref{fig:axioms}).%
\footnote{A $\Box$- and a $\diam$-formulation of the axiom \axD\ could be also considered, namely $\neg(\Box A \land \Box \neg A)$ and $\diam A \lor \diam \neg A$. 
We prefer to consider 
the more standard 
version 
$\Box A \imp \diam A$, which is adequate for both 
formulations
of the logics
and is commonly adopted in the definition of \intu\ modal logics.}
The systems $\EMCP$ and $\KP$ are not listed above 
as they are respectively
equivalent to $\EMCD$ and $\KD$
(\axPbox\ and \axD\ are interderivable given \monbox\ and \axCbox).
The resulting classical modal logics and their inclusion relations are displayed in Figure~\ref{fig:cl dyag}.
In the following, we use $\logic$ or $\MSigma$, without specifying the set $\mathsf\Sigma$, to denote any of the above classical logics.

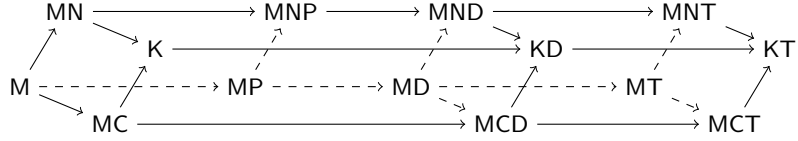
\begin{figure}[t]
\centering
\begin{small}
\begin{tikzpicture}
    \node (M) at  (0,0)  {$\EM$};
    \node (MN) at (0.6, 1) {$\EMN$};
    \node  (MC) at (1.2, -0.5) {$\EMC$};
    \node (K) at (1.8, 0.5) {$\K$};

    \node (MP) at  (3.,0)  {$\EMP$};
    \node (MNP) at (3.6, 1) {$\EMNP$};

    \node (MD) at  (5.2,0)  {$\EMD$};
    \node (MND) at (5.8, 1) {$\EMND$};
    \node  (MCD) at (6.4, -0.5) {$\EMCD$};
    \node (KD) at (7, 0.5) {$\KD$};

    \node (MT) at  (8.3,0)  {$\EMT$};
    \node (MNT) at (8.9, 1) {$\EMNT$};
    \node  (MCT) at (9.5, -0.5) {$\EMCT$};
    \node (KT) at (10.1, 0.5) {$\KT$};

	\draw[->] (M) -- (MN);
	\draw[->] (M) -- (MC);
	\draw[->] (MN) -- (K);	
	\draw[->] (MC) -- (K);
	\draw[->, dashed] (MP) -- (MNP);
	\draw[->, dashed] (MD) -- (MND);
	\draw[->, dashed] (MD) -- (MCD);
	\draw[->] (MND) -- (KD);	
	\draw[->] (MCD) -- (KD);
	\draw[->, dashed] (MT) -- (MNT);
	\draw[->, dashed] (MT) -- (MCT);
	\draw[->] (MNT) -- (KT);	
	\draw[->] (MCT) -- (KT);

	\draw[->] (MN) -- (MNP);
	\draw[->] (MNP) -- (MND);
	\draw[->] (MND) -- (MNT);
	\draw[->] (MC) -- (MCD);
	\draw[->] (MCD) -- (MCT);
	\draw[->, dashed] (MP) -- (MD);
	\draw[->, dashed] (M) -- (MP);
	\draw[->, dashed] (MD) -- (MT);
	\draw[->] (KD) -- (KT);
	\draw[->] (K) -- (KD);
\end{tikzpicture}
\end{small}
\caption{\label{fig:cl dyag} Dyagram of classical modal logics.}
\end{figure}

In the following subsections, we show that the two methods define, 
for each classical modal logic $\logic$, the following minimal counterpart $\MLvar$.

\begin{definition}[Minimal modal logics]
\emph{Minimal modal logics} are axiomatically defined
in the language $\lan$  extending $\MPL$
with the following modal axioms and rules from Figure~\ref{fig:axioms}:
%
\begin{center}
\begin{tabular}{llllllllllllllll}
$\MM :=$ \monbox, \mondiam &&  $\MMP := \MM\ +$ \axPdiam \\
$\MMN := \MM\ +$ \axNbox && $\MMNP := \MMN\ +$ \axPdiam \\
$\MMC := \MM\ +$ \axCbox, \axKdiam \\
\vspace{0.2cm}
$\MK := \MMC\ +$ \axNbox &&   \\

$\MMD := \MM\ +$ \axD, \axPdiam && $\MMT := \MM\ +$ \axTbox, \axTdiam \\  
$\MMND := \MMN\ +$ \axD && $\MMNT := \MMN\ +$ \axTbox, \axTdiam \\
$\MMCD := \MMC\ +$ \axD, \axPdiam && $\MMCT := \MMC\ +$ \axTbox, \axTdiam \\ 
$\MKD := \MK\ +$ \axD && $\MKT := \MK\ +$ \axTbox, \axTdiam \\
\end{tabular}
\end{center}
\end{definition}



\subsection{Minimal modal logics via bimodal companions}

We prove that each minimal logic $\MLvar$ above is the minimal counterpart of the classical logic $\logic$
as defined in Definition~\ref{def:minimal counterpart}
(that is, $\MLvar\vd A$ if and only if $\Sfour\oplus\logic\vd A^{\tr}$).
As before, in order to prove this result, we first provide a semantics for minimal modal logics.

We start recalling the neighbourhood semantics for classical modal logics (cf. \cite{Chellas:1980,Pacuit:2017}).
A \emph{classical \neigh\ model} 
  is a tuple $\M = \langle \W, \N, \V\rangle$, 
where 
$\W$ is a non-empty set of worlds, 
$\V : \atm \longrightarrow \pow(\W)$ is a  valuation function for propositional variables, and
$\N$ is a function $\W\longrightarrow \pow(\pow(\W))$, called  neighbourhood function.
Modal formulas are interpreted in classical neighbourhood models as
$w\Vd\Box B$ iff  there is $\alpha\in\N(w)$ such that 
for all $v \in \alpha$, $v \Vd B$;
and
$w\Vd\diam B$ iff  for all $\alpha\in\N(w)$, 
there is $v\in\alpha$ such that $v\Vd B$.
Each classical modal logic $\logic$ considered in this work is characterised by the class of all classical neighbourhood models satisfying 
the following condition (\cC), (\cN), (\cP), (\cD), or (\cT),  for all $\alpha, \beta\subseteq\W$,
if $\logic$ contains the axiom \axCbox, \axNbox, \axPbox, \axD, or \axTbox,
respectively:
\begin{center}
\begin{tabular}{llllrlll}
(\cC) \ If $\alpha,\beta\in\N(w)$, then $\alpha\cap\beta\in\N(w)$. && 
(\cN) \ $\N(w)\not=\emptyset$. \\
(\cD) \  If $\alpha,\beta\in\N(w)$, then $\alpha\cap\beta\not=\emptyset$. &&
(\cP) \  $\emptyset\notin\N(w)$. \\
(\cT) \ If $\alpha\in\N(w)$, then $w\in\alpha$. \\
\end{tabular}
\end{center}

We also remark that for each considered classical modal logic $\logic$,
the fusion $\Sfour \oplus \logic$ is characterised by the class of models $ \langle \W, \R, \N, \V\rangle$,
where $\R$ is a reflexive and transitive binary relation on $\W$, and $\N$ is a neighbourhood function satisfying the conditions among (\cC), (\cN), (\cD), (\cP), (\cT) satisfied by the models for $\logic$.
This 
characterisation of fusions $\Sfour \oplus \logic$
can be easily proved by combining the completeness proofs by canonical models for $\Sfour$ and for $\logic$
(see e.g.~\cite{Chellas:1980}).

By combining relational models for $\MPL$ and classical neighbourhood models, we now define minimal neighbourhood models for minimal modal logics as follows.

\begin{definition}[Minimal neighbourhood semantics]\label{def:neigh semantics}
A \emph{minimal \neigh\ model} 
  is a tuple $\M = \langle \W, \less, \FW, \N, \V\rangle$, 
where 
$\langle \W, \less, \FW, \V\rangle$ is a minimal relational model, and
$\N$ is 
neighbourhood
a function $\W\longrightarrow \pow(\pow(\W))$. 
The forcing relation $\M, w\Vd A$ is 
inductively defined extending the clauses
for $p, \bot, \land, \lor, \imp$
in Section~\ref{subsec:sem prel}
with the following clauses for the modalities:
\begin{center}
\begin{tabular}{lllll}
$\M, w\Vd\Box B$ & iff &  for all $v \more w$, there is $\alpha\in\N(v)$ such that $\alpha\ufor B$; \\
$\M, w\Vd\diam B$ & iff &  for all $v \more w$, for all $\alpha\in\N(v)$, $\alpha\efor B$; \\ 
\end{tabular}
\end{center}

\noindent
where $\alpha\ufor B$ and $\alpha\efor B$ 
are abbreviations for, respectively, `for all $u\in\alpha$, $\M, u\Vd B$',
and `there is $u\in\alpha$ such that $\M, u\Vd B$'.

For each minimal modal logic $\MMSigma$, we say that a minimal neighbourhood model $\M$ is a model for $\MMSigma$
(or it is a $\MMSigma$-model) if it satisfies the condition \emph{(\cX)}
above
for all $\mathsf{X} \in \mathsf{\Sigma}$.
Note that $\MK$ amounts to $\MMCN$, hence the corresponding models must satisfy both \emph{(\cC)} and \emph{(\cN)}.
\end{definition}

By an easy induction on the construction of formulas one can prove the following.

\begin{proposition}
[Hereditary property]
For every $A\in\lan$, every minimal neighbourhood model $\M$, and every world $w$ of $\M$,
if $w\Vd A$ and $w\less v$, then $v\Vd A$.
\end{proposition}

Now we prove that the logics $\MLvar$ are sound and complete with respect to the corresponding classes of models.

\begin{theorem}\label{th:soundness ML}
For all $A\in\lan$ and all minimal modal logic $\MLvar$,
if $A$ is derivable in $\MLvar$, then $A$ is valid in all minimal neighbourhood models for $\MLvar$.
\end{theorem}
\begin{proof}
We show that all modal axioms and rules of $\MLvar$ are valid, respectively validity preserving, in every minimal neighbourhood model $\M$ for $\MLvar$.
\begin{enumerate}[leftmargin=*, align=left]
\item[(\monbox)] Suppose that $\M\models A \imp B$ and $w \Vd \Box A$.
Then for all $v \more w$, there is $\alpha\in\N(v)$ such that for all $z\in\alpha$, $z\Vd A$,
thus $z\Vd B$, hence $w \Vd \Box B$.
Therefore $\M\models \Box A \imp\Box B$.

\item[(\mondiam)] Suppose that $\M\models A \imp B$ and $w \Vd \diam A$.
Then for all $v \more w$, for all $\alpha\in\N(v)$, there is $z\in\alpha$ such that $z\Vd A$,
thus $z\Vd B$, hence $w \Vd \diam B$.
Therefore $\M\models \diam A \imp\diam B$.

\item[(\axNbox)] For all $w$ and all $v\more w$, by (\cN), there is $\alpha\in\N(v)$.
Since $z\Vd\top$ for all $z\in\alpha$, we have $w \Vd \Box\top$. Thus $\M\models \Box\top$.

\item[(\axCbox)] Suppose that $w\Vd \Box A\land\Box B$.
Then for all $v \more w$, there are $\alpha,\beta\in\N(v)$ such that $\alpha\ufor A$ and $\beta\ufor B$.
By (\cC), $\alpha\cap\beta\in\N(v)$, moreover $\alpha\cap\beta\ufor A \land B$.
Hence $\M\models \Box A \land \Box B \imp \Box(A\land B)$.

\item[(\axKdiam)] Suppose that $w \Vd \Box(A \imp B)$ and $w \Vd \diam A$. Then  for all $v \more w$, 
there is $\alpha\in\N(v)$ such that $\alpha\ufor A \imp B$.
Now, suppose that $\beta\in\N(v)$. By (\cC), $\alpha\cap\beta\in\N(v)$.
Since $\alpha\cap\beta\subseteq\alpha$, $\alpha\cap\beta\ufor A \imp B$.
Moreover, by $w\Vd\diam B$, $\alpha\cap\beta\efor A$.
Thus $\alpha\cap\beta\efor B$, which implies $\beta\efor B$.
Since this holds for every $\beta\in\N(v)$,
$w\Vd\diam B$.
Therefore $\M\models \Box(A \imp B) \imp (\diam A \imp\diam B)$.

\item[(\axPdiam)] For all $w$ and all $v\more w$, by (\cP), $\emptyset\notin\N(v)$.
Hence, for all $\alpha\in\N(v)$, $\alpha\not=\emptyset$,
thus $\alpha\efor\top$. Then we have $w \Vd \diam\top$. Thus $\M\models \diam\top$.

\item[(\axD)] Suppose that $w \Vd \Box A$.
Then  for all $v \more w$, 
there is $\alpha\in\N(v)$ such that $\alpha\ufor A$.
Now, suppose that $\beta\in\N(v)$. By (\cD), there is $z\in\alpha\cap\beta$.
Then $z\Vd A$, hence $\beta\efor A$, therefore $w \Vd\diam A$.
Hence $\M\models\Box A \imp \diam A$.

\item[(\axTbox)] Suppose that $w \Vd \Box A$.
Then  for all $v \more w$, 
there is $\alpha\in\N(v)$ such that $\alpha\ufor A$.
Hence in particular there is $\alpha\in\N(w)$ such that $\alpha\ufor A$.
By (\cT), $w\in\alpha$, then $w \Vd A$.
Therefore $\M\models\Box A \imp A$.

\item[(\axTdiam)] Suppose that $w \Vd A$.
By the hereditary property of minimal neighbourhood models,
for all $v\more w$, $v\Vd A$.
Moreover, by (\cT), for all $\alpha\in\N(v)$, $v\in\alpha$, hence $\alpha\efor A$.
Thus $w \Vd \diam A$, therefore $\M\models \diam A$.\qedhere
\end{enumerate}
\end{proof}


The proof of completeness proceeds essentially as the one in Section~\ref{sec:MK companion}.
First, we observe that Lemma~\ref{lemma:lind rel} also holds for all logics $\MLvar$.
We consider the following definition of neighbourhood segment.

\begin{definition}\label{def:neigh seg}
For every logic $\logic$ in $\lan$, 
an $\logic$-\emph{neighbourhood \seg},  
or just \emph{\seg}, is a pair $(\Phi, \CC)$, 
where $\Phi$ is an $\logic$-full set, and $\CC$ is a class of sets of $\logic$-full sets such that:
\begin{itemize}
\item if $\Box A\in\Phi$, then there is $\U\in\CC$ such that for all $\Psi\in\U$, $A\in\Psi$; and
\item if $\diam A\in\Phi$, then for all $\U\in\CC$, there is $\Psi\in\U$ such that $A\in\Psi$.
\end{itemize}
Moreover, if $\logic$ contains the axiom \axCbox, or the axiom \axD, or the axiom \axTbox, then 
the $\logic$-segments must satisfy the following corresponding condition:
\begin{center}
\begin{tabular}{llllrlll}
\emph{(\cCs)}  If $\U,\VV\in\CC$, then $\U\cap\VV\in\CC$. &
\emph{(\cTs)}  For all $\U\in\CC$, $\Phi\in\U$. \\
\emph{(\cDs)}  If $\U,\VV\in\CC$, then $\U\cap\VV\not=\emptyset$. \\
\end{tabular}
\end{center}
\end{definition}

\begin{lemma}\label{lemma:neigh seg}
For every minimal modal logic $\MLvar$ and every $\MLvar$-full set $\Phi$, 
\begin{itemize}
\item[(i)] there exists an $\MLvar$-neighbourhood segment $(\Phi, \CC)$;
\item[(ii)] if $\Box A\notin\Phi$, then there exists an $\MLvar$-neighbourhood segment $(\Phi, \CC)$ such that
for all $\U\in\CC$, there is $\Psi\in\U$ such that $A\notin\Psi$;
\item[(iii)] if $\diam A\notin\Phi$, then there exists an $\MLvar$-neighbourhood segment $(\Phi, \CC)$ such that 
there is $\U\in\CC$ such that for all $\Psi\in\U$, $A\notin\Psi$.
\end{itemize}
\end{lemma}
\begin{proof} \quad
\begin{enumerate}[leftmargin=*, align=left]
\item[(i)]
Given an $\MLvar$-full set $\Phi$, we construct an $\MLvar$-segment $(\Phi, \CC)$ as follows.
For all $\Box A\in\Phi$, 
we define 
$\U_{A}^- = \{\Psi \ \MLvar\text{-full} \mid A \in \Psi \text{ and there is } \diam B\in\Phi \text{ such that } B\in\Psi\}$;
and $\Ua = \U_{A}^-$ if $\MLvar$ does not contain \axTbox, and 
$\Ua = \U_{A}^- \cup\{\Phi\}$  
$\MLvar$ it contains \axTbox. 
Moreover, we define $\CC = \{\Ua \mid \Box A \in\Phi\}$.
We show that  $(\Phi, \CC)$ is an $\MLvar$-segment.
\begin{itemize}
\item If $\Box A\in\Phi$, then by definition $\Ua\in\CC$.
Moreover, if $\MLvar$ does not contain \axTbox, then $A\in\Psi$  for all $\Psi\in\Ua$.
If instead $\MLvar$ contains \axTbox, then for all $\Psi\in\Ua$ we have $A\in\Psi$ or $\Psi=\Phi$,
where,
by \axTbox\ and closure under derivation of $\MLvar$-full sets, $A\in\Phi$.

\item If $\diam A\in\Phi$, then assume $\U\in\CC$.
Then, by definition, $\U = \Ub$ for some $\Box B\in\Phi$.
By Lemma~\ref{lemma:lind rel}, there is an $\MLvar$-full set $\Psi$ such that $A,B\in\Psi$,
hence $\Psi\in\Ub = \U$ and $A\in\Psi$.
\end{itemize}
Moreover, the conditions (\cCs), (\cDs) and (\cTs) are satisfied if
$\MLvar$ contains the axioms \axCbox, \axD, or \axTbox,
respectively:
\begin{enumerate}[leftmargin=*, align=left]
\item[(\cCs)] Suppose $\U, \VV\in\CC$. Then $\U = \Ua$ and $\VV = \Ub$ for some $\Box A, \Box B\in\Phi$.
Hence, given that $\MLvar$ contains \axCbox, 
by closure under derivation of $\MLvar$-full sets, we have $\Box(A \land B)\in\Phi$, thus 
$\U_{A\land B}\in\CC$.
Note also that for all $\MLvar$-full sets $\Psi$ it holds $A,B\in\Psi$ if and only if $A\land B\in\Psi$.
One can easily verify that this implies $\U_{A\land B} = \Ua \cap \Ub$,
therefore $\U \cap \VV = \Ua \cap \Ub \in\CC$.

\item[(\cDs)] Suppose $\U, \VV\in\CC$. Then $\U = \Ua$ and $\VV = \Ub$ for some $\Box A, \Box B\in\Phi$.
Given that $\MLvar$ contains \axD,
by closure under derivation of $\MLvar$-full sets, we have $\diam A, \diam B\in\Phi$.
By Lemma~\ref{lemma:lind rel}, there is an $\MLvar$-full set
$\Psi$ such that $A, B\in\Psi$.
Then by definition, $\Psi\in\Ua$ and $\Psi\in\Ub$,
hence $\Psi\in\Ua\cap\Ub$,
 therefore $\U \cap \VV = \Ua \cap \Ub \not=\emptyset$.

\item[(\cTs)] By definition, for all $\U\in\CC$, $\Phi\in\U$.
\end{enumerate}

\item[(ii)]
For all $\Box B\in\Phi$, 
we define 
$\U_{B}^- = \{\Psi \ \MLvar\text{-full} \mid B \in \Psi \text{ and there is } \diam C\in\Phi \text{ such that } C\in\Psi\}
\cup\{\Psi \ \MLvar\text{-full} \mid B \in \Psi \text{ and } A \notin\Psi\}$;
and $\Ub = \U_{B}^-$ if $\MLvar$ does not contain \axTbox, and 
$\Ub = \U_{B}^- \cup\{\Phi\}$  
$\MLvar$ it contains \axTbox. 
Moreover, we define $\CC = \{\Ub \mid \Box B \in\Phi\}$.
We can show that $(\Phi, \CC)$ is an $\MLvar$-segment
as in item (i).
Now, suppose that $\U\in\CC$.
Then $\U = \Ub$ for some $\Box B\in\Phi$.
Thus, since $\Box A\notin\Phi$, 
$\{B\}\not\vd A$ 
(otherwise $\vd B \imp A$, and by \monbox, $\vd \Box B\imp\Box A$,
hence by closure under derivation, $\Box A\in\Phi$).
Then by Lemma~\ref{lemma:lind rel}, there is an $\MLvar$-full set $\Psi$ such that $B\in\Psi$ and $A \notin \Psi$,
and by definition, $\Psi\in\Ub = \U$.

\item[(iii)]
 If there is no $\diam B\in\Phi$,
then $(\Phi, \{\emptyset\})$ is an $\MLvar$-segment
(note that this never happens if $\MLvar$ contains the axiom \axD\ or the axiom \axTbox\
bacuse in both cases $\diam \top\in\Phi$
for all $\MLvar$-full sets $\Phi$),
moreover $\emptyset$ satisfies the claim of the lemma.
Otherwise we distinguish two subcases.

(iii.i) $\MLvar$ does not contain \axCbox, \axKdiam.
We define $\U^- = \{\Psi \ \MLvar\text{-full} \mid A \notin \Psi \text{ and there is } \diam B\in\Phi \text{ such that } B\in\Psi\}$, and for all
$\Box C\in\Phi$, 
we define 
$\U_{C}^- = \{\Psi \ \MLvar\text{-full} \mid C \in \Psi \text{ and there is } \diam B\in\Phi \text{ such that } B\in\Psi\}$.
Moreover, we define $\U = \U^-$, $\U_{C} = \U_{C}^-$
if $\MLvar$ does not contain \axTbox, 
and $\U = \U^- \cup\{\Phi\}$, $\U_{C} = \U_{C}^- \cup\{\Phi\}$
if $\MLvar$ contains \axTbox.
Finally, we define $\CC = \{\U\}\cup\{\U_{C} \mid \Box C \in \Phi\}$.
Note that $\U$ satisfies the condition of the lemma, in particular if \axTbox\ belongs to $\MLvar$, then $A\notin\Phi$,
since if $A\in\Phi$, then by \axTdiam, $\diam A\in\Phi$,
against the assumption.
We can show that $(\Phi, \CC)$ is an $\MLvar$-segment.
First, the conditions of $\MLvar$-segments for any $\Box B\in\Phi$ and for any $\diam B\in\Phi$ can be shown to be satisfied similarly to item (i).
Moreover, the property (\cTs) of $\MLvar$-segments for $\MLvar$ containing \axTbox\ follows immediately from the definition. 
We show that (\cDs) is satisfied if $\MLvar$ contains the axiom \axD:
Suppose that $\VV, \ZZ\in\CC$. 
If $\VV = \U_{C}, \ZZ = \U_{D}$ for some some $\Box C, \Box D \in \Phi$, the proof is analogous to
the one of (\cDs) in item (i).
Now suppose $\VV = \U_{C}$ for some some $\Box C \in \Phi$ and $\ZZ = \U$.
Then by axiom \axD, $\diam C\in\Phi$.
Thus we have $\{C\}\not\vd B$,
otherwise we would have $\vd C \imp A$,
and by \mondiam,
$\vd \diam C \imp \diam A$, hence $\diam A \in \Phi$,
against the assumption.
By Lemma~\ref{lemma:lind rel}, there is an $\MLvar$-full set
$\Psi$ such that $C\in\Psi$  and $A\notin \Psi$.
By definition, $\Psi\in\U$, moreover $\Psi\in\U_{C}$
(since $\diam C\in\Phi$),
hence $\Psi\in\U\cap\U_{C}=\VV\cap\ZZ$,
therefore $\VV\cap\ZZ\not=\emptyset$.

(iii.ii) $\MLvar$ contains \axCbox, \axKdiam.
We define $\U^- = \{\Psi \ \MLvar\text{-full} \mid A \notin\Psi, \text{ and } \boxm{\Phi}\subseteq\Psi,
\text{ and  } B\in\Psi \text{ for some } \diam B\in\Phi\}$,
and $\U = \U^-$ if $\MLvar$ does not contain \axTbox, and 
$\U = \U^- \cup\{\Phi\}$  
$\MLvar$ it contains \axTbox. 
Moreover, we define $\CC = \{\U\}$.
Clearly, $\U$ satisfies the claim of the lemma
(in particular, if $\MLvar$ contains \axTbox, then $A\notin\Phi$).
We show that $(\Phi, \CC)$ is an $\MLvar$-segment.
First, observe that for any $\diam B\in\Phi$,
$\boxm\Phi\cup\{B\}\not\vd A$.
Indeed, if $\boxm\Phi\cup\{B\}\vd A$,
then there are $C_1, ..., C_n\in\boxm\Phi$ such that
$\vd C_1\land ...\land C_n\land B \imp A$,
hence
$\vd C_1\land ...\land C_n\imp (B \imp A)$,
then by \monbox,
$\vd \Box(C_1\land ...\land C_n)\imp \Box(B \imp A)$,
thus by \axCbox\ ($n$ times) and \axKdiam,
$\vd \Box C_1\land ...\land \Box C_n\imp (\diam B \imp \diam A)$,
which gives
$\vd \Box C_1\land ...\land \Box C_n\land\diam B \imp \diam A$,
therefore
$\Box C_1, ...,  \Box C_n, \diam B \vd \diam A$;
since $\Box C_1, ...,  \Box C_n, \diam B\in\Phi$,
this entails $\diam A\in\Phi$,
against the assumption.
We then have:
the condition for any $\Box B\in\Phi$ follows immediately from the definition.
If $\diam B\in\Phi$,
then $\boxm\Phi\cup\{B\}\not\vd A$,
thus by Lemma~\ref{lemma:lind rel}, there is an $\MLvar$-full set $\Psi$ such that $\boxm\Phi\subseteq\Psi$, $B\in\Psi$ and $A \notin \Psi$,
hence $\Psi\in\U$.
By the same argument, the property (\cDs) is satisfied
for $\MLvar$ containing the axiom \axD\
given that $\diam\top\in\Phi$ entails the existence of such an $\MLvar$-full set $\Psi$, hence $\U\not=\emptyset$.
Moreover, (\cCs) is trivial, and (\cTs) for $\MLvar$ containing \axTbox\ follows immediately from the definition.\qedhere
\end{enumerate}
\end{proof}

\begin{definition}\label{def:can model neigh}
For every logic $\logic$ in $\lan$,
the \emph{canonical neighbourhood model} for $\logic$ is
the tuple $\Mc = \langle \Wc, \lessc, \FWc, \Nc, \Vc \rangle$,
where:
\begin{itemize}
\item $\Wc$ is the class of all $\logic$-neighbourhood segments;
\item for all $(\Phi, \CC), (\Psi, \DD)\in\Wc$,  $(\Phi, \CC) \lessc (\Psi, \DD)$ 
if and only if $\Phi\subseteq\Psi$;
\item for all $(\Phi, \CC)\in\Wc$,  $(\Phi,\CC)\in\FWc$ if and only if  $\bot\in\Phi$;
\item for all sets $\U$ of $\MLvar$-full sets, $\alpha_{\U} = \{(\Phi, \CC) \mid \Phi\in\U\}$; 
\item for all $(\Phi, \CC)\in\Wc$,  $\alpha_{\U}\in\Nc((\Phi, \CC))$ if and only if $\U\in\CC$;
\item for all $(\Phi, \CC)\in\Wc$,  $(\Phi,\CC)\in\Vc(p)$ if and only if $p\in\Phi$.
\end{itemize}
\end{definition}


\begin{lemma}
For every minimal modal logic $\MLvar$, the canonical neighbourhood model $\M$ for $\MLvar$ is a minimal neighbourhood model for $\MLvar$.
\end{lemma}
\begin{proof}
It is easy to very that $\M$ is a minimal neighbourhood model.
We show that $\M$ satisfies the conditions among
(\cC), (\cN), (\cP), (\cD), (\cT)
associated to the axioms of $\MLvar$.
\begin{enumerate}[leftmargin=*, align=left]
\item[(\cC)] Suppose that $\alpha, \beta\in\Nc((\Phi, \CC))$. Then, by definition, 
$\alpha=\alpha_{\U}$ and $\beta= \alpha_{\VV}$ for some $\U, \VV\in\CC$.
By the property (\cCs) of $\MLvar$-segments, $\U\cap\VV\in\CC$, thus $\alpha_{\U\cap\VV}\in\Nc((\Phi, \CC))$,
where
$\alpha_{\U\cap\VV} = \{(\Phi, \CC) \mid \Phi\in\U\cap\VV\} =
\{(\Phi, \CC) \mid \Phi\in\U\} \cap  \{(\Phi, \CC) \mid \Phi\in\VV\}
= \alpha_{\U} \cap \alpha_{\VV} = \alpha\cap\beta$.

\item[(\cN)] For all $\MLvar$-full sets $\Phi$, $\Box\top\in\Phi$,
then for all $\MLvar$-segments $(\Phi, \CC)$, $\CC\not=\emptyset$,
thus $\Nc((\Phi, \CC))\not=\emptyset$.

\item[(\cP)] For all $\MLvar$-full sets $\Phi$, $\diam\top\in\Phi$,
then for all $\MLvar$-segments $(\Phi, \CC)$ and all $\U\in\CC$, $\U\not=\emptyset$,
thus for all $\alpha_{\U}\in\Nc((\Phi, \CC))$, $\alpha_{\U}\not=\emptyset$,
that is, $\emptyset\notin\Nc((\Phi, \CC))$.

\item[(\cD)] Suppose that $\alpha, \beta\in\Nc((\Phi, \CC))$. Then 
$\alpha=\alpha_{\U}$ and $\beta= \alpha_{\VV}$ for some $\U, \VV\in\CC$.
By (\cDs), $\U\cap\VV\not=\emptyset$, which implies 
$\alpha_{\U} \cap \alpha_{\VV} = \alpha\cap\beta\not=\emptyset$.

\item[(\cT)] Suppose that $\alpha\in\Nc((\Phi, \CC))$. Then 
$\alpha=\alpha_{\U}$ for an $\U\in\CC$.
By (\cTs), $\Phi\in\U$, thus $(\Phi, \CC)\in\alpha_{\U} = \alpha$.\qedhere
\end{enumerate}
\end{proof}

\begin{lemma}\label{lemma:can model neigh}
Let $\MLvar$ be a minimal modal logic 
and $\Mc = \langle \Wc, \lessc, \FWc, \Nc, \Vc \rangle$ be the canonical neighbourhood model for $\MLvar$. Then for all $(\Phi, \CC)\in\Wc$ and all $A \in\lan$, 
$(\Phi, \CC)\Vd A$ if and only if $A\in\Phi$.
\end{lemma}
\begin{proof}
By induction on the construction of $A$.
For the cases $A = p, \bot, B\land C, B \lor C, B \imp C$
the proof is exactly as the proof of Lemma \ref{lemma:can model rel}. We consider the inductive cases $A = \Box B, \diam B$.
\begin{enumerate}[leftmargin=*, align=left]
\item[($A = \Box B$)]
Suppose that $\Box B\in\Phi$.
Then for all $(\Psi,\DD)\morec(\Phi,\CC)$, $\Box B\in\Psi$.
By definition of segment, there is $\U\in\DD$ such that for all $\Theta\in\U$, $B\in\Theta$.
Then, by definition of canonical model, $\alpha_{\U}\in\N((\Psi, \DD))$, and by \ih, $(\Theta,\EE)\Vd B$ for all $(\Theta,\EE)\in\alpha_{\U}$.
Therefore $(\Phi, \CC)\Vd \Box B$. 
Now suppose that $\Box B\notin\Phi$.
By Lemma~\ref{lemma:neigh seg} (ii),
there is an $\MLvar$-segment $(\Phi, \DD)$
such that for all $\U\in\DD$, there is $\Psi\in\U$ such that $B\notin\Psi$.
By definition, $(\Phi, \DD)\in\Wc$
and $(\Phi, \CC) \lessc (\Phi, \DD)$.
Moreover, assume $\alpha\in\Nc((\Phi, \DD))$.
Then $\alpha = \alpha_{\U}$ for some $\U\in\DD$.
Thus, there is $\Psi\in\U$ such that $B\notin\Psi$.
By Lemma~\ref{lemma:neigh seg} (i),
there is an $\MLvar$-segment $(\Psi, \EE)$,
thus by definition, 
$(\Psi, \EE)\in\alpha_{\U}$,
and by \ih, $(\Psi, \EE)\not\Vd B$.
Hence $\alpha = \alpha_{\U} \not\ufor B$,
therefore $(\Phi, \CC)\not\Vd\Box B$.

\item[($A = \diam B$)]
Suppose that $\diam B\in\Phi$.
Then for all $(\Psi,\DD)\morec(\Phi,\CC)$, $\diam B\in\Psi$.
By definition of segment, 
for all $\U\in\DD$, there is $\Psi\in\U$ such that $B\in\Psi$.
Now, assume $\alpha\in\Nc((\Psi,\DD))$.
By definition, $\alpha = \alpha_{\U}$ for some $\U\in\DD$.
Then there is $\Psi\in\U$ such that $B\in\Psi$.
By Lemma~\ref{lemma:neigh seg} (i),
there is an $\MLvar$-segment $(\Psi, \EE)$,
thus by definition, 
$(\Psi, \EE)\in\alpha_{\U}$,
and by \ih, $(\Psi, \EE)\Vd B$,
which implies $\alpha = \alpha_{\U} \efor B$.
Since this holds for every $\alpha\in\Nc((\Psi,\DD))$,
we have $(\Phi, \CC)\Vd\diam B$.
Now suppose that $\diam B\notin \Phi$.
By Lemma~\ref{lemma:neigh seg} (iii),
there is an $\MLvar$-segment $(\Phi, \DD)$
and a $\U\in\DD$ such that for all $\Psi\in\U$, $B\notin\Psi$.
By definition, $(\Phi, \DD)\in\Wc$,
$(\Phi, \CC) \lessc (\Phi, \DD)$,
and $\alpha_{\U}\in\Nc((\Phi, \DD))$.
Moreover, for all $(\Psi, \EE)\in\alpha_{\U}$,
$B\notin\Psi$,
then by \ih, $(\Psi, \EE)\not\Vd B$.
Hence,
$\alpha_{\U} \not\efor B$,
therefore $(\Phi, \CC)\not\Vd\diam B$.\qedhere
\end{enumerate}
\end{proof}

As a consequence of these lemmas, we obtain the following completeness result
(cf.~proof of Theorem~\ref{th:compl MK}).

\begin{theorem}[Completeness]
For all $A\in\lan$  and all minimal modal logics $\MLvar$,
if $A$ is valid in every minimal neighbourhod model for $\MLvar$, then $A$ is derivable in $\MLvar$.
\end{theorem}

Finally, on the basis of this semantic characterisation of logics $\MLvar$,
we can show that, for each classical logic $\logic$,
the fusion $\Sfour\oplus\logic$ is the bimodal companion of the corresponding minimal logic $\MLvar$.

\begin{theorem}\label{th:companion ML}
For all $A\in\lan$ and all minimal modal logics $\MLvar$,
$A$ is derivable in $\MK$ if and only if $A^{\tr}$ is derivable in $\Sfour\oplus\K$.
\end{theorem}
\begin{proof}
($\Rightarrow$)
Suppose that $\Sfour\oplus\logic\not\vd A^{\tr}$. 
Then there are a model $\M = \langle \W, \R, \N, \V\rangle$
for $\Sfour\oplus\logic$ and a world $w$ such that $\M, w\not \Vd A^{\tr}$.
We define $\M' = \langle \W, \less, \FW, \N, \V'\rangle$ over the same $\W$ and $\N$,
where
$\less \ = \R$, 
for all $p\in\atm$, $\V'(p) = \{v \mid \textup{for all } u,  v\R u \textup{ implies } u\in\V(p)\}$,
and $\FW = \{v \mid \textup{for all } u,  v\R u \textup{ implies } u\in\V(\fp)\}$.
By the properties of $\N$ in $\M$,
it immediately follows that $\M'$ is a minimal neighbourhood model for $\MLvar$.
We show that for all $v\in\W$ and all $B\in\lan$,
$\M', v \Vd B$ if and only if $\M, v \Vd B^{\tr}$,
which implies that $\M', w \not\Vd A$,
therefore $\MLvar\not\vd A$.
The proof is by induction on the construction of $B$.
The cases $B = p, \bot, C\land D, C\lor D, C\imp D$ are as in the proof of Theorem~\ref{th:companion MK},
case ($\Rightarrow$).
We show the cases $B = \Box C, \diam C$.

\begin{enumerate}[leftmargin=*, align=left]
\item[($B = \Box C$)] $\M', v \Vd \Box C$ iff
for all $u\more v$, there is $\alpha\in\N(u)$ such that for all $z\in\alpha$, $\M', z \Vd C$;
iff
(by definition of $\less$ and \ih)
for all $u$, if $v \R u$, then there is $\alpha\in\N(u)$ such that for all $z\in\alpha$, $\M, z \Vd C^{\tr}$;
iff
for all $u$, if $v \R u$, then $\M, v \Vd \Box_2 C^{\tr}$;
iff
$\M, v \Vd \Box_1\Box_2 C^{\tr}$.

\item[($B = \diam C$)] $\M', v \Vd \diam C$ iff
for all $u\more v$, for all $\alpha\in\N(u)$, there is $z\in\alpha$ such that $\M', z \Vd C$;
iff
(by definition of $\less$ and \ih)
for all $u$, if $v \R u$, then for all $\alpha\in\N(u)$, there is $z\in\alpha$ such that $\M, z \Vd C^{\tr}$;
iff
for all $u$, if $v \R u$, then $\M, v \Vd \diam_2 C^{\tr}$;
iff
$\M, v \Vd \Box_1\diam_2 C^{\tr}$.
\end{enumerate}

\smallskip
\noindent
($\Leftarrow$) 
Suppose that $\MLvar\not\vd A$. 
Then there are a minimal neighbourhood model $\M = \langle \W, \less, \FW, \N, \V\rangle$
for $\MLvar$ and a world $w$ such that $\M, w\not \Vd A$.
We define $\M'' = \langle \W, \R, \N, \V''\rangle$ over the same $\W$ and $\N$,
where
$\R = \less$, 
for all $p\in\atm$, $\V''(p) = \V(p)$,
and $\V''(\fp) = \FW$.
Then $\M''$ is a model for $\Sfour\oplus\logic$.
We show that for all $v\in\W$ and all $B\in\lan$,
$\M, v \Vd B$ if and only if $\M'', v \Vd B^{\tr}$,
which implies that $\M'', w \not\Vd A$,
therefore $\Sfour\oplus\logic\not\vd A$.
The proof is by induction on the construction of $B$.
The cases $B = p, \bot, C\land D, C\lor D, C\imp D$ are as in the proof of Theorem~\ref{th:companion MK},
case ($\Leftarrow$).
We show the cases $B = \Box C, \diam C$.

\begin{enumerate}[leftmargin=*, align=left]
\item[($B = \Box C$)] $\M, v \Vd \Box C$ iff
for all $u\more v$, there is $\alpha\in\N(u)$ such that for all $z\in\alpha$, $\M, z \Vd C$;
iff
(by definition of $\R$ and \ih)
for all $u$, if $v \R u$, then there is $\alpha\in\N(u)$ such that for all $z\in\alpha$, $\M'', z \Vd C^{\tr}$;
iff
for all $u$, if $v \R u$, then $\M'', v \Vd \Box_2 C^{\tr}$;
iff
$\M'', v \Vd \Box_1\Box_2 C^{\tr}$.

\item[($B = \diam C$)] $\M, v \Vd \diam C$ iff
for all $u\more v$, for all $\alpha\in\N(u)$, there is $z\in\alpha$ such that $\M, z \Vd C$;
iff
(by definition of $\R$ and \ih)
for all $u$, if $v \R u$, then for all $\alpha\in\N(u)$, there is $z\in\alpha$ such that $\M'', z \Vd C^{\tr}$;
iff
for all $u$, if $v \R u$, then $\M'', v \Vd \diam_2 C^{\tr}$;
iff
$\M'', v \Vd \Box_1\diam_2 C^{\tr}$.\qedhere
\end{enumerate}
\end{proof}

\subsection{Minimal modal logics via sequent calculi}

\begin{figure}[t]
\begin{small}

\ax{$A \Seq B$}
\llab{\ruleMboxcl}
\uinf{$\Box A \Seq \Box B$}
\disp
\hfill
\ax{$A \Seq B$}
\llab{\ruleMdiamcl}
\uinf{$\diam A \Seq \diam B$}
\disp
\hfill
\ax{$A, B \Seq$}
\llab{\rulemncMcl}
\uinf{$\Box A, \diam B \Seq$}
\disp
\hfill
\ax{$\Seq A, B$}
\llab{\rulememMcl}
\uinf{$\Seq \Box A,  \diam B$}
\disp

\vspace{0.2cm}
\ax{$\Sigma, A \Seq B, \Pi$}
\llab{\ruleCboxcl}
\uinf{$\Box\Sigma, \Box A \Seq \Box B,  \diam\Pi$}
\disp
\hfill
\ax{$\Sigma, A \Seq B, \Pi$}
\llab{\ruleCdiamcl}
\uinf{$\Box\Sigma, \diam A \Seq \diam B, \diam\Pi$}
\disp
\hfill
\ax{$\Sigma, A, B \Seq$}
\llab{\rulemncCcl}
\uinf{$\Box\Sigma, \Box A, \diam B \Seq$}
\disp

\vspace{0.2cm}
\ax{$\Seq A, B, \Pi$}
\llab{\rulememCcl}
\uinf{$\Seq \Box A,  \diam B, \diam\Pi$}
\disp
\hfill
\ax{$\Seq A$}
\llab{\ruleNboxcl}
\uinf{$\Seq \Box A$}
\disp
\hfill
\ax{$A \Seq$}
\llab{\ruleNdiamcl}
\uinf{$\diam A \Seq$}
\disp
\hfill
\ax{$\Sigma \Seq A, \Pi$}
\llab{\ruleKboxcl}
\uinf{$\Box\Sigma \Seq \Box A,  \diam\Pi$}
\disp

\vspace{0.2cm}
\ax{$\Sigma, A \Seq \Pi$}
\llab{\ruleKdiamcl}
\uinf{$\Box\Sigma, \diam A \Seq \diam\Pi$}
\disp
\hfill
\ax{$A \Seq$}
\llab{\rulePboxcl}
\uinf{$\Box A \Seq$}
\disp
\hfill
\ax{$\Seq A$}
\llab{\rulePdiamcl}
\uinf{$\Seq\diam A$}
\disp
\hfill
\ax{$A \Seq B$}
\llab{\ruleDcl}
\uinf{$\Box A \Seq \diam B$}
\disp
\hfill
\ax{$A, B \Seq$}
\llab{\ruleDboxcl}
\uinf{$\Box A, \Box B \Seq$}
\disp

\vspace{0.2cm}
\ax{$\Seq A, B$}
\llab{\ruleDdiamcl}
\uinf{$\Seq \diam A, \diam B$}
\disp
\hfill
\ax{$\Sigma \Seq \Pi$}
\llab{\ruleCDcl}
\uinf{$\Box\Sigma \Seq \diam\Pi$}
\disp
\hfill
\ax{$\G, A \Seq \D$}
\llab{\ruleTboxcl}
\uinf{$\G, \Box A \Seq \D$}
\disp
\hfill
\ax{$\G \Seq A, \D$}
\llab{\ruleTdiamcl}
\uinf{$\G \Seq \diam A, \D$}
\disp
\end{small}
\caption{\label{fig:classical modal seq rules} 
Modal  rules for classical sequent calculi $\Gone\logic$.} 
\end{figure}

%
%

\begin{figure}[t]
\begin{small}
\ax{$A \Seq B$}
\llab{\rulemMbox}
\uinf{$\Box A \Seq \Box B$}
\disp
\hfill
\ax{$A \Seq B$}
\llab{\rulemMdiam}
\uinf{$\diam A \Seq \diam B$}
\disp
\hfill
\ax{$\Seq A$}
\llab{\rulemNbox}
\uinf{$\Seq \Box A$}
\disp
\hfill
\ax{$\Sigma, A \Seq B$}
\llab{\rulemCbox}
\uinf{$\Box\Sigma, \Box A \Seq \Box B$}
\disp

\vspace{0.2cm}
\ax{$\Sigma \Seq A$}
\llab{\rulemKbox}
\uinf{$\Box\Sigma \Seq \Box A$}
\disp
\hfill
\ax{$\Sigma, A \Seq B$}
\llab{\rulemKdiam}
\uinf{$\Box\Sigma, \diam A \Seq \diam B$}
\disp
\hfill
\ax{$\Seq A$}
\llab{\rulemPdiam}
\uinf{$\Seq\diam A$}
\disp
\hfill
\ax{$A \Seq B$}
\llab{\rulemD}
\uinf{$\Box A \Seq \diam B$}
\disp

\vspace{0.2cm}
\qquad \quad \hfill
\ax{$\Sigma \Seq A$}
\llab{\rulemCD}
\uinf{$\Box\Sigma \Seq \diam A$}
\disp
\hfill
\ax{$\G, A \Seq C$}
\llab{\rulemTbox}
\uinf{$\G, \Box A \Seq C$}
\disp
\hfill
\ax{$\G \Seq A$}
\llab{\rulemTdiam}
\uinf{$\G \Seq \diam A$}
\disp
\hfill \quad\qquad
\end{small}
\caption{\label{fig:minimal modal seq rules} 
Modal rules for minimal sequent calculi $\Gone\MLvar$.}
\end{figure}

G1-style sequent calculi for the considered classical modal logics are defined extending $\Gone\CPL$
(Figure~\ref{fig:prop seq rules}) with the following modal rules from Figure~\ref{fig:classical modal seq rules}:
\begin{center}
\begin{tabular}{lllllll}
$\seqEM$ := \ruleMboxcl, \ruleMdiamcl, \rulemncMcl, \rulememMcl &
$\seqEMP$ := $\seqEM$ + \rulePboxcl, \rulePdiamcl \\

$\seqEMN$ := $\seqEM$ + \ruleNboxcl, \ruleNdiamcl &
$\seqEMNP$ := $\seqEMN$ + \rulePboxcl, \rulePdiamcl \\

$\seqEMC$ := \ruleCboxcl, \ruleCdiamcl, \rulemncCcl, \rulememCcl \\

\vspace{0.2cm}
$\seqK$ := \ruleKboxcl, \ruleKdiamcl \\

$\seqEMD$ := $\seqEM$ + \ruleDcl, \ruleDboxcl, \ruleDdiamcl &
$\seqEMT$ := $\seqEM$ + \ruleTboxcl, \ruleTdiamcl \\

$\seqEMND$ := $\seqEMN$ + \ruleDcl, \ruleDboxcl, \ruleDdiamcl &
$\seqEMNT$ := $\seqEMN$ + \ruleTboxcl, \ruleTdiamcl \\

$\seqEMCD$ := $\seqEMC$ + \ruleCDcl &
$\seqEMCT$ := $\seqEMC$ + \ruleTboxcl, \ruleTdiamcl \\

$\seqKD$ := $\seqK$ + \ruleCDcl &
$\seqKT$ := $\seqK$ + \ruleTboxcl, \ruleTdiamcl \\
\end{tabular}
\end{center}

These calculi are studied and shown to be cut-free complete in \cite{indrzejczak:2005,lavendhomme:2000,Lellmann:2019,orlandelli:2020}.
By applying the single-succedent restriction to the classical calculi $\Gone\logic$,
we obtain the corresponding calculi $\Gone\MLvar$ 
which extend $\Gone\MPL$
(Figure~\ref{fig:prop seq rules}) with the following modal rules from Figure~\ref{fig:classical modal seq rules}:

\begin{center}
\begin{tabular}{lllllll}
$\Gone\MM$ := \rulemMbox, \rulemMdiam &
$\Gone\MMP$ := $\Gone\MM$ + \rulemPdiam \\

$\Gone\MMN$ := $\Gone\MM$ + \rulemNbox &
$\Gone\MMNP$ := $\Gone\MMN$ + \rulemPdiam \\

$\Gone\MMC$ := \rulemCbox, \rulemKdiam\\

\vspace{0.2cm}
$\Gone\MK$ := \rulemKbox, \rulemKdiam \\

$\Gone\MMD$ := $\Gone\MM$ + \rulemD, \rulemPdiam &
$\Gone\MMT$ := $\Gone\MM$ + \rulemTbox, \rulemTdiam \\

$\Gone\MMND$ := $\Gone\MMN$ + \rulemD, \rulemPdiam &
$\Gone\MMNT$ := $\Gone\MMN$ + \rulemTbox, \rulemTdiam \\

$\Gone\MMCD$ := $\Gone\MMC$ + \rulemCD &
$\Gone\MMCT$ := $\Gone\MMC$ + \rulemTbox, \rulemTdiam \\

$\Gone\MKD$ := $\Gone\MK$ + \rulemCD &
$\Gone\MKT$ := $\Gone\MK$ + \rulemTbox, \rulemTdiam \\
\end{tabular}
\end{center}

As before, the rules containing sequents with an empty succedent
or with two active/principal formulas in the succedent are dropped
(namely, \rulemncMcl, \rulememMcl, \rulemncCcl, \rulememCcl, \ruleNdiamcl, \rulePboxcl, \ruleDboxcl\ and  \ruleDdiamcl),
while the remaining rules preserve only one formula in the consequent of sequents
(note in particular that the modal context $\diam\Pi$ is removed from \rulemCbox, \rulemKbox\ and \rulemKdiam).
Observe also that the single-succedent restriction applied to \ruleCdiamcl\ and \ruleKdiamcl\ produces the same rule \rulemKdiam.
Finally, the calculi $\Gone\MMD$ and $\Gone\MMND$ contain the rule \rulemPdiam\ 
that corresponds to the restriction of the rule \ruleDdiamcl\ in the particular case where $A = B$
(\rulePdiamcl\ is derivable in $\seqEMD$ and $\seqEMND$ from \ruleDdiamcl\ and \lctrcl).

We now show that the rule \cut\ is admissible in the calculi $\Gone\MLvar$.
As a consequence of this result,
we prove that the calculi $\Gone\MLvar$ are equivalent to the corresponding axiomatic systems $\MLvar$.

\begin{restatable}{theorem}{ThCutML}\label{th:cut ML}
For every calculus $\Gone\MLvar$, 
the rule \cut\ is admissible in $\Gone\MLvar$.
\end{restatable}
\begin{proof}
The proof is in the appendix.
\end{proof}

\begin{theorem}\label{th:ML equiv G1ML}
For every calculus $\Gone\MLvar$, 
for all $A\in\lan$,
$A$ is derivable in $\Gone\MLvar$
if and only if
$A$ is derivable in $\MLvar$.
\end{theorem}
\begin{proof}

($\Rightarrow$) For every modal rule $\varseq_1, ..., \varseq_n/\varseq$ of $\Gone\MLvar$,
we show that the rule $\fint(\varseq_1), ..., \fint(\varseq_n)/\fint(\varseq)$
is derivable in $\MLvar$.
(\rulemMbox) From $A \imp B$, by \monbox\ we get $\Box A \imp \Box B$.
(\rulemMdiam) From $A \imp B$, by \mondiam\ we get $\diam A \imp \diam B$.
(\rulemNbox) From $A$ we get $\top\imp A$, then by \monbox, $\Box \top \imp \Box A$,
hence with $\Box\top$ we obtain $\Box A$.
(\rulemPdiam) From $A$ we get $\top\imp A$, then by \mondiam, $\diam \top \imp \diam A$,
hence with $\diam\top$ we obtain $\diam A$.
(\rulemD) From $A \imp B$, by \monbox, $\Box A \imp \Box B$,
then with $\Box B \imp \diam B$ we get $\Box A imp \diam B$.
(\rulemCD) Assume $\Sigma = A_1 \land ... \land A_n$.
Then from $A_1 \land ... \land A_n \imp B$, by \monbox\ we get $\Box(A_1 \land ... \land A_n) \imp \Box B$.
From \axCbox\ we have $\Box A_1 \land ... \land \Box A_n \imp\Box(A_1 \land ... \land A_n)$,
then 
with $\Box B \imp \diam B$ we obtain $\Box A_1 \land ... \land \Box A_n \imp \diam B$.
(\rulemTbox) From $\AND\G \land A \imp B$, with $\Box A \imp A$ we get $\AND\G \land \Box A \imp B$.
(\rulemTdiam) From $\AND\G \imp A$, with $A \imp\diam A$ we get $\AND\G \imp \diam A$.
For \rulemCbox, \rulemKbox\ and \rulemKdiam\ see the derivations in Figure~\ref{fig:der MK},
replacing  consecutive applications of \nec\ and \axKbox\ with one application of \monbox.

($\Leftarrow$)
For the other direction, it is easy to see that the modal axioms and rules of $\MLvar$ are derivable, respectively admissible, in $\Gone\MLvar$.
We show as examples the derivations of \axCbox\ and \monbox.
\begin{center}
\begin{small}
\ax{$A, B \seq A$}
\ax{$A, B \seq B$}
\rlab{\mrland}
\binf{$A, B \seq A \land B$}
\rlab{\rulemCbox}
\uinf{$\Box A, \Box B \seq \Box(A\land B)$}
\rlab{\mlland}
\uinf{$\Box A \land \Box B, \Box B \seq \Box(A\land B)$}
\rlab{\mlland}
\uinf{$\Box A \land \Box B, \Box A \land \Box B \seq \Box(A\land B)$}
\rlab{\mlctr}
\uinf{$\Box A \land \Box B \seq \Box(A\land B)$}
\rlab{\mrimp}
\uinf{$\Seq \Box A \land \Box B \imp \Box(A\land B)$}
\disp
\qquad
\ax{$\seq A \imp B$}
\ax{$A \imp B, A \seq B$}
\rlab{\cut}
\binf{$A \seq B$}
\rlab{\ruleMbox}
\uinf{$\Box A \seq \Box B$}
\rlab{\mrimp}
\uinf{$\Seq \Box A \imp \Box B$}
\disp
\end{small}
\end{center}
\end{proof}

\subsection{Constructive modal logics}

On the basis of the relations between $\MK$ and $\CK$  observed in Section~\ref{sec:rel MK CK},
we now define a constructive counterpart for each logic $\MLvar$.
First, the constructive modal logics $\CLvar$ are defined extending $\MLvar$ with ex falso quodlibet $\bot\imp A$.

\begin{definition}[Constructive modal logics]
For every minimal modal logic $\MLvar$,
the corresponding constructive modal logic $\CLvar$
is defined 
as $\MLvar + \bot\imp A$.
\end{definition}

We show that each logic $\CLvar$ is semantically characterised by neighbourhood models obtained  by suitably restricting the minimal neighbourhood models for the corresponding system $\MLvar$. 
The restriction is analogous to the one of Definition~\ref{def:const birel model},
with the difference that the neighbourhood function is now involved.

\begin{definition}[Constructive neighbourhood semantics]\label{def:const neigh model}
For every constructive modal logic $\CLvar$,
a \emph{constructive neighbourhood model} for $\CLvar$
is any minimal neighbourhood model $\M = \langle \W, \less, \FW, \N, \V\rangle$
for the corresponding minimal logic $\MLvar$ such that 
the following hold for all $w\in\FW$:
\begin{itemize}
\item[(i)] $w\in\V(p)$ for all $p\in\atm$;
\item[(ii)] there is $\alpha\in\N(w)$ such that $\alpha\subseteq\FW$; 
\item[(iii)] for all $\alpha\in\N(w)$, $\alpha\cap\FW\not=\emptyset$.
\end{itemize}
\end{definition}

\begin{theorem}
For all $A\in\lan$  and all constructive modal logics $\CLvar$,
$A$ is valid in all constructive neighbourhod models for $\CLvar$
if and only if $A$ is derivable in $\CLvar$.
\end{theorem}
\begin{proof}
($\Rightarrow$)
The proof extends the one of Theorem~\ref{th:soundness ML} by showing that $\bot\imp A$ is 
valid in every constructive neighbourhood model.
Suppose that $w \Vd \bot$.
We show by construction on $A$ that $w \Vd A$, considering only the cases $A = \Box A, \diam A$
(see the proof of Theorem~\ref{th:sound CK} for $A = p, \bot, B\land C, B\lor C, B\imp C$).
($A = \Box B$)
Suppose $w \less v$. Since $\FW$ is $\less$-upward closed, $v \in\FW$.
Then by Definition~\ref{def:const neigh model}, item (ii),
there is $\alpha\in\N(v)$ such that $\alpha\subseteq\FW$.
Hence $\alpha\ufor\bot$, and by \ih, $\alpha\ufor B$.
Therefore $w \Vd \Box B$.
($A = \diam B$)
Suppose $w \less v$. Since $\FW$ is $\less$-upward closed, $v \in\FW$.
Then by Definition~\ref{def:const neigh model}, item (iii),
for all $\alpha\in\N(v)$, $\alpha\cap\FW\not=\emptyset$.
Hence for all $\alpha\in\N(v)$, $\alpha\efor\bot$, then by \ih, $\alpha\efor B$.
Therefore $w \Vd \diam B$.

($\Leftarrow$)
The proof extends the completeness proof of minimal modal logics by showing that the canonical neighbourhood model for $\CLvar$ (Definition~\ref{def:can model neigh})
satisfies the conditions (i), (ii), (iii) in Definition~\ref{def:const neigh model}.
Suppose that $(\Phi, \CC)\in\FWc$. Then $\bot\in\Phi$.
Since $\Phi$ is closed under derivation, by ex falso quodlibet we obtain $\Phi = \lan$,
which entails the following.
(i) For all $p\in\atm$, $p\in\Phi$, hence by definition, $(\Phi, \CC)\in\Vc(p)$.
(ii) $\Box\bot\in\Phi$, hence by Definition~\ref{def:neigh seg}, 
there is $\U\in\CC$ such that  for all $\Psi\in\U$, $\bot\in\Psi$.
Then by Definition~\ref{def:can model neigh},
there is $\alpha_{\U}\in\Nc((\Phi, \CC))$ such that for all $(\Psi, \DD)\in\alpha_{\U}$, $\bot\in\Psi$.
Then for all $(\Psi, \DD)\in\alpha_{\U}$, $(\Psi, \DD)\in\FWc$, thus $\alpha_{\U}\subseteq\FWc$.
(iii) $\diam\bot\in\Phi$,  hence by Definition~\ref{def:neigh seg}, 
for all $\U\in\CC$, there is $\Psi\in\U$ such that $\bot\in\Psi$.
Then by Definition~\ref{def:can model neigh} and Lemma~\ref{lemma:neigh seg}
(which holds for $\CLvar$-full segments as well),
for all $\alpha_{\U}\in\Nc((\Phi, \CC))$, there is $(\Psi, \DD)\in\alpha_{\U}$ such that $\bot\in\Psi$,
hence $(\Psi, \DD)\in\FWc$, thus $\alpha_{\U}\cap\FWc\not=\emptyset$.
\end{proof}

Now, we show that for each logic $\CLvar$, a sequent calculus $\Gone\CLvar$ can be obained by 
extending $\Gone\IPL$ with the modal rules of the corresponding calculus $\Gone\MLvar$.

\begin{definition}
For every logic $\CLvar$, the sequent calculus $\Gone\CLvar$ contains the rules of $\Gone\IPL$
(Figure~\ref{fig:G1IPL}) plus the modal rules of the corresponding minimal calculus $\Gone\MLvar$,
except for \rulemTbox\ which is replaced by its intuitionistic version \ruleiTbox, with $0 \leq |\delta| \leq 1$:
\begin{center}
\ax{$\G, A \seq \delta$}
\llab{\ruleiTbox}
\uinf{$\G, \Box A \seq \delta$}
\disp
\end{center}
\end{definition}

Differently from the other modal rules, \ruleTboxcl\ and \ruleTdiamcl\ are local and must therefore be treated like the propositional rules. 
Since \ruleTdiamcl\ has a principal formula in the succedent
which is preserved by both kinds of sequent restrictions, this only impacts on \ruleiTbox\
that requires an intuitionistic succedent containing zero or one formula.

\begin{restatable}{theorem}{ThCutCL}\label{th:cut CL}
For every calculus $\Gone\CLvar$, 
the rule \cut\ is admissible in $\Gone\CLvar$.
\end{restatable}
\begin{proof}
The proof is in the appendix.
\end{proof}

\begin{theorem}\label{th:CL equiv G1CL}
For every calculus $\Gone\CLvar$, 
for all $A\in\lan$,
$A$ is derivable in $\Gone\CLvar$
if and only if
$A$ is derivable in $\CLvar$.
\end{theorem}
\begin{proof}
The derivations of the intuitionistic axioms and sequent rules are standard.
For the derivations of the modal axioms and sequent rules we refer to the proof of Theorem~\ref{th:ML equiv G1ML}.
\end{proof}

\section{Discussion}\label{sec:discussion}

\textbf{A framework of minimal and constructive modal logics.}
We have defined a family of minimal modal logics and a related family of constructive modal logics
corresponding each to a different classical modal logic.
The minimal modal logics have been defined
by
means of (1) a reduction to fusions of classical modal logics via the extended G{\"o}del-Johansson translation,
as well as (2)  the restriction of sequent calculi for classical modal logics to single-succedent sequents.
We have seen that the resulting minimal counterpart of $\K$ is strictly 
connected with the constructive modal logic $\CK$,
as the two systems essentially validate the same modal principles.
In particular, $\CK$ can be obtained from $\MK$ 
(1) semantically, by forcing the validity of ex falso quodlibet in minimal birelational models;
(2) based on the sequent calculi, by adding the minimal modal rules to an intuitionistic sequent calculus;
(3) axiomatically, by extending $\MK$ with ex falso quodlibet $\bot\imp A$.
Based on this 
relation between $\MK$ and $\CK$, we have defined a constructive correspondent for each minimal system.
Among the resulting constructive logics, the systems
$\CK$, $\CKD$ and $\CKT$ coincides with the constructive counterparts of $\K$, $\KD$ and $\KT$ studied in 
\cite{arisaka,Mendler2}
($\CKT$ also coincides with the propositional fragment of Fitch's first-order intuitionistic modal logic \cite{Fitch:1948}).
All in all, this work 
organises some constructive modal logics already studied in the literature into a uniform framework
 and also extends this family with constructive counterparts of some non-normal modal logics.
 As a consequence of the applied methodology, all minimal and constructive modal logics are
 endowed with cut-free sequent calculi and a modular semantic characterisation.
 
 \textbf{Simpson's requirements.} 
 Simpson~\cite{Simpson:1994}  listed some requirements 
 to evaluate whether an intuitionistic modal logic $\mf{IL}$ 
 can be understood as an intuitionistic counterpart of a classical modal logic $\logic$:
 $\mf{IL}$ should be a conservative extension of $\IPL$,
 it should contain all axioms of $\IPL$ (over the whole language $\lan$) and be closed under modus ponens,
 it should satisfy the disjunction property
 (if $A \lor B$ is derivable, then $A$ is derivable or $B$ is derivable),
 the modalities in $\mf{IL}$ should be independent,
 the extension of $\mf{IL}$ with $A \lor \neg A$ should coincide with $\logic$.
It looks natural to adapt these requirements to pairs of minimal and constructive/intuitionistic modal logics.
It is easy to verify that each logic $\MLvar$
is a conservative extension of $\MPL$,
contains all axioms of $\MPL$ and modus ponens,
satisfies the disjunction property
and has independent modalities.
Moreover,  the extension of $\MLvar$ with $\bot\imp A$ coincides with the corresponding logic $\CLvar$.
In this sense, each pair of corresponding logics $\MLvar$ and $\CLvar$ is a Simpsonian pair of modal logics.

\textbf{\wij-style constructive modal logics.}
In \cite{dalmonte:2022},
a family of \wij-style constructive modal logics $\WLvar$
(namely, $\WK$ and related systems)
is defined, where each logic $\WLvar$ corresponds to a different classical modal logic 
among the same set of classical modal logics considered in the present paper.
The logics $\WLvar$ are defined by extending
 to modal logics the relation that holds between the sequent calculi $\Gone\CPL$ and $\Gone\IPL$:
 The sequent calculi for the classical modal logics are restricted to sequents with \emph{at most} one
 formula in the succedent.
 At the same time, the logics $\WLvar$ can be shown to be reducible to fusions $\Sfour\oplus\logic$ by means of an extended G{\"o}del translation $g$ 
identical to $\tr$ in Definition~\ref{def:translation}
except for the clause $\bot^{g}=\Box\bot$.
Semantically, the models for $\WLvar$ logics defined in \cite{dalmonte:2022}
coincide with the restriction of minimal neighbourhood models 
to $\FW=\emptyset$.

The difference between the logics $\WLvar$ and $\CLvar$
is particularly evident from the point of view of the sequent calculi.
First, we observe that each
classical sequent calculus contains (possibly as particular cases) the two rules
\begin{center}
\ax{$A, B \Seq$}
\llab{$\mathsf{mnc}$}
\uinf{$\Box A, \diam B \Seq$}
\disp
\qquad
\ax{$\Seq A, B$}
\llab{$\mathsf{mem}$}
\uinf{$\Seq \Box A,  \diam B$}
\disp
\end{center}
that allow one to derive the axioms $\neg(\Box A \land \diam\neg A)$ 
and $\Box A \lor \diam\neg A$,
whose conjuction is equivalent to the duality principle $\Box A \tto \neg\diam\neg A$. 
Because of the different sequent restrictions, the calculi $\Gone\WLvar$ preserve $\mathsf{mnc}$ but reject $\mathsf{mem}$,
whereas the calculi $\Gone\CLvar$ reject both rules.
This is the essential 
difference between the logics $\WLvar$ and $\CLvar$.
Indeed, each logic $\WLvar$ can be obtained axiomatically by extending the corresponding
system $\CLvar$ with $\neg(\Box A \land \diam\neg A)$.
In this respect, the 
difference between $\CK$ and $\WK$ does not rely 
that much 
on a stronger $\diam$ of the latter,
but rather on a different interaction of $\Box$ and $\diam$ in the two systems.
In particular, 
 the modalities in $\CK$ are barely related.
At the same time,
this relation between logics $\CLvar$ and $\WLvar$ provides us with a 
framework of corresponding minimal, constructive, \wij-style and classical modal logics,
where the corresponding quadruples of systems 
are axiomatically related as follows:
\begin{center}
\begin{small}
\begin{tikzpicture}
    \node (ML) at  (0,0)  {$\MLvar$};
    \node (CL) at (3.5,0) {$\CLvar$};
    \node (WL) at (7,0) {$\WLvar$};
    \node (L) at (10.5,0) {$\logic$};

    \draw [->] (ML) -- (CL) node[midway,above] {+\begin{tabular}{c}$\bot\imp A$\\\end{tabular}};
    \draw[->] (CL) -- (WL) node[midway,above] {+\begin{tabular}{c}$\neg(\Box A \land \diam \neg A)$\\\end{tabular}};
    \draw[->] (WL) -- (L) node[midway,above] {+\begin{tabular}{c}$A \lor \neg A$ \\ $\Box A \lor \diam \neg A$ \\\end{tabular}};
\end{tikzpicture}
\end{small}
\end{center}

%
%
%
%

\textbf{Computational properties.}
In this work,
we have not considered the computational properties of the logics $\MLvar$ and $\CLvar$.
However, 
we can observe that the
equation ($\ast$) is not only a definitorial property of the logics $\MLvar$,
it is also a polynomial reduction of the derivability problem for $\MLvar$ into the derivability
problem for $\Sfour\oplus\logic$.
Considering that the derivability problems for $\Sfour\oplus\K$,  $\Sfour\oplus\KD$ and $\Sfour\oplus\KT$ are known to be \textsc{PSpace}-complete
\cite{YB},
and that $\MK$ is a conservative extension of $\MPL$ 
(with respect to the fragment of the language without the modalities)
which is also \textsc{PSpace}-complete, 
we can conclude that the derivability problems for $\MK$, $\MKD$ and $\MKT$ are \textsc{PSpace}-complete.

We conjecture that the same complexity bound applies to all logics $\MLvar$ and $\CLvar$.
In future work we would like to address this problem by studying terminating sequent calculi and construction of finite models in the style of \cite{dalmonte:2021}.
Moreover, we conjecture that 
\textsc{PSpace}-complexity can be proved for $\MM$ by combining the translation $\tr$ with the reduction of classical $\EM$ into multi-modal $\K$ presented in \cite{kracht:1999,gasquet:1996}.
We would also like to study 
 reductions for the constructive  logics $\CLvar$ along the lines of \cite{Fairtlough:1997}.


%
%

\bibliographystyle{plain}
\bibliography{mybibliography}
 
\section*{Appendix: Proofs of cut admissibility}

\ThCutMK*
\begin{proof}
The proof follows a standard strategy that goes back to Gentzen~\cite{gentzen}
(cf.~\cite{Troelstra:2000} for more details)
and consists in proving the admissibility of the following generalisation of \cut
\begin{center}
\ax{$\G \seq A$}
\ax{$\Sigma, A^n \seq C$}
\llab{\mix}
\binf{$\G, \Sigma \seq C$}
\disp
\end{center}
also known as \emph{multicut}, 
where $A^n$ denotes one or more occurrences of $A$.
The proof shows that every derivation containing one or more applications of \mix\
can be transformed into an equivalent derivation not containing applications of \mix\
by removing step by step all topmost applications of \mix.
Let us call \emph{mix formula} the formula which is deleted by the application of \mix.
The proof proceeds by induction on lexicographically ordered pairs
$(c, h)$,
where $c$ is the \emph{complexity} of the mix formula,
defined as usual as $c(p) = c(\bot) = 1$, 
$c(B \circ C) = c(B) + c(C) + 1$, 
$C(\heartsuit B) = c(B) + 1$, with
 $\circ \in \{\land, \lor, \imp\}$,
  $\heartsuit\in\{\Box,\diam\}$,
  and $h$ is the \emph{cut height},
  defined as the sum of the heights of the mix-free derivations of the premisses of \mix,
  where the height of a mix-free derivation is in turn defined as the length of the longest branch from the root to an initial sequent.
  The proof distinguishes among the following cases.
  In each case, the derivation on the left is converted into the derivation on the right,
  where the original application of \mix\ is possibly replaced by one or more applications of \mix,
  each of them having a mix formula with lower complexity or a having lower mix height.
  In the derivations, given a rule $R$, we denote $R^*$ an arbitrary number of repeated applications of $R$.

\begin{enumerate}[leftmargin=*, align=left]
\item[\framebox{1}] At least one premiss of \mix\ is an initial sequent.
There are two subcases.

\item[\framebox{1.1}] The left premiss of \mix\ is an initial sequent.
\begin{center}
\begin{small}
\ax{$A \Seq A$}
\ax{$\Der$}
\noLine
\uinf{$\G, A^n \Seq C$}
\llab{\mix}
\binf{$\G, A \Seq C$}
\disp
\quad
$\leadsto$
\quad
\ax{$\Der$}
\noLine
\uinf{$\G, A^n \Seq C$}
\rlab{\mlctr$^*$}
\uinf{$\G, A \seq C$}
\disp
\end{small}
\end{center}

\item[\framebox{1.2}] The right premiss of \mix\ is an initial sequent.
\begin{center}
\begin{small}
\ax{$\Der$}
\noLine
\uinf{$\G \Seq A$}
\ax{$A \Seq A$}
\llab{\mix}
\binf{$\G \Seq A$}
\disp
\quad
$\leadsto$
\quad
\ax{$\Der$}
\noLine
\uinf{$\G \Seq A$}
\disp
\end{small}
\end{center}

\item[\framebox{2}] Neither premiss of \mix\ is an initial sequent. 
There are  three subcases.

\item[\framebox{2.1}] The mix formula is not principal in the last rule applied in the derivation $\mathcal D$ of the left premiss of \mix.
We consider several cases depending on the last rule applied in $\mathcal D$.

\item[(\mlland)]

\begin{center}
\hfill
\begin{footnotesize}
\ax{$\Der$}
\noLine
\uinf{$\G, B_i \Seq A$}
\llab{\mlland}
\uinf{$\G, B_1 \land B_2 \Seq A$}
\ax{$\Der$}
\noLine
\uinf{$\Sigma, A^n \Seq C$}
\llab{\mix}
\binf{$\G, \Sigma, B_1 \land B_2 \Seq C$}
\disp
\hfill
$\leadsto$
\hfill
\ax{$\Der$}
\noLine
\uinf{$\G, B_i \Seq A$}
\ax{$\Der$}
\noLine
\uinf{$\Sigma, A^n \Seq C$}
\rlab{\mix}
\binf{$\G, \Sigma, B_i \seq C$}
\rlab{\mlland}
\uinf{$\G, \Sigma, B_1 \land B_2 \Seq C$}
\disp
\end{footnotesize}
\end{center}

\item[(\mllor)]

\begin{center}
\
\begin{footnotesize}
\ax{$\Der$}
\noLine
\uinf{$\G, B \Seq A$}
\ax{$\Der$}
\noLine
\uinf{$\G, C \Seq A$}
\llab{\mllor}
\binf{$\G, B \lor C \Seq A$}
\ax{$\Der$}
\noLine
\uinf{$\Sigma, A^n \Seq D$}
\llab{\mix}
\binf{$\G, \Sigma, B \lor C \Seq D$}
\disp
\ \
$\leadsto$
\hfill \ \ 

\vspace{0.2cm}
\ \ \hfill 
\ax{$\Der$}
\noLine
\uinf{$\G, B \Seq A$}
\ax{$\Der$}
\noLine
\uinf{$\Sigma, A^n \Seq D$}
\llab{\mix}
\binf{$\G, \Sigma, B \Seq D$}
\ax{$\Der$}
\noLine
\uinf{$\G, C \Seq A$}
\ax{$\Der$}
\noLine
\uinf{$\Sigma, A^n \Seq D$}
\rlab{\mix}
\binf{$\G, \Sigma, C \Seq D$}
\llab{\mllor}
\binf{$\G, \Sigma, B \lor C \Seq C$}
\disp
\end{footnotesize}
\end{center}

\item[(\mlimp)]

\begin{center}
\
\begin{footnotesize}
\ax{$\Der$}
\noLine
\uinf{$\G \Seq B$}
\ax{$\Der$}
\noLine
\uinf{$\G, C \Seq A$}
\llab{\mlimp}
\binf{$\G, B \imp C \Seq A$}
\ax{$\Der$}
\noLine
\uinf{$\Sigma, A^n \Seq D$}
\llab{\mix}
\binf{$\G, \Sigma, B \imp C \Seq D$}
\disp
\ \
$\leadsto$
\hfill \ \ 

\vspace{0.2cm}
\ \ \hfill 
\ax{$\Der$}
\noLine
\uinf{$\G \Seq B$}
\llab{\mlwk$^*$}
\uinf{$\G, \Sigma \Seq B$}
\ax{$\Der$}
\noLine
\uinf{$\G, C \Seq A$}
\ax{$\Der$}
\noLine
\uinf{$\Sigma, A^n \Seq D$}
\rlab{\mix}
\binf{$\G, \Sigma, C \Seq D$}
\llab{\mlimp}
\binf{$\G, \Sigma, B \imp C \Seq D$}
\disp
\end{footnotesize}
\end{center}

\item[(\mlwk)]

\begin{center}
\hfill
\begin{footnotesize}
\ax{$\Der$}
\noLine
\uinf{$\G \Seq A$}
\llab{\mlwk}
\uinf{$\G, B \Seq A$}
\ax{$\Der$}
\noLine
\uinf{$\Sigma, A^n \Seq C$}
\llab{\mix}
\binf{$\G, \Sigma, B \Seq C$}
\disp
\hfill
$\leadsto$
\hfill
\ax{$\Der$}
\noLine
\uinf{$\G \Seq A$}
\ax{$\Der$}
\noLine
\uinf{$\Sigma, A^n \Seq C$}
\rlab{\mix}
\binf{$\G, \Sigma \seq C$}
\rlab{\mlwk}
\uinf{$\G, \Sigma, B \Seq C$}
\disp
\end{footnotesize}
\end{center}

\item[(\mlctr)]

\begin{center}
\hfill
\begin{footnotesize}
\ax{$\Der$}
\noLine
\uinf{$\G, B, B \Seq A$}
\llab{\mlctr}
\uinf{$\G, B \Seq A$}
\ax{$\Der$}
\noLine
\uinf{$\Sigma, A^n \Seq C$}
\llab{\mix}
\binf{$\G, \Sigma, B \Seq C$}
\disp
\hfill
$\leadsto$
\hfill
\ax{$\Der$}
\noLine
\uinf{$\G, B, B \Seq A$}
\ax{$\Der$}
\noLine
\uinf{$\Sigma, A^n \Seq C$}
\rlab{\mix}
\binf{$\G, \Sigma, B, B \seq C$}
\rlab{\mlctr}
\uinf{$\G, \Sigma, B \Seq C$}
\disp
\end{footnotesize}
\end{center}

\item[Right rules and \rulemKbox, \rulemKdiam\ are not possible.]

\item[\framebox{2.2}] The mix formula is not principal in the last rule applied in the derivation $\mathcal D$ of the right premiss of \mix.
We consider several cases depending on the last rule applied in $\mathcal D$.

\item[(\mlland)]

\begin{center}
\hfill
\begin{footnotesize}
\ax{$\Der$}
\noLine
\uinf{$\Gamma \Seq A$}
\ax{$\Der$}
\noLine
\uinf{$\Sigma, A^n, B_i \Seq C$}
\rlab{\mlland}
\uinf{$\Sigma, A^n, B_1 \land B_2 \Seq C$}
\llab{\mix}
\binf{$\G, \Sigma, B_1 \land B_2 \Seq C$}
\disp
\hfill
$\leadsto$
\hfill
\ax{$\Der$}
\noLine
\uinf{$\G \Seq A$}
\ax{$\Der$}
\noLine
\uinf{$\Sigma, A^n, B_i \Seq C$}
\rlab{\mix}
\binf{$\G, \Sigma, B_i \seq C$}
\rlab{\mlland}
\uinf{$\G, \Sigma, B_1 \land B_2 \Seq C$}
\disp
\end{footnotesize}
\end{center}

\item[(\mrland)] 
\begin{center}
\begin{footnotesize}
\qquad\quad
\ax{$\Der$}
\noLine
\uinf{$\Gamma \Seq A$}
\ax{$\Der$}
\noLine
\uinf{$\Sigma, A^n \Seq B$}
\ax{$\Der$}
\noLine
\uinf{$\Sigma, A^n \Seq C$}
\rlab{\mrland}
\binf{$\Sigma, A^n \Seq B \land C$}
\llab{\mix}
\binf{$\G, \Sigma \Seq B \land C$}
\disp
\ \
$\leadsto$
\hfill \ \

\ \ \hfill
\ax{$\Der$}
\noLine
\uinf{$\G \Seq A$}
\ax{$\Der$}
\noLine
\uinf{$\Sigma, A^n \Seq B$}
\llab{\mix}
\binf{$\G, \Sigma \seq B$}
\ax{$\Der$}
\noLine
\uinf{$\G \Seq A$}
\ax{$\Der$}
\noLine
\uinf{$\Sigma, A^n \Seq C$}
\rlab{\mix}
\binf{$\G, \Sigma \seq C$}
\rlab{\mrland}
\binf{$\G, \Sigma \Seq B \land C$}
\disp
\end{footnotesize}
\end{center}

\item[(\mllor)] 
\begin{center}
\begin{footnotesize}
\qquad\quad
\ax{$\Der$}
\noLine
\uinf{$\Gamma \Seq A$}
\ax{$\Der$}
\noLine
\uinf{$\Sigma, A^n, B \Seq D$}
\ax{$\Der$}
\noLine
\uinf{$\Sigma, A^n, C \Seq D$}
\rlab{\mllor}
\binf{$\Sigma, A^n, B \lor C \Seq D$}
\llab{\mix}
\binf{$\G, \Sigma, B \lor C \Seq D$}
\disp
\ \
$\leadsto$
\hfill \ \

\ \ \hfill
\ax{$\Der$}
\noLine
\uinf{$\G \Seq A$}
\ax{$\Der$}
\noLine
\uinf{$\Sigma, A^n, B \Seq D$}
\llab{\mix}
\binf{$\G, \Sigma, B \seq D$}
\ax{$\Der$}
\noLine
\uinf{$\G \Seq A$}
\ax{$\Der$}
\noLine
\uinf{$\Sigma, A^n, C \Seq D$}
\rlab{\mix}
\binf{$\G, \Sigma, C \seq D$}
\rlab{\mllor}
\binf{$\G, \Sigma, B \lor C \Seq D$}
\disp
\end{footnotesize}
\end{center}

\item[(\mrlor)]

\begin{center}
\hfill
\begin{footnotesize}
\ax{$\Der$}
\noLine
\uinf{$\Gamma \Seq A$}
\ax{$\Der$}
\noLine
\uinf{$\Sigma, A^n \Seq B_i$}
\rlab{\mrlor}
\uinf{$\Sigma, A^n \Seq B_1 \lor B_2 $}
\llab{\mix}
\binf{$\G, \Sigma \Seq B_1 \lor B_2$}
\disp
\hfill
$\leadsto$
\hfill
\ax{$\Der$}
\noLine
\uinf{$\G \Seq A$}
\ax{$\Der$}
\noLine
\uinf{$\Sigma, A^n \Seq B_i$}
\rlab{\mix}
\binf{$\G, \Sigma \seq B_i$}
\rlab{\mrlor}
\uinf{$\G, \Sigma \Seq B_1 \lor B_2$}
\disp
\end{footnotesize}
\end{center}

\item[(\mlimp)]

\begin{center}
\ 
\begin{footnotesize}
\ax{$\Der$}
\noLine
\uinf{$\Gamma \Seq A$}
\ax{$\Der$}
\noLine
\uinf{$\Sigma, A^n \Seq B$}
\ax{$\Der$}
\noLine
\uinf{$\Sigma, A^n, C \Seq D$}
\rlab{\mlimp}
\binf{$\Sigma, A^n, B \imp C \Seq D$}
\llab{\mix}
\binf{$\G, \Sigma, B \imp C \Seq D$}
\disp
\
$\leadsto$ 
\hfill \ \ 

\ \ \hfill
\ax{$\Der$}
\noLine
\uinf{$\G \Seq A$}
\ax{$\Der$}
\noLine
\uinf{$\Sigma, A^n \Seq B$}
\llab{\mix}
\binf{$\G, \Sigma \seq B$}
\ax{$\Der$}
\noLine
\uinf{$\G \Seq A$}
\ax{$\Der$}
\noLine
\uinf{$\Sigma, A^n, C \Seq D$}
\rlab{\mix}
\binf{$\G, \Sigma, C \seq D$}
\rlab{\mlimp}
\binf{$\G, \Sigma, B \imp C \Seq D$}
\disp
\end{footnotesize}
\end{center}

\item[(\mrimp)]

\begin{center}
\hfill
\begin{footnotesize}
\ax{$\Der$}
\noLine
\uinf{$\Gamma \Seq A$}
\ax{$\Der$}
\noLine
\uinf{$\Sigma, A^n, B \Seq C$}
\rlab{\mrimp}
\uinf{$\Sigma, A^n \Seq B \imp C $}
\llab{\mix}
\binf{$\G, \Sigma \Seq B \imp C$}
\disp
\hfill
$\leadsto$
\hfill
\ax{$\Der$}
\noLine
\uinf{$\G \Seq A$}
\ax{$\Der$}
\noLine
\uinf{$\Sigma, A^n, B \Seq C$}
\rlab{\mix}
\binf{$\G, \Sigma, B \seq C$}
\rlab{\mrimp}
\uinf{$\G, \Sigma \Seq B \imp C$}
\disp
\end{footnotesize}
\end{center}

\item[(\mlwk)]

\begin{center}
\hfill
\begin{footnotesize}
\ax{$\Der$}
\noLine
\uinf{$\Gamma \Seq A$}
\ax{$\Der$}
\noLine
\uinf{$\Sigma, A^n \Seq C$}
\rlab{\mlwk}
\uinf{$\Sigma, A^n, B \Seq C$}
\llab{\mix}
\binf{$\G, \Sigma, B \Seq C$}
\disp
\hfill
$\leadsto$
\hfill
\ax{$\Der$}
\noLine
\uinf{$\G \Seq A$}
\ax{$\Der$}
\noLine
\uinf{$\Sigma, A^n \Seq C$}
\rlab{\mix}
\binf{$\G, \Sigma \seq C$}
\rlab{\mlwk}
\uinf{$\G, \Sigma, B \Seq C$}
\disp
\end{footnotesize}
\end{center}

\item[(\mlctr)]

\begin{center}
\hfill
\begin{footnotesize}
\ax{$\Der$}
\noLine
\uinf{$\Gamma \Seq A$}
\ax{$\Der$}
\noLine
\uinf{$\Sigma, A^n, B, B \Seq C$}
\rlab{\mlctr}
\uinf{$\Sigma, A^n, B \Seq C$}
\llab{\mix}
\binf{$\G, \Sigma, B \Seq C$}
\disp
\hfill
$\leadsto$
\hfill
\ax{$\Der$}
\noLine
\uinf{$\G \Seq A$}
\ax{$\Der$}
\noLine
\uinf{$\Sigma, A^n, B, B \Seq C$}
\rlab{\mix}
\binf{$\G, \Sigma, B, B \seq C$}
\rlab{\mlctr}
\uinf{$\G, \Sigma, B \Seq C$}
\disp
\end{footnotesize}
\end{center}

\item[\rulemKbox, \rulemKdiam\ are not possible.]

\item[\framebox{2.3}] The mix formula is principal in the last rule applied in the derivations $\mathcal D_1$, $\mathcal D_2$ of both premisses of \mix.
We consider several cases depending on the last rule applied in $\mathcal D_1$, $\mathcal D_2$.

%

\item[(\mrland\ - \mlland)] The mix formula $A$ has the form $B \land C$.
We consider the following case, the other case where the premiss of \mlland\ is 
$\Sigma, (B \land C)^{n-1}, C \seq D$
 is analogous.


\begin{footnotesize}
\begin{center}
\ax{$\Der$}
\noLine
\uinf{$\G \Seq B$}
\ax{$\Der$}
\noLine
\uinf{$\G \Seq C$}
\llab{\mrland}
\binf{$\G \Seq B \land C$}
\ax{$\Der$}
\noLine
\uinf{$\Sigma, (B \land C)^{n-1}, B \Seq D$}
\rlab{\mlland}
\uinf{$\Sigma, (B \land C)^n \seq D$}
\llab{\mix}
\binf{$\G, \Sigma \Seq D$}
\disp
\
$\leadsto$ 
\hfill \ \ 

\ \ \hfill
\ax{$\Der$}
\noLine
\uinf{$\G \Seq B$}
\ax{$\Der$}
\noLine
\uinf{$\G \Seq B$}
\ax{$\Der$}
\noLine
\uinf{$\G \Seq C$}
\llab{\mrland}
\binf{$\G \Seq B \land C$}
\ax{$\Der$}
\noLine
\uinf{$\Sigma, (B \land C)^{n-1}, B \Seq D$}
\rlab{\mix}
\binf{$\G, \Sigma, B \Seq D$}
\rlab{\mix}
\binf{$\G, \G, \Sigma \Seq D$}
\rlab{\mlwk$^*$}
\uinf{$\G, \Sigma \Seq D$}
\disp
\end{center}
\end{footnotesize}

\item[(\mrlor\ - \mllor)] The mix formula $A$ has the form $B \lor C$.
We consider the following case, the other case where the premiss of \mrlor\ is $\G \seq C$ is analogous.


\begin{footnotesize}
\begin{center}
\quad
\ax{$\Der$}
\noLine
\uinf{$\G \Seq B$}
\llab{\mrlor}
\uinf{$\G \Seq B \lor C$}
\ax{$\Der$}
\noLine
\uinf{$\Sigma, (B \lor C)^{n-1}, B \Seq D$}
\ax{$\Der$}
\noLine
\uinf{$\Sigma, (B \lor C)^{n-1}, C \Seq D$}
\rlab{\mllor}
\binf{$\Sigma, (B \lor C)^n \seq D$}
\llab{\mix}
\binf{$\G, \Sigma \Seq D$}
\disp
\
$\leadsto$ 
\hfill \ \ 

\ \ \hfill
\ax{$\Der$}
\noLine
\uinf{$\G \Seq B$}
\ax{$\Der$}
\noLine
\uinf{$\G \Seq B$}
\llab{\mrlor}
\uinf{$\G \Seq B \lor C$}
\ax{$\Der$}
\noLine
\uinf{$\Sigma, (B \lor C)^{n-1}, B \Seq D$}
\rlab{\mix}
\binf{$\G, \Sigma, B \Seq D$}
\rlab{\mix}
\binf{$\G, \G, \Sigma \Seq D$}
\rlab{\mlwk$^n$}
\uinf{$\G, \Sigma \Seq D$}
\disp
\end{center}
\end{footnotesize}

\item[(\mrimp\ - \mlimp)] The mix formula $A$ has the form $B \imp C$.


\begin{footnotesize}
\begin{center}
\ax{$\Der$}
\noLine
\uinf{$\G, B \Seq C$}
\llab{\mrimp}
\uinf{$\G \Seq B \imp C$}
\ax{$\Der$}
\noLine
\uinf{$\Sigma, (B \imp C)^{n-1} \Seq B$}
\ax{$\Der$}
\noLine
\uinf{$\Sigma, (B \imp C)^{n-1}, C \Seq D$}
\rlab{\mlimp}
\binf{$\Sigma, (B \imp C)^n \seq D$}
\llab{\mix}
\binf{$\G, \Sigma \Seq D$}
\disp
\
$\leadsto$ 
\hfill \ \ 

\ \ \hfill
\ax{$\Der$}
\noLine
\uinf{$\G, B \Seq C$}
\llab{\mrimp}
\uinf{$\G \Seq B \imp C$}
\ax{$\Der$}
\noLine
\uinf{$\Sigma, (B \imp C)^{n-1} \Seq B$}
\llab{\mix}
\binf{$\G, \Sigma \Seq B$}
\ax{$\Der$}
\noLine
\uinf{$\G, B \Seq C$}
\rlab{\mix}
\binf{$\G, \G, \Sigma \seq C$ $(\ast)$}
\disp

\ \ \hfill
\ax{$\G, \G, \Sigma \seq C$ $(\ast)$}
\ax{$\Der$}
\noLine
\uinf{$\G, B \Seq C$}
\llab{\mrimp}
\uinf{$\G \Seq B \imp C$}
\ax{$\Der$}
\noLine
\uinf{$\Sigma, (B \imp C)^{n-1}, C \Seq D$}
\rlab{\mix}
\binf{$\G, \Sigma, C \seq D$}
\rlab{\mix}
\binf{$\G, \G, \G, \Sigma \seq D$}
\rlab{\mlctr$^*$}
\uinf{$\G, \Sigma \seq D$}
\disp
\end{center}
\end{footnotesize}

\item[($R$ - \mlctr)] The transformation below 
applies for any last rule $R$ in the derivation of the left premiss of \mix.

\begin{footnotesize}
\begin{center}
\ax{$\Der$}
\noLine
\uinf{$\G \Seq A$}
\ax{$\Der$}
\noLine
\uinf{$\Sigma, A^n, A \Seq C$}
\rlab{\mlctr}
\uinf{$\Sigma, A^n \Seq C$}
\llab{\mix}
\binf{$\G, \Sigma \seq C$}
\disp
\ \
$\leadsto$
\ \
\ax{$\Der$}
\noLine
\uinf{$\G \Seq A$}
\ax{$\Der$}
\noLine
\uinf{$\Sigma, A^n, A \Seq C$}
\rlab{\mix}
\binf{$\G, \Sigma \seq C$}
\disp
\end{center}
\end{footnotesize}

\item[($R$ - \mlwk)] The transformation below applies for any last rule $R$ in the derivation of the left premiss of \mix\ (note that if $n=1$, then the conclusion of \mix\ $\G, \Sigma \Seq C$ can be obtained from $\Sigma \Seq C$ by \mlwk).

\begin{footnotesize}
\begin{center}
\ax{$\Der$}
\noLine
\uinf{$\G \Seq A$}
\ax{$\Der$}
\noLine
\uinf{$\Sigma, A^{n-1} \Seq C$}
\rlab{\mlwk}
\uinf{$\Sigma, A^{n} \Seq C$}
\llab{\mix}
\binf{$\G, \Sigma \seq C$}
\disp
\ \
$\leadsto$
\ \
\ax{$\Der$}
\noLine
\uinf{$\G \Seq A$}
\ax{$\Der$}
\noLine
\uinf{$\Sigma, A^{n-1} \Seq C$}
\rlab{\mix}
\binf{$\G, \Sigma \seq C$}
\disp
\end{center}
\end{footnotesize}

\item[(\rulemKbox\ - \rulemKbox)] The mix formula $A$ has the form $\Box B$.

\begin{footnotesize}
\begin{center}
\ax{$\Der$}
\noLine
\uinf{$\Sigma \Seq B$}
\llab{\rulemKbox}
\uinf{$\Box\Sigma \Seq \Box B$}
\ax{$\Der$}
\noLine
\uinf{$B^n, \Pi \Seq C$}
\rlab{\rulemKbox}
\uinf{$(\Box B)^n, \Box\Pi \Seq \Box C$}
\llab{\mix}
\binf{$\Box\Sigma, \Box\Pi \Seq \Box C$}
\disp
\ \ 
$\leadsto$
\ \ 
\ax{$\Der$}
\noLine
\uinf{$\Sigma \Seq B$}
\ax{$\Der$}
\noLine
\uinf{$B^n, \Pi \Seq C$}
\rlab{\mix}
\binf{$\Sigma, \Pi \Seq C$}
\rlab{\rulemKbox}
\uinf{$\Box\Sigma, \Box\Pi \Seq \Box C$}
\disp
\end{center}
\end{footnotesize}

\item[(\rulemKbox\ - \rulemKdiam)]
The mix formula $A$ has the form $\Box B$.

\begin{footnotesize}
\begin{center}
\ax{$\Der$}
\noLine
\uinf{$\Sigma \Seq B$}
\llab{\rulemKbox}
\uinf{$\Box\Sigma \Seq \Box B$}
\ax{$\Der$}
\noLine
\uinf{$B^n, \Pi, C \Seq D$}
\rlab{\rulemKdiam}
\uinf{$(\Box B)^n, \Box\Pi, \diam C \Seq \diam D$}
\llab{\mix}
\binf{$\Box\Sigma, \Box\Pi, \diam C \Seq \diam D$}
\disp
\hfill
$\leadsto$ 
\hfill
\ax{$\Der$}
\noLine
\uinf{$\Sigma \Seq B$}
\ax{$\Der$}
\noLine
\uinf{$B^n, \Pi, C \Seq D$}
\rlab{\mix}
\binf{$\Sigma, \Pi, C \Seq D$}
\rlab{\rulemKdiam}
\uinf{$\Box\Sigma, \Box\Pi, \diam C \Seq \diam D$}
\disp
\end{center}
\end{footnotesize}

\item[(\rulemKdiam\ - \rulemKdiam)] The mix formula $A$ has the form $\diam C$.

\begin{footnotesize}
\begin{center}
\ax{$\Der$}
\noLine
\uinf{$\Sigma, B \Seq C$}
\llab{\rulemKdiam}
\uinf{$\Box\Sigma, \diam B \Seq \diam C$}
\ax{$\Der$}
\noLine
\uinf{$\Pi, C \Seq D$}
\rlab{\rulemKdiam}
\uinf{$\Box\Pi, \diam C \Seq \diam D$}
\llab{\mix}
\binf{$\Box\Sigma, \Box\Pi, \diam B \Seq \diam D$}
\disp
\ \
$\leadsto$
\ \
\ax{$\Der$}
\noLine
\uinf{$\Sigma, B \Seq C$}
\ax{$\Der$}
\noLine
\uinf{$\Pi, C \Seq D$}
\rlab{\mix}
\binf{$\Sigma, \Pi, B \Seq D$}
\rlab{\rulemKdiam}
\uinf{$\Box\Sigma, \Box\Pi, \diam B \Seq \diam D$}
\disp
\end{center}
\end{footnotesize}
\end{enumerate}
\end{proof}

\ThCutML*
\begin{proof}
We extend the cases in the proof of Theorem~\ref{th:cut MK} with the analysis of the new modal rules.
The cases 1.1 and 1.2 are as before.
\begin{enumerate}[leftmargin=*, align=left]
\item[\framebox{2.1}] The mix formula is not principal in the last rule applied in the derivation $\mathcal D$ of the left premiss of \mix.

\item[(\rulemTbox)]

\begin{center}
\hfill
\begin{footnotesize}
\ax{$\Der$}
\noLine
\uinf{$\G, B \Seq A$}
\llab{\rulemTbox}
\uinf{$\G, \Box B \Seq A$}
\ax{$\Der$}
\noLine
\uinf{$\Sigma, A^n \Seq C$}
\llab{\mix}
\binf{$\G, \Sigma, \Box B \Seq C$}
\disp
\hfill
$\leadsto$
\hfill
\ax{$\Der$}
\noLine
\uinf{$\G, B \Seq A$}
\ax{$\Der$}
\noLine
\uinf{$\Sigma, A^n \Seq C$}
\rlab{\mix}
\binf{$\G, \Sigma, B \seq C$}
\rlab{\rulemTbox}
\uinf{$\G, \Sigma, \Box B \Seq C$}
\disp
\end{footnotesize}
\end{center}

\item[\framebox{2.2}] The mix formula is not principal in the last rule applied in the derivation $\mathcal D$ of the right premiss of \mix.

\item[(\rulemTbox)]

\begin{center}
\hfill
\begin{footnotesize}
\ax{$\Der$}
\noLine
\uinf{$\Gamma \Seq A$}
\ax{$\Der$}
\noLine
\uinf{$\Sigma, A^n, B \Seq C$}
\rlab{\rulemTbox}
\uinf{$\Sigma, A^n, \Box B \Seq C$}
\llab{\mix}
\binf{$\G, \Sigma, \Box B \Seq C$}
\disp
\hfill
$\leadsto$
\hfill
\ax{$\Der$}
\noLine
\uinf{$\G \Seq A$}
\ax{$\Der$}
\noLine
\uinf{$\Sigma, A^n, B \Seq C$}
\rlab{\mix}
\binf{$\G, \Sigma, B \seq C$}
\rlab{\rulemTbox}
\uinf{$\G, \Sigma, \Box B \Seq C$}
\disp
\end{footnotesize}
\end{center}

\item[(\rulemTdiam)]

\begin{center}
\hfill
\begin{footnotesize}
\ax{$\Der$}
\noLine
\uinf{$\Gamma \Seq A$}
\ax{$\Der$}
\noLine
\uinf{$\Sigma, A^n \Seq B$}
\rlab{\rulemTdiam}
\uinf{$\Sigma, A^n \Seq \diam B$}
\llab{\mix}
\binf{$\G, \Sigma \Seq \diam B$}
\disp
\hfill
$\leadsto$
\hfill
\ax{$\Der$}
\noLine
\uinf{$\G \Seq A$}
\ax{$\Der$}
\noLine
\uinf{$\Sigma, A^n \Seq B$}
\rlab{\mix}
\binf{$\G, \Sigma \seq B$}
\rlab{\rulemTdiam}
\uinf{$\G, \Sigma \Seq \diam B$}
\disp
\end{footnotesize}
\end{center}

\item[\framebox{2.3}]  The mix formula is principal in the last rule applied in the derivations $\mathcal D_1$, $\mathcal D_2$ of both premisses of \mix.
For the cases where the last rule applied in $\mathcal D_1$, $\mathcal D_2$
is propositional see the proof of Theorem~\ref{th:cut MK}.
We show the other cases.

\item[(\rulemNbox\ - \rulemMbox)] The mix formula $A$ has the form $\Box B$.

\begin{footnotesize}
\begin{center}
\ax{$\Der$}
\noLine
\uinf{$\Seq B$}
\llab{\rulemNbox}
\uinf{$\Seq \Box B$}
\ax{$\Der$}
\noLine
\uinf{$B \seq C$}
\rlab{\rulemMbox}
\uinf{$\Box B \Seq \Box C$}
\llab{\mix}
\binf{$\Seq \Box C$}
\disp
\
$\leadsto$ 
\
\ax{$\Der$}
\noLine
\uinf{$\Seq B$}
\ax{$\Der$}
\noLine
\uinf{$B \Seq C$}
\rlab{\mix}
\binf{$\Seq C$}
\rlab{\rulemNbox}
\uinf{$\Seq \Box C$}
\disp
\end{center}
\end{footnotesize}

\item[(\rulemPdiam\ - \rulemMdiam)] The mix formula $A$ has the form $\diam B$.

\begin{footnotesize}
\begin{center}
\ax{$\Der$}
\noLine
\uinf{$\Seq B$}
\llab{\rulemPdiam}
\uinf{$\Seq \diam B$}
\ax{$\Der$}
\noLine
\uinf{$B \seq C$}
\rlab{\rulemMdiam}
\uinf{$\diam B \Seq \diam C$}
\llab{\mix}
\binf{$\Seq \diam C$}
\disp
\
$\leadsto$ 
\
\ax{$\Der$}
\noLine
\uinf{$\Seq B$}
\ax{$\Der$}
\noLine
\uinf{$B \Seq C$}
\rlab{\mix}
\binf{$\Seq C$}
\rlab{\rulemPdiam}
\uinf{$\Seq \diam C$}
\disp
\end{center}
\end{footnotesize}

\item[(\rulemNbox\ - \rulemD)] The mix formula $A$ has the form $\Box B$
(note that by definition \rulemPdiam\ belongs to the calculus).

\begin{footnotesize}
\begin{center}
\ax{$\Der$}
\noLine
\uinf{$\Seq B$}
\llab{\rulemNbox}
\uinf{$\Seq \Box B$}
\ax{$\Der$}
\noLine
\uinf{$B \seq C$}
\rlab{\rulemD}
\uinf{$\Box B \Seq \diam C$}
\llab{\mix}
\binf{$\Seq \diam C$}
\disp
\
$\leadsto$ 
\
\ax{$\Der$}
\noLine
\uinf{$\Seq B$}
\ax{$\Der$}
\noLine
\uinf{$B \Seq C$}
\rlab{\mix}
\binf{$\Seq C$}
\rlab{\rulemPdiam}
\uinf{$\Seq \diam C$}
\disp
\end{center}
\end{footnotesize}

\item[(\rulemNbox\ - \rulemTbox)] The mix formula $A$ has the form $\Box B$.

\begin{footnotesize}
\begin{center}
\ax{$\Der$}
\noLine
\uinf{$\Seq B$}
\llab{\rulemNbox}
\uinf{$\Seq \Box B$}
\ax{$\Der$}
\noLine
\uinf{$\G, (\Box B)^{n-1}, B \Seq C$}
\rlab{\rulemTbox}
\uinf{$\G, (\Box B)^n \Seq C$}
\llab{\mix}
\binf{$\G \Seq C$}
\disp
\
$\leadsto$ 
\hfill \ \ 

\ \ \hfill
\ax{$\Der$}
\noLine
\uinf{$\Seq B$}
\ax{$\Der$}
\noLine
\uinf{$\Seq B$}
\llab{\rulemNbox}
\uinf{$\Seq \Box B$}
\ax{$\Der$}
\noLine
\uinf{$\G, (\Box B)^{n-1}, B \Seq C$}
\rlab{\mix}
\binf{$\G, B \Seq C$}
\rlab{\mix}
\binf{$\G \Seq C$}
\disp
\end{center}
\end{footnotesize}

\item[(\rulemCbox\ - \rulemCbox)] The mix formula $A$ has the form $\Box C$.

\begin{footnotesize}
\begin{center}
\ax{$\Der$}
\noLine
\uinf{$\Sigma, B \Seq C$}
\llab{\rulemCbox}
\uinf{$\Box\Sigma, \Box B \Seq \Box C$}
\ax{$\Der$}
\noLine
\uinf{$C^n, \Pi \Seq D$}
\rlab{\rulemCbox}
\uinf{$(\Box C)^n, \Box\Pi \Seq \Box D$}
\llab{\mix}
\binf{$\Box\Sigma, \Box\Pi, \Box B \Seq \Box D$}
\disp
\
$\leadsto$ 
\
\ax{$\Der$}
\noLine
\uinf{$\Sigma, B \Seq C$}
\ax{$\Der$}
\noLine
\uinf{$C^n, \Pi \Seq D$}
\rlab{\mix}
\binf{$\Sigma, \Pi, B \Seq D$}
\rlab{\rulemCbox}
\uinf{$\Box\Sigma, \Box\Pi, \Box B \Seq \Box D$}
\disp
\end{center}
\end{footnotesize}

\item[(\rulemCbox\ - \rulemKdiam)] The mix formula $A$ has the form $\Box C$.

\begin{footnotesize}
\begin{center}
\ax{$\Der$}
\noLine
\uinf{$\Sigma, B \Seq C$}
\llab{\rulemCbox}
\uinf{$\Box\Sigma, \Box B \Seq \Box C$}
\ax{$\Der$}
\noLine
\uinf{$C^n, \Pi, D \Seq E$}
\rlab{\rulemKdiam}
\uinf{$(\Box C)^n, \Box\Pi, \diam D \Seq \diam E$}
\llab{\mix}
\binf{$\Box\Sigma, \Box\Pi, \Box B, \diam D \Seq \diam E$}
\disp
\
$\leadsto$ 
\hfill \ \ 

\ \ \hfill
\ax{$\Der$}
\noLine
\uinf{$\Sigma, B \Seq C$}
\ax{$\Der$}
\noLine
\uinf{$C^n, \Pi, D \Seq E$}
\rlab{\mix}
\binf{$\Sigma, \Pi, B, D \Seq E$}
\rlab{\rulemKdiam}
\uinf{$\Box\Sigma, \Box\Pi, \Box B, \diam D \Seq \diam E$}
\disp
\end{center}
\end{footnotesize}


\item[(\rulemCbox\ - \rulemCD)] The mix formula $A$ has the form $\Box C$.

\begin{footnotesize}
\begin{center}
\ax{$\Der$}
\noLine
\uinf{$\Sigma, B \Seq C$}
\llab{\rulemCbox}
\uinf{$\Box\Sigma, \Box B \Seq \Box C$}
\ax{$\Der$}
\noLine
\uinf{$C^n, \Pi \Seq D$}
\rlab{\rulemCD}
\uinf{$(\Box C)^n, \Box\Pi \Seq \diam D$}
\llab{\mix}
\binf{$\Box\Sigma, \Box B, \Box\Pi \Seq \diam D$}
\disp
\hfill
$\leadsto$ 
\hfill
\ax{$\Der$}
\noLine
\uinf{$\Sigma, B \Seq C$}
\ax{$\Der$}
\noLine
\uinf{$C^n, \Pi \Seq D$}
\rlab{\mix}
\binf{$\Sigma, B, \Pi \Seq D$}
\rlab{\rulemCD}
\uinf{$\Box\Sigma, \Box B, \Box\Pi \Seq \diam D$}
\disp
\end{center}
\end{footnotesize}

\item[(\rulemCD\ - \rulemKdiam)] The mix formula $A$ has the form $\diam B$.

\begin{footnotesize}
\begin{center}
\ax{$\Der$}
\noLine
\uinf{$\Sigma \Seq B$}
\llab{\rulemCD}
\uinf{$\Box\Sigma \Seq \diam B$}
\ax{$\Der$}
\noLine
\uinf{$\Pi, B \Seq C$}
\rlab{\rulemKdiam}
\uinf{$\Box\Pi, \diam B \Seq \diam C$}
\llab{\mix}
\binf{$\Box\Sigma, \Box\Pi \Seq \diam C$}
\disp
\
$\leadsto$ 
\
\ax{$\Der$}
\noLine
\uinf{$\Sigma \Seq B$}
\ax{$\Der$}
\noLine
\uinf{$\Pi, B \Seq C$}
\rlab{\mix}
\binf{$\Sigma, \Pi \Seq C$}
\rlab{\rulemCD}
\uinf{$\Box\Sigma, \Box\Pi \Seq \diam C$}
\disp
\end{center}
\end{footnotesize}

\item[(\rulemCbox\ - \rulemTbox)] The mix formula $A$ has the form $\Box C$.

\begin{footnotesize}
\begin{center}
\ax{$\Der$}
\noLine
\uinf{$\Sigma, B \Seq C$}
\llab{\rulemCbox}
\uinf{$\Box\Sigma, \Box B \Seq \Box C$}
\ax{$\Der$}
\noLine
\uinf{$\G, (\Box C)^{n-1}, C \Seq D$}
\rlab{\rulemTbox}
\uinf{$\G, (\Box C)^n \Seq D$}
\llab{\mix}
\binf{$\G, \Box\Sigma, \Box B \Seq D$}
\disp
\
$\leadsto$ 
\hfill \ \ 

\ \ \hfill
\ax{$\Der$}
\noLine
\uinf{$\Sigma, B \Seq C$}
\ax{$\Der$}
\noLine
\uinf{$\Sigma, B \Seq C$}
\llab{\rulemCbox}
\uinf{$\Box\Sigma, \Box B \Seq \Box C$}
\ax{$\Der$}
\noLine
\uinf{$\G, (\Box C)^{n-1}, C \Seq D$}
\rlab{\mix}
\binf{$\G, \Box\Sigma, \Box B, C \Seq D$}
\rlab{\mix}
\binf{$\G, \Box\Sigma, \Box B, \Sigma, B \Seq D$}
\rlab{\rulemTbox$^*$}
\uinf{$\G, \Box\Sigma, \Box B, \Box\Sigma, \Box B \Seq D$}
\rlab{\mlctr$^*$}
\uinf{$\G, \Box\Sigma, \Box B \Seq D$}
\disp
\end{center}
\end{footnotesize}

\item[(\rulemTdiam\ - \rulemKdiam)] The mix formula $A$ has the form $\diam B$.

\begin{footnotesize}
\begin{center}
\ax{$\Der$}
\noLine
\uinf{$\G \Seq B$}
\llab{\rulemTdiam}
\uinf{$\G \Seq \diam B$}
\ax{$\Der$}
\noLine
\uinf{$\Sigma, B \Seq C$}
\rlab{\rulemKdiam}
\uinf{$\Box\Sigma, \diam B \Seq \diam C$}
\llab{\mix}
\binf{$\G, \Box\Sigma \Seq \diam C$}
\disp
\
$\leadsto$ 
\
\ax{$\Der$}
\noLine
\uinf{$\G \Seq B$}
\ax{$\Der$}
\noLine
\uinf{$\Sigma, B \Seq C$}
\rlab{\mix}
\binf{$\G, \Sigma \Seq C$}
\rlab{\rulemTbox$^*$}
\uinf{$\G, \Box\Sigma \Seq C$}
\rlab{\rulemTdiam}
\uinf{$\G, \Box\Sigma \Seq \diam C$}
\disp
\end{center}
\end{footnotesize}
\end{enumerate}

\noindent
For the remaining combinations,
(\rulemMbox\ - \rulemMbox) is analogous to (\rulemCbox\ - \rulemCbox) with $|\Sigma|=|\Pi|=0$ and $n = 1$;
(\rulemMdiam\ - \rulemMdiam) is analogous to (\rulemKdiam\ - \rulemKdiam) with $|\Sigma|=|\Pi|=0$;
(\rulemMbox\ - \rulemD) is analogous to (\rulemCbox\ - \rulemCD) with $|\Sigma|=|\Pi|=0$ and $n = 1$;
(\rulemD\ - \rulemMdiam) is analogous to (\rulemCD\ - \rulemKdiam) with $|\Sigma|=|\Pi|=0$;
(\rulemMbox\ - \rulemTbox) is analogous to (\rulemCbox\ - \rulemTbox) with $|\Sigma|=0$;
and
(\rulemTdiam\ - \rulemMdiam) is analogous to (\rulemTdiam\ - \rulemKdiam) with $|\Sigma|=0$.
\end{proof}

\ThCutCL*
\begin{proof}
We extend the cases in the proof of Theorem~\ref{th:cut ML} with the 
combinations involving \ilbot\ and \irwk.
The cases 1.1 and 2.1 are as in the proof of Theorem~\ref{th:cut ML}.

\begin{enumerate}[leftmargin=*, align=left]
\item[\framebox{1.2}] The right premiss of \mix\ is the initial sequent \ilbot:
\begin{center}
\begin{small}
\ax{$\Der$}
\noLine
\uinf{$\G \Seq \bot$}
\ax{$\bot \Seq$}
\llab{\mix}
\binf{$\G \Seq$}
\disp
\end{small}
\end{center}
We need to consider the last rule applied in the derivation of the left premiss of \mix\ $\G \seq \bot$,
which is a left propositional rule or \ruleiTbox. We show as an example the latter possibility.
\begin{center}
\begin{small}
\ax{$\Der$}
\noLine
\uinf{$\G, B \Seq \bot$}
\llab{\ruleiTbox}
\uinf{$\G, \Box B \Seq \bot$}
\ax{$\bot \Seq$}
\llab{\mix}
\binf{$\G, \Box B \Seq$}
\disp
\ \
$\leadsto$
\ \
\ax{$\Der$}
\noLine
\uinf{$\G, B \Seq \bot$}
\ax{$\bot \Seq$}
\llab{\mix}
\binf{$\G, B \Seq$}
\llab{\ruleiTbox}
\uinf{$\G, \Box B \Seq$}
\disp
\end{small}
\end{center}

\item[\framebox{2.2}] The mix formula is not principal in the last rule applied in the derivation $\mathcal D$ of the right premiss of \mix.

\item[(\irwk)]

\begin{center}
\hfill
\begin{footnotesize}
\ax{$\Der$}
\noLine
\uinf{$\Gamma \Seq A$}
\ax{$\Der$}
\noLine
\uinf{$\Sigma, A^n \seq$}
\rlab{\irwk}
\uinf{$\Sigma, A^n \seq B$}
\llab{\mix}
\binf{$\G, \Sigma \Seq B$}
\disp
\ \
$\leadsto$
\ \
\ax{$\Der$}
\noLine
\uinf{$\G \Seq A$}
\ax{$\Der$}
\noLine
\uinf{$\Sigma, A^n \Seq$}
\rlab{\mix}
\binf{$\G, \Sigma \seq$}
\rlab{\irwk}
\uinf{$\G, \Sigma \Seq B$}
\disp
\end{footnotesize}
\end{center}

\item[\framebox{2.3}]  The mix formula is principal in the last rule applied in the derivations $\mathcal D_1$, $\mathcal D_2$ of both premisses of \mix.

\item[(\irwk\ - $R$)] The transformation below 
applies for any last rule $R$ in the derivation of the left premiss of \mix.

\begin{footnotesize}
\begin{center}
\ax{$\Der$}
\noLine
\uinf{$\Gamma \Seq$}
\llab{\irwk}
\uinf{$\G \seq A$}
\ax{$\Der$}
\noLine
\uinf{$\Sigma, A^n \seq \delta$}
\llab{\mix}
\binf{$\G, \Sigma \Seq \delta$}
\disp
\ \
$\leadsto$
\ \
\ax{$\Der$}
\noLine
\uinf{$\G \seq$}
\rlab{\ilwk$^*$ (if $|\Sigma| \geq 0$)}
\uinf{$\G, \Sigma \seq$}
\rlab{\irwk$^*$ (if $|\delta| = 1$)}
\uinf{$\G, \Sigma \seq \delta$}
\disp
\end{center}
\end{footnotesize}

\end{enumerate}

\end{proof}

\end{document}